\documentclass[a4paper,11pt]{article}

\usepackage{amsthm}
\usepackage{fullpage}%
\usepackage{graphicx,tikz,enumerate}
\usepackage[colorlinks=true]{hyperref}
\usepackage[title]{appendix}
\usepackage{amssymb, amsfonts,mathtools,stmaryrd}
\usepackage{xcolor}
\usepackage{enumitem}
\usepackage{etoolbox}

\usepackage{algorithm}
\usepackage{algpseudocode}

\usepackage{dsfont}
\usepackage{bbm}
\newtheorem{definition}{\textbf{Definition}}

\newtheorem{theorem}[definition]{\textbf{Theorem}}

\newtheorem{lemma}[definition]{\textbf{Lemma}}

\newtheorem{corollary}[definition]{\textbf{Corollary}}

\newtheorem{remark}[definition]{\textbf{Remark}}

\newtheorem{proposition}[definition]{\textbf{Proposition}}

\newtheorem{example}[definition]{\textbf{Example}}

\newtheorem{notation}[definition]{\textbf{Notation}}

\newtheorem*{proof*}{\textbf{Proof}}

\newcommand{\Col}{\ensuremath{\msf{Col}}}
\newcommand{\ColB}{\ensuremath{\msf{ColBSCC}}}

\newcommand{\msf}[1]{\mathsf{#1}}

\newcommand{\Zone}{\mathcal{L}}

\newcommand{\ScPar}{\ensuremath{\msf{Succ}}}

\newcommand{\brck}[1]{\ensuremath{\llbracket #1 \rrbracket}}
\newcommand{\brd}{\ensuremath{\brck{0,e}}}

\newcommand{\GF}{\mathsf{GF}}

\newcommand{\Frm}{\mathsf{Form}}
\newcommand{\formNF}{\mathcal{F}}

\newcommand{\gameNF}[2]{\langle #1,#2 \rangle}

\newcommand{\outComeNF}{\mathsf{O}}
\newcommand{\outCNF}{\varrho}

\newcommand{\newC}{\mathsf{NewCol}}

\newcommand{\qinit}{\ensuremath{q_{\msf{init}}}}

\newcommand{\Prb}{\ensuremath{\mathbb{P}}}

\newcommand{\A}{\mathsf{A} }
\newcommand{\B}{\mathsf{B} }
\newcommand{\C}{\mathsf{C} }

\newcommand{\colSet}{\mathsf{K} }

\newcommand{\s}{\ensuremath{\mathsf{s}}}

\newcommand{\Act}{\ensuremath{\mathsf{Act}}}

\newcommand{\colFunc}{\mathsf{col}}
\newcommand{\colVirt}{\mathsf{vcol}}
\newcommand{\Dist}{\mathcal{D} }

\newcommand{\Supp}{\mathsf{Sp} }
\newcommand{\distribSet}{\mathsf{D} }
\newcommand{\distribFunc}{\mathsf{dist} }
\newcommand{\prob}[2]{\mathbb{P}^{#1}_{#2} }

\newcommand{\MarVal}[1]{\chi_{#1} }

\newcommand{\G}{\mathcal{G} }
\newcommand{\Aconc}{\mathcal{C} }

\newcommand{\Games}[2]{\langle #1,#2 \rangle }

\newcommand{\fsf}{\ensuremath{\mathsf{F}}}

\newcommand{\SetStrat}[2]{\ensuremath{\mathsf{S}_{#1}^{#2} }}
\newcommand{\SetPosStrat}[2]{\ensuremath{\mathsf{PS}_{#1}^{#2} }}

\newcommand{\restr}[2]{\ensuremath{\left.#1\right|_{#2}}}

\newcommand{\cyl}{\ensuremath{\mathsf{Cyl}}}

\newcommand{\outM}{\ensuremath{\mathsf{out}}}
\newcommand{\va}{\ensuremath{\mathsf{val}}}
\newcommand{\N}{\ensuremath{\mathbb{N}}}

\newcommand{\R}{\ensuremath{\mathbb{R}}}

\newcommand{\pat}[1]{{\color{magenta} #1}}
\usepackage{amssymb}
\usepackage{anyfontsize}
\usepackage{tikz}
\usetikzlibrary{arrows,shapes,snakes,automata,backgrounds,positioning}

\tikzstyle{player}=[state,draw=none,circle,align=center,scale=0.8]
\tikzstyle{player1}=[initial,initial text={},initial where=above,scale=0.8,state,draw,circle,align=center]
\tikzstyle{player1'}=[scale=0.8,state,draw,circle,align=center]
\tikzstyle{player1ell}=[scale=0.8,state,draw,ellipse,align=center]
\tikzstyle{player2}=[state,draw,rectangle,align=center]
\tikzstyle{player2'}=[state,draw,diamond,align=center]
\tikzstyle{player3}=[state,draw,rectangle,align=center,style={align=center},scale=0.7]
\tikzstyle{widget}=[draw=red,rectangle, rounded rectangle=10pt,dashed,minimum size=6mm,fill=yellow]
\tikzset{every loop/.style={looseness=7}}

\usepackage{amsmath} 

\usepackage{bm}

\newcommand{\lenH}{1}
\newcommand{\lenV}{1}

\newcommand{\leta}{\ensuremath{{q}_0}}

\newcommand{\lete}{\ensuremath{{q}_4}}
\newcommand{\letf}{\ensuremath{{q}_5}}
\newcommand{\letg}{\ensuremath{{q}_6}}

\definecolor{color0}{rgb}{0.59, 0.47, 0.71}
\definecolor{color1}{rgb}{0.95, 0.25, 0.25}
\definecolor{color2}{rgb}{0.25, 0.41, 0.88}
\definecolor{color3}{rgb}{0.9, 0., 0.}
\definecolor{color4}{rgb}{0, 0.06, 0.54}
\definecolor{colorEdge}{rgb}{0, 0, 0.}

\definecolor{colorWait}{rgb}{0.59, 0.47, 0.71}

\definecolor{colorTreated}{rgb}{1, 0.84, 0.}

\usepackage{authblk}
\author[1,2,3]{Benjamin Bordais}
\author[1]{Patricia Bouyer}
\author[1]{Stéphane Le Roux}
\affil[1]{Université Paris-Saclay, CNRS, ENS Paris-Saclay, LMF, 91190 Gif-sur-Yvette, France}
\affil[2]{TU Dortmund University, Dortmund, Germany}
\affil[3]{Center for Trustworthy Data Science and Security, University Alliance Ruhr, Dortmund, Germany}

\date{}

\begin{document}
	\title{From Local To Global Optimality in Concurrent Parity Games}
	\maketitle
	
	\begin{abstract}
		We study two-player games on finite graphs. Turn-based games have many nice properties, but concurrent games are harder to tame: e.g. turn-based stochastic parity games have positional optimal strategies, whereas even basic concurrent reachability games may fail to have optimal strategies. We study concurrent stochastic parity games, and identify a local structural condition that, when satisfied at each state, guarantees existence of positional optimal strategies for both players.  
		
	\end{abstract}

	\section{Introduction}
	Two-player games played on finite graphs have been a helpful model in areas of computer science. In such games, states/vertices of the graph are colored; the actions of the players induce an infinite path in the graph, thus inducing an infinite sequence of colors. Who wins depends on the color sequence. These games can be turn-based, i.e. at each state a unique player chooses an outgoing edge leading to a probability distribution over successor states, or concurrent, i.e. at each state, the combination of one action per player determines the probability distribution over successor states. In such a stochastic setting, Player $\A$ wants to maximize her probability to win, and Player $\B$ to minimize the very same probability.
	
	We study the above games in the case of parity objective: colors are
	natural numbers, and a sequence is winning for Player $\A$ iff the
	maximal color seen infinitely often is even. This objective has been
	well-studied in connection with model-checking. The turn-based version
	of these games has nice properties involving deterministic positional
	strategies, where "positional" means that the played action depends
	only on the current state: with only deterministic probability distributions over successor states, either of the
	players has a winning such strategy \cite{zielonka98}; with arbitrary probability distributions over successor states, both players have optimal such strategies
	\cite{zielonka04,DBLP:conf/soda/ChatterjeeJH04}. Note that in
	concurrent parity games, positional optimal strategies for distinct
	starting states yield one positional strategy optimal for all states
	uniformly. So we omit the word "uniform" in this paper.
	
	The above properties and others break in basic concurrent games but for safety games, as recorded in the table below, where safety and reachability are special cases of (co-)Büchi, and parity subsumes (co-)Büchi. Büchi games may not have optimal strategies but when they have, they have positional ones; co-Büchi games may not have optimal strategies, and they may have only infinite-memory ones; and the other way around for $\varepsilon$-optimal strategies, hence incomparable difficulty between Büchi and co-Büchi. This shows that concurrent parity is strictly harder than (co-)Büchi, since parity games may have only infinite-memory 
	($\varepsilon$-)optimal strategies. (See also the column on subgame optimal strategies or Appendix~\ref{sec:explain_table}.)
	
	\bigskip
	\hspace*{-0.4cm}
	\label{intro}
	\begin{table}
		\centering
		\begin{tabular}{|c|c|c|c|c|}
			\hline
			& $\exists$ opt strat? & opt & $\varepsilon$-opt & SubG-opt\\
			\hline
			Safety & always \cite{AM04b} & pos. \cite{AM04b} & pos. \cite{AM04b} & pos. \cite{AM04b} \\	
			Reach. & not always \cite{kumar1981existence} & pos. \cite{BBSCSL22} & pos. \cite{everett57} & pos. \cite{BBSCSL22}\\
			Büchi & not always & pos. \cite{BBSFSTTCS22} & $\infty$ \cite{quantitativeAH00} & pos. \cite{BBSFSTTCS22}  \\
			Co-B. & not always & $\infty$ \cite{BBSFSTTCS22} & pos. \cite{AH06} & pos. \cite{BBSFoSSaCs23} \\
			Parity & not always & $\infty$ \cite{quantitativeAH00} & $\infty$ \cite{quantitativeAH00} & $\infty$ \cite{quantitativeAH00} \\ \hline
		\end{tabular}
		\caption{The summary of the situation in concurrent parity games in terms of existence and nature of (subgame) (almost-)optimal strategies.}
		\label{tab:intro}
	\end{table}
	\medskip
	
	Nevertheless, concurrent games' poor behavior in general should not deter us from studying them, as many complex systems are inherently concurrent. See \cite{KNPS+22} for further arguments. Continuing a recent line of research~\cite{BBSFSTTCS21,BBSCSL22,BBSFSTTCS22} we seek local structural good behaviors that scale up to the whole game. 
	
	The following key property still holds in parity concurrent games, by determinacy of Blackwell games \cite{martin98}: each state has a value $u \in [0,1]$, i.e. for all $\varepsilon > 0$ Player $\A$ has an $\varepsilon$-optimal strategy for plays starting in the state, i.e. guaranteeing probability at least $u - \varepsilon$ to win, regardless of Player $\B$'s strategy; and likewise Player $\B$ can guarantee that Player $\A$ wins with probability at most $u + \varepsilon$. 
	Here we are interested in strategies that are optimal, i.e. that
	realize exactly the value $u$, and that are positional. They do not
	exist in general (see~\cite[Figure 2]{BBSCSL22}), but our result below is a transfer from local to global positional optimality, and the terminology that it uses is explained afterwards.
	\begin{theorem}
		\label{thm:intro}
		If at every state the induced game form is positionally optimizable, both players have positional optimal strategies in the game.
	\end{theorem}	
	A game form is a map from pairs of actions to a probability distribution over successor states, see Figure~\ref{fig:GFEnv} to the left, or \cite{BBSFSTTCS21} page 5 forcmore. So each state induces a game form. Given a game form $\formNF$, an $\formNF$-game has a unique non-trivial state which induces
	$\formNF$ 
	and the other player states are of two kinds: first colored kind, they loop
	back to the non-trivial state; second uncolored kind, they have an explicit value in $[0,1]$, and the game stops there. Moreover Player $\A$ prefers higher explicit values. See Figure~\ref{fig:GFEnv} or Figure~\ref{fig:extracted_environment} for examples of such a construction. For all $n \in \N$, we say that $\formNF$ is called \emph{positionally optimizable up to $n$} if both players have positional optimal strategies in all $\formNF$-games with at most $n+1$ different colors. We say that $\formNF$ is positionally optimizable if it positionally optimizable up to $n$ for all $n \in \N$. An overview of the proof of Theorem~\ref{thm:intro} is given in Subsection~\ref{subsec:big_picture}, after having introduced the relevant definitions. 
	
	As a corollary of Theorem 1, we recover the known result that positional strategies are sufficient to play optimally in turn-based stochastic
	games~\cite{zielonka04,DBLP:conf/soda/ChatterjeeJH04}.
	
	The version of Theorem~\ref{thm:intro} that we show in this paper (i.e. Theorem~\ref{thm:transfer_local_global}) is actually a little more precise. Indeed, for Theorem~\ref{thm:intro} to hold, it is actually sufficient that all the game forms appearing in the game are positionally optimizable up to $n$ where $n$ denotes the different between the highest and least color appearing in the game. After having proved Theorem~\ref{thm:transfer_local_global}, we show that it is decidable if a game form is positionally optimizable (resp. positionally optimizable up to $n$), see Proposition~\ref{prop:positionally_optim_decidable}. Furthermore, we also show that, for all $n \in \N$, the set of game forms positionally optimizable up to $n$ is strictly larger that the set of game forms up to $n+1$, see Theorem~\ref{prop:inclustion_parm}. 
	
	Finally, note that the results presented in this paper appear in the PhD thesis of one of the authors, which will soon be available online.

\section{Preliminaries and Game Forms}
\label{sec:preliminaries}
Consider a non-empty set $Q$. We denote by $Q^*$, $Q^+$ and $Q^\omega$ the set of finite sequences, non-empty finite sequences and infinite sequences of elements of $Q$ respectively. 


For all subsets $S \subseteq Q$, we let $\mathbbm{1}_S: Q \rightarrow \{ 0,1 \}$ be the indicator function of the set $S$ defined by, for all $q \in Q$, we have $\mathbbm{1}_S(q) = 1$ if and only if $q \in S$. When $S = \{ q \}$ for any $q \in Q$, we write $\mathbbm{1}_q$ for $\mathbbm{1}_S$.

A \emph{discrete probability distribution} over a non-empty finite set $Q$ is a function $\mu: Q \rightarrow [0,1]$ such that $\sum_{x \in Q} \mu(x) = 1$. The distribution $\mu$ is \emph{deterministic} if $\mu(q) = 1$ for some $q \in Q$. For all $S \subseteq Q$, we let $\mu[S] := \sum_{q \in S} \mu(q)$. 
The set of distributions over the set $Q$ is denoted $\Dist(Q)$.	

For all $i \leq j \in \N$, we denote by $\brck{i,j}$ the set of integers between $i$ and $j$: $\brck{i,j} := \{ k \in \N \mid i \leq k \leq j \}$. This set is typed in the sense that these are seen as integers and not reals numbers: $\brck{0,1} \cap [0,1] = \emptyset$\footnote{This is done because, in parity games, the integers refer to the colors appearing in the game whereas the real numbers between in $[0,1]$ refer to the values of the states in the game. These are disjoint.}. For all finite sets $S \subseteq \N$, we let $\msf{Even}(S)$ (resp. $\msf{Odd}(S)$) be the smallest even (resp. odd) integer that is greater than or equal to all elements in $S$. 

We recall the definition of game forms, which 
informally are finite two-dimensional tables with variables, and of games in normal forms which are game forms whose outcomes are values between $0$ and $1$.
\begin{definition}[Game form and game in normal form]
	\label{def:arena_game_nf}
	Consider a non-empty set of outcomes $\outComeNF$. A \emph{game form} ($\GF$ for short) on $\outComeNF$ is a tuple
	$\formNF = \langle \Act_\A,\Act_\B,\outComeNF,\outCNF \rangle$
	where $\Act_\A$ (resp. $\Act_\B$) is the non-empty finite set of actions available to Player $\A$ (resp. $\B$) and $\outCNF: \Act_\A \times \Act_\B \rightarrow \Dist(\outComeNF)$ is a function that associates a distribution over the outcomes to each pair of actions. We denote by $\Frm(\outComeNF)$ the set of all game forms on $\outComeNF$. 
	
	When the set of outcomes $\outComeNF$ is equal to $[0,1]$, we say that $\formNF$ is a \emph{game in normal	form}. For a valuation $v: \outComeNF \rightarrow [0,1]$ of the outcomes, the notation $\gameNF{\formNF}{v}$ refers to the game in normal form $\langle \Act_\A,\Act_\B,[0,1],v \circ \outCNF \rangle$ induced from $\formNF$ by $v$.
	
	Finally, a Player-$\A$ (resp. Player-$\B$) game
	form $\formNF$ is such that $|\Act_\B| = 1$ (resp. $|\Act_\A|
	= 1$). The only Player-$\B$ (resp. Player-$\A$) action is denoted
	$*$. Such a game form is denoted $\formNF =
	\langle \A:\outCNF(\Act_\A,*) \rangle$ (resp. $\formNF =
	\langle \B:\outCNF(*,\Act_\B) \rangle$)
	. When $\formNF$ is both a Player-$\A$ and a Player-$\B$ game form, it is called a \emph{trivial} game form. Such a game form is denoted $\formNF = \{ \outCNF(*,*) \}$.
	%
	%
\end{definition}
An example of a game form is depicted on the left of
Figure~\ref{fig:GFEnv} where the actions available to Player $\A$ are
the rows and the actions available to Player $\B$ are the columns.
Strategies available to Player-$\A$ (resp. $\B$) are then the probability distributions over their respective sets of actions. In a game in normal form, Player $\A$ tries to maximize the outcome, whereas Player $\B$ tries to minimize
it.
\begin{definition}[Outcome and value of a game in normal form]
	\label{def:outcome_game_form}
	Consider a game in normal form
	$\formNF = \langle \Act_\A,\Act_\B,[0,1],\outCNF \rangle$. The
	set $\Dist(\Act_\A)$ (resp. $\Dist(\Act_\B)$) corresponds to the
	set of strategies available to Player $\A$
	(resp. $\B$), it is denoted $\Sigma_\A(\formNF)$ (resp. $\Sigma_\B(\formNF)$). For a pair of strategies
	$(\sigma_\A,\sigma_\B) \in \Sigma_\A(\formNF) \times
	\Sigma_\B(\formNF)$, the outcome
	$\outM_\formNF(\sigma_\A,\sigma_\B)$ in $\formNF$ of the
	strategies $(\sigma_\A,\sigma_\B)$ is defined as:
	$\outM_\formNF(\sigma_\A,\sigma_\B) := \sum_{a \in \Act_\A}
	\sum_{b \in \Act_\B} \sigma_\A(a) \cdot \sigma_\B(b) \cdot
	\outCNF(a,b) \in [0,1]$. 
	
	Consider a Player-$\A$ strategy $\sigma_\A \in
	\Sigma_\A(\formNF)$.
	The \emph{value} of strategy $\sigma_\A$, denoted
	$\va_\formNF(\sigma_\A)$ is equal to:
	$\va_\formNF(\sigma_\A) := \inf_{\sigma_\B \in
		\Sigma_\B(\formNF)} \outM_{\formNF}(\sigma_\A,\sigma_\B)$,
	and dually for Player $\B$
	. When
	$\sup_{\sigma_\A \in \Sigma_\A(\formNF)}
	\va_\formNF(\sigma_\A) = \inf_{\sigma_\B \in
		\Sigma_\B(\formNF)} \va_\formNF(\sigma_\B)$, it defines the
	\emph{value} of the game $\formNF$, denoted $\va_\formNF$.  If
	$\va_\formNF(\sigma_\A) = \va_\formNF$, the strategy
	$\sigma_\A$ is said to be \emph{optimal} for Player $\A$. This
	is defined analogously for Player $\B$. Since the sets of actions is finite, Von Neumann's minimax
	theorem~\cite{vonNeuman} ensures the existence of a value and
	of optimal strategies for both players.
	%
\end{definition}
In the following, strategies in game forms will be called
$\GF$-strategies in order not to confuse them with strategies in
concurrent games (on graphs).

\section{Concurrent games, local interactions
}
\label{sec:concurrent_games}
\subsection{Concurrent arenas and games}
\begin{definition}[Finite stochastic concurrent arena]
	A finite \emph{concurrent arena} $\Aconc$ is a tuple 
	$\langle Q,\fsf \rangle$ where $Q$ is a non-empty finite set of states and
	$\fsf: Q \rightarrow \Frm(Q)$ maps each state to its
	induced game form, which describes the interaction of
	the players at this state.
\end{definition}

In the following, the arena $\Aconc$ will refer to the tuple $\langle Q,\fsf \rangle$, unless otherwise stated. Furthermore, for all $q \in Q$, we denote by $\Act_\A^q$ and $\Act_\B^q$ the sets of actions available to both players in the local interaction $\fsf(q)$, unless otherwise stated.
%
In this paper, we focus on (max) parity objectives: 
the goal of Player $\A$ is that the maximum of the colors (i.e. label on states) visited infinitely often is even.
\begin{definition}[Parity game]
	Consider a finite set of states $Q$ and a coloring function
	$\colFunc: Q \rightarrow
	\N$. 
	This coloring function induces
	the 
	\emph{parity objective} $W(\colFunc) \subseteq Q^\omega$ defined by
	$W(\colFunc) := \{ \rho \in Q^\omega \mid
	\max(\colFunc(\rho)_\infty) \text{ is even } \}$ where
	$\colFunc(\rho)_\infty := \{ k \in \N \mid \forall i \in \N,\;
	\exists j \geq i,\; \colFunc(\rho)_j = k \} \neq \emptyset$ (since
	$Q$ is finite) denotes the set of colors seen infinitely often in
	$\colFunc(\rho)$. 
	
	A \emph{parity game} $\G = \Games{\Aconc}{\colFunc}$ is a pair formed of a concurrent arena $\Aconc$ and a coloring function $\colFunc: Q \rightarrow
	\N$. 
\end{definition}

In such a game, strategies 
map the history of the game (i.e. the finite sequence of states visited so far)
to a $\GF$-strategy in the game form 
corresponding to the current state of the game.
\begin{definition}[Strategies]
	Consider a concurrent arena $\Aconc$. A \emph{strategy} for Player
	$\A$ is a function
	$\s_\A: \cup_{q \in Q} (Q^* \cdot q \rightarrow
	\Sigma_\A(\fsf(q)))$. 
	It is \emph{positional} if for
	all $q \in Q$, there is a $\GF$-strategy
	$\sigma_\A^q \in \Sigma_\A(\fsf(q))$ such that, for all
	$\pi = \rho \cdot q \in Q^+$: $\s_\A(\pi) = \sigma_A^q$. In that
	case, the strategy $\s_\A$ is said to be \emph{defined} by
	$(\sigma_\A^q)_{q \in Q}$. 
	
	We denote by $\SetStrat{\Aconc}{\A}$ and
	$\SetPosStrat{\Aconc}{\A}$ the set of all strategies and positional
	strategies respectively in arena $\Aconc$ for Player $\A$. A strategy $\s_\A$ is \emph{deterministic} if, for all $\rho \in Q^+$, $\s_\A(\rho)$ is deterministic. The definitions are analogous for Player $\B$.
\end{definition}

Unlike deterministic games with deterministic strategies, the
outcome of a game, given two strategies (one for each Player), is not a single play but
rather a distribution over plays. To formalize this, we first define 
the probability to go from a state $q$ to a state $q'$ given two
$\GF$-strategies in $\fsf(q)$.
\begin{definition}[Probability transition given two strategies]
	\label{def:probability_transition}
	Consider a concurrent arena $\Aconc
	$, a state $q \in Q$ and two strategies $(\sigma_\A,\sigma_\B)
	\in \Sigma_\A(\fsf(q)) \times \Sigma_\B(\fsf(q))$. 
	Let $q' \in Q$. The probability to go from $q$ to $q'$ if the players plays, in q, $\sigma_\A$ and $\sigma_\B$, denoted $\prob{\sigma_\A,\sigma_\B}{}{}(q,q')$, is equal to:
	\begin{displaymath}
	\prob{\sigma_\A,\sigma_\B}{}{}(q,q') := \outM_{\Games{\fsf(q)}{\mathbbm{1}_{q'}}}(\sigma_\A,\sigma_\B)
	\end{displaymath}
\end{definition}

We now define the probability of occurrence of finite paths, and
consequently of any Borel set, given a strategy per player.
\begin{definition}[Probability distribution given two strategies]
	\label{def:probability_in_parity_game}
	Consider a concurrent arena $\Aconc
	$ and two arbitrary strategies $(\s_\A,\s_\B) \in \SetStrat{\Aconc}{\A} \times \SetStrat{\Aconc}{\B}$.
	We denote by $\mathbb{P}^{\s_\A,\s_\B}: Q^+ \rightarrow \Dist(Q)$ the function giving the probability distribution over the next state of the arena given the sequence of states already seen. That is, for all finite path $\pi = \pi_0 \ldots \pi_n \in Q^+$ and $q \in Q$, we have:
	\begin{displaymath}
	\mathbb{P}^{\s_\A,\s_\B}(\pi)[q] =
	\prob{\s_\A(\pi),\s_\B(\pi)}{}{}(\pi_n,q)
	\end{displaymath}
	
	The probability 
	of a finite path $\pi = \pi_0 \cdots \pi_n \in Q^+$ from a state $q_0 \in Q$ with the pair of strategies $(\s_\A,\s_\B)$ is then equal to $\prob{\Aconc,q_0}{\s_\A,\s_\B}(\pi) := \Pi_{i = 0}^{n-1} \mathbb{P}^{\s_\A,\s_\B}(\pi_{\leq i})[\pi_{i+1}]$ if $\pi_0 = q_0$ and $0$ otherwise. 
	The probability of a cylinder set $\cyl(\pi) := \{ \pi \cdot \rho \mid \rho \in Q^\omega \}$ is $\prob{\Aconc,q_0}{\s_\A,\s_\B}[\cyl(\pi)] := \prob{\s_\A,\s_\B}{}(\pi)$ for any finite path $\pi \in Q^*$. This induces the probability measure over Borel sets in the usual way. 
	We denote by $\prob{\Aconc,q_0}{\s_\A,\s_\B}$ this probability measure. mapping each Borel set to a value in $[0,1]$
\end{definition}

The values of strategies and of the game follow.
\begin{definition}[Value of strategies and of the game]
	Let $\G = \Games{\Aconc}{\colFunc}$ be a parity game 
	and $\s_\A \in \SetStrat{\Aconc}{\A}$ be a Player-$\A$ strategy. The vector $\MarVal{\G}[\s_\A]: Q \rightarrow [0,1]$ giving the value of the strategy $\s_\A$ is such that, for all $q_0 \in Q$, we have $\MarVal{\G}[\s_\A](q_0) := \inf_{\s_\B \in \SetStrat{\Aconc}{\B}} \prob{\Aconc,q_0}{\s_\A,\s_\B}[W(\colFunc)]$. The vector $\MarVal{\G}[\A]: Q \rightarrow [0,1]$ giving the value for Player $\A$ is such that, for all $q_0 \in Q$, we have $\MarVal{\G}[\A](q_0) := \sup_{\s_\A \in \SetStrat{\Aconc}{\A}} \MarVal{\G}[\s_\A](q_0)$. The vector $\MarVal{\G}[\B]: Q \rightarrow [0,1]$ giving the value of the game for Player $\B$ is defined similarly by reversing $\inf$ and $\sup$.
	
	By Martin's result on the determinacy of Blackwell
	games \cite{martin98}:
	$\MarVal{\G}[\A] = \MarVal{\G}[\B]$, which defines the
	\emph{value} of the game:
	$\MarVal{\G} := \MarVal{\G}[\A] = \MarVal{\G}[\B]$. A
	Player-$\A$ strategy $\s_\A$ such that
	$\MarVal{\G} := \MarVal{\G}[\s_\A]$ is \emph{optimal} (and
	similarly for Player $\B$).
	\label{def:determinacy}
\end{definition}
Note that optimal strategies may not exist in general; when they
exist they can be arbitrarily complex; see the table page~\pageref{intro}.

Finally, we extend our formalism by considering stopping states
with output value, i.e. states that, when visited, immediately
stop the game and induce a specific value in $[0,1]$. The fact that
the value of a stopping state $q$ is set to be $u$ is denoted
$\msf{val}(q) \gets u$. Note that such stopping states can be encoded in
our formalism (see
Appendix~\ref{appen:implement_stopping_states}). They will be
depicted as dashed states.

\subsection{Markov chains and sufficient condition for optimality}
In this subsection, we give a condition on positional strategies to be
optimal in a concurrent parity game. First, we need to introduce the
notions of Markov chain and bottom strongly connected component.
\begin{definition}[Markov chain, BSCC]
	A \emph{Markov chain} $\mathcal{M}$ is a pair
	$\mathcal{M} = \langle Q,\mathbb{P} \rangle$ where $Q$ is a finite
	set of states and $\mathbb{P}: Q \rightarrow \Dist(Q)$. A
	\emph{bottom strongly connected component} (BSCC for short)
	$H \subseteq Q$ is a subset of states such that the underlying graph
	of $H$ is strongly connected 
	(w.r.t. edges given by positive  probability transitions from states to states) and $H$ cannot be exited: for all $q \in H$ and $q' \in Q$, $\mathbb{P}(q)(q') > 0$ implies $q' \in H$.
	
	We say that a subset of states $S \subseteq Q$ \emph{occurs} in $H$ if $S \cap H \neq \emptyset$. We say that a state $q \in Q$ \emph{occurs} in $H$ if $q \in H$.
\end{definition}
Two positional strategies (one per player) in a concurrent arena
not only induce a probability measure on infinite sequences of states
(recall Definition~\ref{def:probability_in_parity_game}), but also a
Markov chain, whose graph is a subgraph of the arena.
(This is done formally in Appendix~\ref{appen:markov_chain}.) If we only fix a strategy, we can consider the BSCCs that are
compatible with that strategy as follows.

\begin{definition}[Induced Markov chains, BSCCs compatible with a strategy]
	Consider a parity game $\G = \Games{\Aconc}{\colFunc}$.
	Let $\s$ be a positional strategy for one of the players. For
	every positional and deterministic strategy $\s'$ for the other
	player, we denote by
	$\mathcal{M}_{\s,\s'} = \langle Q,\mathbb{P}^{\s,\s'} \rangle$ the
	Markov chain \emph{induced} by $\s$ and $\s'$. 
	
	We let $H_{\s}$
	denote the set of BSCCs \emph{compatible} with $\s$, i.e. the
	BSCCs of some Markov chain $\mathcal{M}_{\s,\s'}$, where $\s'$
	ranges over ther other player's positional deterministic
	strategies. A BSCC $H \in H_{\s}$ is \emph{even-colored} if
	$\max \colFunc[H]$ is even. Otherwise, it is \emph{odd-colored}.
\end{definition}

We now define several properties relating a positional strategy and a valuation of the states.
\begin{definition}
	\label{def:dominate_strongly_dominate_gurantee}
	Let $\G = \Games{\Aconc}{\colFunc}$ be a 
	game and $v: Q \rightarrow [0,1]$ be a valuation over its states. Consider
	a positional Player-$\A$ strategy
	$\s_\A$ (resp. Player-$\B$ strategy $\s_\B$). 
	The strategy $\s_\A$ (resp. $\s_\B$):
	\begin{itemize}
		\item \emph{dominates} the valuation $v$ if for all $q \in Q$, it
		holds that $v(q) \leq \va_{\Games{\fsf(q)}{v}}(\s_\A(q))$
		(resp. $v(q) \geq \va_{\Games{\fsf(q)}{v}}(\s_\B(q))$);
		\item \emph{parity dominates} the valuation $v$ if it dominates $v$
		and all BSCCs $H$ compatible with $\s_\A$ (resp. $\s_\B$)
		such that $\min v[H] > 0$ (resp. $\max v[H] < 1$) are even-colored
		(resp. odd-colored);
		\item \emph{guarantees} the valuation $v$ if, for all $q \in Q$, it
		holds $v(q) \leq \MarVal{\G}[\s_\A](q)$ (resp. $v(q) \geq \MarVal{\G}[\s_\B](q)$).
	\end{itemize}
\end{definition}
Parity dominating a valuation implies guaranteeing it (the theorem
below can be deduced from~\cite[Lemma 34]{BBSArXivICALP}). 
\begin{theorem}[Proof~\ref{appen:proof_thm_even_in_bscc_ok}]
	\label{thm:strongly_domnates_imply_guarantee}
	Consider a parity game $\G = \Games{\Aconc}{\colFunc}$, a
	Player-$\A$ positional strategy $\s_\A \in \SetPosStrat{\Aconc}{\A}$
	and a valuation of the states $v: Q \rightarrow [0,1]$. If the strategy $\s_\A$	dominates the valuation $v$, then for all BSCCs $H \in H_{\s_\A}$, there is some $v_H \in [0,1]$ such that $v[H] = \{ v_H \}$.If in addition $\s_\A$ parity dominates $v$, it also guarantees $v$.
\end{theorem}

\section{Local Environment}
The goal of this section is to define simple parity games with a
single non-trivial local interaction.
This will allow us to define exactly the game forms that should be
used in parity games if one requires positional optimal strategies for both players. We first define what a (parity) environment on a given set of outcomes is. We can then define the parity game induced by such an
environment (along with a game form). 
\begin{definition}[Parity environment]
	\label{def:parity_env}
	Consider a non-empty finite set of outcomes $\outComeNF$. An
	\emph{environment} $E$ on $\outComeNF$ is a tuple
	$E := \langle c,e,p \rangle$ where $c,e \in \N$ with $c \leq e$ and
	$p: \outComeNF \rightarrow \{q_{\msf{init}}\} \uplus K_e
	\uplus [0,1]$. We let $p_{ [0,1] } := p[\outComeNF] \cap [0,1]$
	. The \emph{size} w.r.t. Player $\A$ (resp. $\B$) $\msf{Sz}_\A(E)$ (resp. $\msf{Sz}_\B(E)$) of the environment $E$ is equal to $\msf{Sz}_\A(E) := \msf{Even}(e)-c$ (resp.  $\msf{Sz}_\B(E) := \msf{Odd}(e)-c$). We denote by
	$\mathsf{Env}(\outComeNF)$ the set of all environments on the set
	of outcomes $\outComeNF$. 
\end{definition}
\begin{remark}
	\label{rmk:size_environement}
	A quick note on the size of an environment. From Player $\A$'s perspective, 
	the size of an environment $E := \langle c,e,p \rangle$, assuming that $c = 0$, is equal to: 0 if $e = 0$, 2 if $e = 1$ or $e = 2$, etc. 
	Informally, when $c = 0$, if $e = 0$ this corresponds to safety game, if $e = 1$ this corresponds to a co-Büchi game, etc. This will become more apparent with Definition~\ref{def:parity_game_from_gf} below.
\end{remark}

We can then define the simple parity game corresponding to a parity environment.
\begin{definition}[Parity game induced by an environment]
	\label{def:parity_game_from_gf}
	Consider a non-empty finite set of outcomes $\outComeNF$, a
	game form $\formNF \in \Frm(\outComeNF)$ and an environment	$E = \langle c,e,p \rangle \in \msf{Env}(\outComeNF)$. Let $Y := (\formNF,E)$. The local arena $\Aconc_{Y} = \langle Q,\fsf \rangle$ and the coloring function $\colFunc_Y$ induced by $Y$ is such that:
	\begin{itemize}
		\item $Q := \{ q_\mathsf{init} \} \cup K_e \cup p_{[0,1]}$, $Q_\msf{s} := p_{[0,1]}$, and for all $u \in p_{[0,1]}$, we set the value of the stopping state $u$ to be $u$ itself: $\msf{val}(u) \gets u$; 
		\item $\fsf(q_\msf{init}) := \formNF^p = \langle \Act_\A,\Act_\B,Q,\outCNF_p \rangle$ where, for all $\sigma_\A \in \Dist(\Act_\A)$, $\sigma_\B \in \Dist(\Act_\B)$, and $q \in Q$, we have $\outCNF_p(\sigma_\A,\sigma_\B)(q) := \outCNF(\sigma_\A,\sigma_\B)[p^{-1}[q]]$;
		\item for all $i \in \brck{0,e}$, the set $k_i$ is trivial, its only outcome is the state $\qinit$;
		\item $\colFunc_Y(q_\msf{init}) := c$ and for all $i \in \brd$, we have $\colFunc_Y(k_i) := i$.
	\end{itemize}
	For all $u \in [0,1]$, we denote by $v_Y^u: Q \rightarrow [0,1]$ the valuation such that: $v_Y^u(q_\msf{init}) = v_Y^u(k_i) := u$ for all $i \in \brd$ and $v_Y^u(x) := x$ for all $x \in p_{[0,1]}$. Furthermore, for all Player-$\A$ $\GF$-strategies $\sigma_\A \in \Sigma_\A(\formNF)$, we denote by $\s_\A^Y(\sigma_\A) \in \msf{S}_\A^{\Aconc_Y}$ the Player-$\A$ positional strategy defined by $\sigma_\A$ in the arena $\Aconc_Y$. 
	
	The game $\G_{Y}$ is then equal to $\G_{Y} := \Games{\Aconc_{Y}}{\colFunc_Y}$.
\end{definition}

\begin{example}
	Definition~\ref{def:parity_game_from_gf} above is illustrated in Figure~\ref{fig:GFEnv}
	. Note that the colors of the non-stopping states are depicted in red next to the states. 
\end{example}

As in the previous chapter, we can define the notion of Player-$\A$ $\GF$-strategy being optimal in a parity environment.
\begin{definition}[Optimal $\GF$-strategies]
	\label{def:optimal_strat_in_gf}
	Consider a non-empty finite set of outcomes $\outComeNF$, a
	game form $\formNF \in \Frm(\outComeNF)$, an environment $E =
	\langle c,e,p \rangle \in \msf{Env}(\outComeNF)$, and let $Y
	:= (\formNF,E)$. A Player-$\A$ $\GF$-strategy $\sigma_\A \in \Sigma_\A(\formNF)$ is said to be \emph{optimal} w.r.t. $Y$ if the Player-$\A$ positional strategy $\s_\A^Y(\sigma_\A)$ is optimal in $\G_Y$. The definition is analogous for Player $\B$.
\end{definition}

Given a finite set of outcomes $\outComeNF$, we can now define the game forms on $\outComeNF$ ensuring the existence of optimal strategies w.r.t. all environments. 
\begin{definition}[Game forms with optimal strategies]
	\label{def:positionally_maximizable}
	Consider a non-empty finite set of outcomes $\outComeNF$, a game form $\formNF \in \Frm(\outComeNF)$ and some $n \in \N$. 
	
	Consider a Player $\msf{C} \in \{ \A,\B \}$. The game form $\formNF$ is said to be \emph{positionally maximizable up to} $n$ w.r.t. Player $\msf{C}$ if, for each environment $E \in \msf{Env}(\outComeNF)$ with $\msf{Sz}_\msf{C}(E) \leq n$, there is an optimal $\GF$-strategy for Player $\msf{C}$ w.r.t. $(\formNF,E)$. 
	
	When this holds for both Players, $\formNF$ is said to be \emph{positionally optimizable up to $n$}. The corresponding set of game forms is denoted $\msf{ParO}(n)$. If this holds for
	all $n \in \N$, $\formNF$ is simply said to be
	\emph{positionally optimizable}, and the corresponding set of game forms is denoted $\msf{ParO}$. 
\end{definition}

\begin{remark}
	\label{rmk:parm_necessary}
	First, note that all game forms are positionally optimizable up to 0. This is because environement of size 0 induce safety games. However, there are some game forms that are not positionally maximizable w.r.t. any player up to 1. 
	
	Furthermore, by definition, from a game form $\formNF \in
	\Frm(\outComeNF)$ that is not positionally optimizable up to
	some $n \in \N$, there exists an environment $E \in
	\msf{Env}(\outComeNF)$ such that either of the players has no
	positional optimal strategy in the simple parity game $\G_{(\formNF,E)}$ where the difference between $\colFunc(q_{\msf{init}})$ and the maximum of the
	colors appearing in $\G_{(\formNF,E)}$ is at most $n$. 
\end{remark}
In the game $\G_{(\formNF,E)}$ depicted on the right of Figure~\ref{fig:GFEnv}, Player $\A$ has positional optimal strategies: it suffices to play both rows with positive probability. (This is similar for Player $\B$.) As a side remark, the game form in the left of Figure~\ref{fig:GFEnv} is positionally optimizable.

In Lemma~\ref{lem:opt_in_gf_equiv_opt_in_parity} below, we formulate more explicitly (using the notion of parity domination from
Definition~\ref{def:dominate_strongly_dominate_gurantee}) what optimal $\GF$-strategies are.

\begin{lemma}[Proof~\ref{appen:proof_prop_opt_in_gf_equiv_opt_in_parity}]
	\label{lem:opt_in_gf_equiv_opt_in_parity}
	Consider a non-empty finite set of outcomes $\outComeNF$, a game form $\formNF\in \Frm(\outComeNF)$, an environment
	$E = \langle c,e,p\rangle \in \msf{Env}(\outComeNF)$ and
	$Y := (\formNF,E)$. A Player-$\A$ $\GF$-strategy $\sigma_\A \in \Sigma_\A(\formNF)$ is \emph{optimal} w.r.t. $Y$ if and only if, letting
	$u := \MarVal{\G_Y}(q_\msf{init})$, either (i) $u = 0$, or (ii)
	the positional Player-$\A$ strategy $\s^Y_\A(\sigma_\A)$ parity
	dominates the valuation $v_Y^u$.
	
	Furthermore (ii) is equivalent to: 
	\begin{itemize}
		\item (1) the Player-$\A$ positional
		strategy $\s^Y_\A(\sigma_\A)$ dominates the valuation $v_Y^u$, i.e. $\sigma_\A$ is optimal in the game in normal form $\Games{\formNF}{v^u_Y \circ p}$; and
		\item (2) for all $b \in \Act_\B$, if $\outM_{\Games{\formNF}{\mathbbm{1}_{p^{-1}[0,1]}}}(\sigma_\A,b) = 0$
		(i.e. the probability under $\sigma_\A$ and $b$ to reach a stopping state is null), then $\max (\mathsf{Color}(\formNF,p,\sigma_\A,b) \cup \{ c \})$ is even where $\mathsf{Color}(\formNF,p,\sigma_\A,b) := \{ i \in \brd \mid 
		\outM_{\Games{\formNF}{\mathbbm{1}_{p^{-1}[k_i]}}}(\sigma_\A,b) > 0 \}$ is the set
		of colors that can be seen with positive probability under
		$\sigma_\A$ and $b$. 
	\end{itemize}
	This is symmetrical for Player $\B$.
\end{lemma}

\begin{remark}
	\label{rmk:explanation_optim_gf_strat}
	Informally, this proposition states that for a Player-$\A$
	$\GF$-strategy $\sigma_\A$ to be optimal in a simple game
	$\G_Y$ with positive value, it must be the case that for every
	Player-$\B$ action $b$: either there is a positive probability
	(w.r.t. $\sigma_\A$ and $b$) to exit $q_\msf{init}$ and 
	the expected value of the stopping states visited is at
	least $u$; or the game loops on $q_\msf{init}$ with
	probability $1$, and 
	the maximum of the colors that can be seen with positive probability (w.r.t. $\sigma_\A$ and $b$) is even. In particular, if
	$c \leq \max \mathsf{Color}(\formNF,p,\sigma_\A,b)$ or if $c$ is
	odd, then $\max \mathsf{Color}(\formNF,p,\sigma_\A,b)$ is even.
\end{remark}

\section{The main theorem}
\label{sec:global_environment_global_strategy}
\begin{figure}
	\centering
	\begin{minipage}{0.3\linewidth}
		\resizebox{1.2\linewidth}{!}{\input{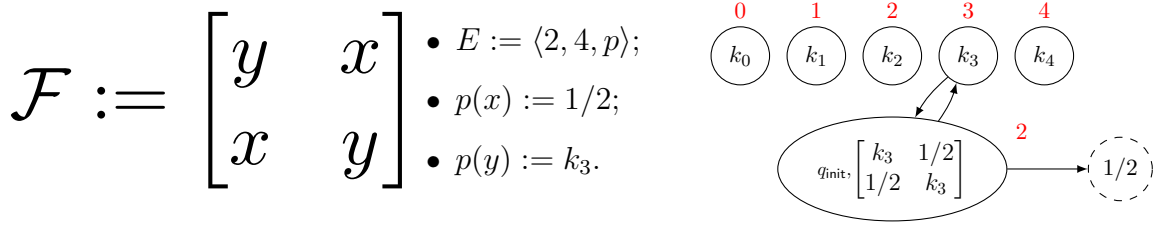}}
		\label{fig:ASimpleGF}
	\end{minipage}
	\begin{minipage}{0.24\linewidth}
		\begin{itemize}
			\item[{$\bullet$}] {\large $E := \langle 2,4,p \rangle$;}
			\item[{$\bullet$}] {\large $p(x) := 1/2$;}
			\item[{$\bullet$}] {\large $p(y) := k_3$.}
		\end{itemize}
	\end{minipage}
	\begin{minipage}{0.42\linewidth}
		\centering
		\begin{tikzpicture}
	\node[player] (cs0) at (4.7*\lenV,2.5*\lenH) {\color{red}2} ;
	\node[player1ell,scale=0.8] (s0) at (3*\lenV,2*\lenH) {$q_{\msf{init}}$,{\Large$\begin{bmatrix}
			k_3 & 1/2 \\
			1/2 & k_3
			\end{bmatrix}$}} ;
	\node[player] (c0) at (1*\lenV,4.1*\lenH) {\color{red}0} ;
	\node[player1'] (k0) at (1*\lenV,3.5*\lenH) {$k_0$} ;
	\node[player] (c1) at (2*\lenV,4.1*\lenH) {\color{red}1} ;
	\node[player1'] (k1) at (2*\lenV,3.5*\lenH) {$k_1$} ;
	\node[player] (c2) at (3*\lenV,4.1*\lenH) {\color{red}2} ;
	\node[player1'] (k2) at (3*\lenV,3.5*\lenH) {$k_2$} ;
	\node[player] (c3) at (4*\lenV,4.1*\lenH) {\color{red}3} ;
	\node[player1'] (k3) at (4*\lenV,3.5*\lenH) {$k_3$} ;
	\node[player] (c4) at (5*\lenV,4.1*\lenH) {\color{red}4} ;
	\node[player1'] (k4) at (5*\lenV,3.5*\lenH) {$k_4$} ;
	\node[player1',dashed] (s2) at (6*\lenV,2*\lenH) {\color{black}$1/2$} ;
	
	\path[-latex]  	
	(s0) edge[bend right=10] (k3)
	(k3) edge[bend right=10] (s0)
	(s0) edge (s2)
	;
\end{tikzpicture}
	\end{minipage}
	\caption{On the left, a game form on the  set of outcomes
		$\outComeNF := \{ x,y \}$, in the middle the description
		of an environment on $\outComeNF$ and on the right the parity game $\G_{(\formNF,E)}$ obtained from what is depicted on the left. The dashed state is a stopping of value $1/2$.}
	\label{fig:GFEnv}
\end{figure}

The goal of this section is to formally state the main theorem of this paper,
Theorem~\ref{thm:transfer_local_global} below. Informally, this theorem consists in extracting, for every state of a game and for each Player, a local
environment which will summarize the context of the state to the
Player, and tell her 
how to play optimally (and positionally).

Before stating this theorem, let us recall and define below some useful notations, in particular we recall the notation for value slices.
\begin{definition}[Value slice]
	Consider a parity game $\G$. For all subsets of states $S \subseteq Q$, we denote by $V_S := \{ u \in [ 0,1] \mid \exists q \in S,\; \MarVal{\G}(q) = u \}$ the finite set of values of states in $S$. Furthermore, for all $u \in V_Q$, we let $Q_u := \{ q \in Q \mid \MarVal{\G}(q) = u \}$ be the set of states whose value is $u$: it is the $u$-slice of $\G$. Finally, for all $u \in V_Q$, we let $e_u := \msf{Even}(\colFunc[Q_u])$ and $o_u := \msf{Odd}(\colFunc[Q_u])$.
\end{definition}

We also introduce the notion of positional strategies generated by an environment function: this is a strategy that, at each state, plays a $\GF$-strategy that is optimal in the corresponding environment. 
\begin{definition}[Positional strategy generated by an environment function]
	\label{def:strat_induced_by_environment}
	For all environment functions
	$\msf{Ev}: Q \rightarrow \msf{Env}(\outComeNF)$, a
	Player-$\A$ positional strategy $\s_\A$ is \emph{generated} by
	$\msf{Ev}$ if for all $q \in Q$, the $\GF$-strategy $\s_\A(q)$
	is optimal w.r.t.
	$(\fsf(q),\msf{Ev}(q))$.
	
	This is similar for Player $\B$.
\end{definition}

We can now state the main result of this chapter: Theorem~\ref{thm:transfer_local_global}.
\begin{theorem}
	\label{thm:transfer_local_global}
	Let $\G = \Games{\Aconc}{\colFunc}$ be a parity game. Assume that for all $q \in Q$, letting $u := \MarVal{\G}(q) \in [0,1]$, the game form $\fsf(q)$ is:
	\begin{itemize}
		\item positionally
		maximizable up to $e_{u} - \colFunc(q)$ w.r.t. Player $\A$; and
		\item positionally maximizable up to $o_{u} - \colFunc(q)$ w.r.t. Player $\B$.
	\end{itemize}
	Then, there is a function $\msf{Ev}_\A: Q \rightarrow \msf{Env}(\outComeNF)$ (resp. $\msf{Ev}_\B$) such that all Player-$\A$
	(resp. Player-$\B$) positional strategies $\s_\A$ (resp. $\s_\B$) generated by $\msf{Ev}_\A$ (resp. $\msf{Ev}_\B$) are optimal in $\G$; and such strategies exist.
\end{theorem}
\begin{remark}
	\label{rmk:smaller_color_more_complex_env}
	Given some $u \in V_Q$, one can realize that the requirement at states $q,q' \in Q_u$ changes depending on the color of $q$ and $q'$. More specifically, if $\colFunc(q) < \colFunc(q')$, then the requirement at state $q$ is (a priori) stronger than the requirement at state $q'$ since the game form $\fsf(q)$ should behave well for environments of larger size than the game form $\fsf(q')$.
\end{remark}

The remainder of this section and the entirety of the next is devoted to the explanation of the construction of the environment function $\msf{Ev}_\A$ (the
construction being similar for Player $\B$), hence we are taking the
point of view of Player $\A$. First, let us argue that we can restrict ourselves to a specific $u$-slice $Q_u$ for some $u \in V_Q$. Such a restriction is properly defined (using stopping states) in Definition~\ref{def:game_restricted_value_slice} below. 
\begin{definition}[Game restricted to a $u$-slice]
	\label{def:game_restricted_value_slice}
	For all $u \in V_{Q}$, we let $\G^u$ be the concurrent game
	$\G$ where all states outside $Q_u$ are made stopping states:
	for every $q \in Q \setminus Q_u$, we set
	$\msf{val}(q) \gets \MarVal{\G}(q)$. 
	The states, game forms and coloring function on $Q_u$ are left unchanged.
\end{definition}

Interestingly, a Player-$\A$ positional strategy optimal in $\G$ can be obtained by merging appropriate positional strategies $\s_\A^u$ in the games $\G^u$ for all $u \in V_Q \setminus \{ 0 \}$, which is actually a straightforward consequence of Theorem~\ref{thm:strongly_domnates_imply_guarantee}.
\begin{lemma}[Proof~\ref{appen:proof_prop_optimal_in_all_value_slice_optimal_everywhere}]
	\label{lem:optimal_in_all_value_slice_optimal_everywhere}
	Consider a parity game $\G$. For all $u \in V_Q \setminus \{ 0 \}$, consider a positional Player-$\A$ strategy $\s_\A^u$ that parity dominates the valuation $\MarVal{\G}$ in $\G^u$. Then, the Player-$\A$ positional strategy $\s_\A$ such that $\s_\A(q) := \s_\A^u(q)$ for all $u \in V_Q \setminus \{ 0 \}$ and $q \in Q_u$ guarantees the valuation $\MarVal{\G}$ in $\G$ (i.e. it is optimal).
\end{lemma}

\section{The proof}
In Appendix~\ref{appen:property_of_the_update_function}, we give a quick overview of the technical properties and lemmas that we show in the appendix and that we will use to prove the lemmas (we will state later in this chapter) leading to the proof of Theorem~\ref{thm:transfer_local_global}.

In this section, we fix a parity game $\G = \Games{\Aconc}{\colFunc}$. In particular, the set of states $Q$ is fixed: r  ecall that $Q$ is also the set of outcomes of all game forms occurring in $\G$. Also, Lemma~\ref{lem:optimal_in_all_value_slice_optimal_everywhere} above justifies that we focus on a given $u$-slice $Q_u$ for some positive $u \in (0,1]$. We also let $e := e_u$, $o := o_u$ and $K := \{ k_i \mid i\in \brd \}$ and for all $n \in \brd$, we let $K^n := \{ k_i^n \mid i\in \brd \}$.

\subsection{Big picture of the proof}
\label{subsec:big_picture}

In order to give an idea of the steps that we take to prove Theorem~\ref{thm:transfer_local_global}, let us first consider the very simple case of finite turn-based deterministic reachability games. In this setting, computing the area $L_\A$ from which Player $\A$ has a winning strategy can be done inductively. That is, initially we set $L_\A := T$ where $T$ denotes the target that Player $\A$ wants to reach. Then, the inductive step is handled with a (deterministic) attractor: we add to $L_\A$ any Player-$\A$ state with a successor in $L_\A$ and any Player-$\B$ state with all successors in $L_\A$. After finitely many steps, there is no more state to add in $L_\A$: this exactly corresponds to the states from which Player $\A$ has a winning strategy.

Computing a single attractor is not merely enough to take into account the intricate behavior of parity objectives, which is what Theorem~\ref{thm:transfer_local_global} deals with
. Therefore, we are going to iteratively compute several layers of (virtual) colors, 
with a local update to change the (virtual) color (and therefore the layer it belongs to) of a state. This local update can be seen as an attractor except 
in a concurrent stochastic setting. 
Hence, when we update the (virtual) color of a state, we take into account the 
concurrent interaction of the players at each state along with the probability to see stopping states or states with different (virtual) colors. 
We define this local update in Subsection~\ref{subsec:local_operator}. Let us describe below the steps that we take to capture the behavior of the parity objective. 


We compute layers of successive probabilistic attractors with leaks towards the stopping states. Although we compute a strategy, e.g., for Player $\A$, we alternate players to build layers, then move the last non-empty layer into the closest layer with same parity, then backtrack the attractor computation from this layer downwards, and start over again the full attractor computation on the new layer structure. In a more concrete way, let us assume below that the highest color in the $u$-slice is $6$. We proceed as follows:
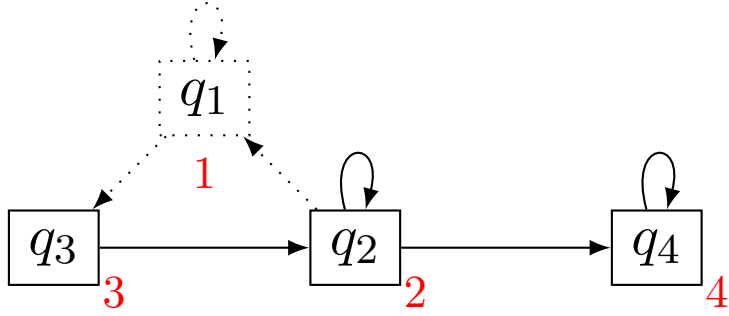
\begin{figure}
	\centering
	\resizebox{0.65\linewidth}{!}{\begin{tikzpicture}
	\node[draw=black,rectangle,scale=1] (q3) at (0*\lenV,0*\lenH) {$q_3$} ;
	\node[player] (c3) at (0.4*\lenV,-0.3*\lenH) {\color{red}3} ;
	\node[draw=black,rectangle,scale=1] (q2) at (2*\lenV,0*\lenH) {$q_2$} ;
	\node[player] (c2) at (2.4*\lenV,-0.3*\lenH) {\color{red}2} ;
	\node[draw=black,rectangle,scale=1] (q4) at (4*\lenV,0*\lenH) {$q_4$} ;
	\node[player] (c4) at (4.4*\lenV,-0.3*\lenH) {\color{red}4} ;
	\node[draw=black,rectangle,scale=1,dotted] (q1) at (1*\lenV,1*\lenH) {$q_1$} ;
	\node[player] (c1) at (1*\lenV,0.5*\lenH) {\color{red}1} ;
	
	\path[-latex]  	
	(q3) edge (q2)
	(q2) edge (q4)
	(q2) edge[loop above] (q2)
	(q4) edge[loop above] (q4)
	(q2) edge[dotted] (q1)
	(q1) edge[dotted] (q3)
	(q1) edge[loop above,dotted] (q1)
	
	;
\end{tikzpicture}}
	\caption{A (deterministic turn-based) game.}
	\label{fig:simpleGameThreeStatesIntro}  
\end{figure}
\begin{enumerate}
	\item Add the states colored with $6$ to layer $L_6$.
	
	\item Recursively add to $L_6$ the states where Player $\A$ can guarantee that with positive probability (pp) either a leak towards stopping states occurs now and its expected explicit value is $u$ or more ($\msf{Leak}_{\geq u}$), or with pp the next state is in $L_6$, i.e. with pp color $6$ or $\msf{Leak}_{\geq u}$ will occur.
	
	\item Add the remaining states colored with $5$ to layer $L_5$.
	
	\item Recursively add to $L_5$ the states where Player $\B$
	can guarantee that either $\msf{Leak}_{<u}$ occurs now with
	pp, or the next state is surely not in $L_6$ and with pp in
	$L_5$;  i.e. if $\msf{Leak}_{<u}$ occurs with probability $0$, 
	then color
	$6$ will not occur but color $5$ will eventually occur with pp. 
	
	\item Add the remaining states colored with $4$ to layer $L_4$.
	
	\item Recursively add to $L_4$ the states where Player $\A$ can guarantee that either $\msf{Leak}_{\geq u}$ will occur with pp, or the maximal layer index of the next states seen with pp is $4$ or $6$, i.e. even and at least $4$.
	
	\item And so on, from color $3$ to $0$. The layers so far only give information about what can happen at finite horizon.
	
	\item For instance, from $L_2$, Player $\A$ can guarantee that either $\msf{Leak}_{\geq u}$ will occur with pp, or the maximal color that will be seen with pp 
	is in $\{2,4,6\}$. Now, if e.g. $L_0 = L_1 = \emptyset$, we merge $L_2$ into $L_4$. This is, arguably, the most surprising step, let us try to give a conceptual intuition behind this step, we will then illustrate it on a concrete example. Consider what happens in the $L_2$ layer, assuming $L_0 = L_1 = \emptyset$. From states in that layer, either:
	\begin{itemize}
		\item $\msf{Leak}_{\geq u}$ occurs with pp; 
		\item otherwise, $L_3 \cup L_4 \cup L_5 \cup L_6$ occurs with pp, and the maximum index seen with pp is even;
		\item otherwise, the game loops surely in $L_2$, and with pp the maximum color seen is $2$. In that case, Player $\A$ wins the (real) parity game almost-surely.
	\end{itemize}
	In other words, either the game loops in $L_2$ and Player $\A$ wins, or what happens is alike to what happens in the layer $L_4$. In Figure~\ref{fig:simpleGameThreeStatesIntro} without dotted state $q_1$ where Player $\B$ plays alone, $L_4 = \{q_4\}$, then $L_3 = \{q_3\}$, then $L_2 = \{q_2\}$, but if the play stays in $L_2$, Player $\A$ wins, so Player $\B$ may just as well go to $q_4$ since there is no $L_0,L_1$ below to escape. 
	However, in Figure~\ref{fig:simpleGameThreeStatesIntro} with dotted state $q_1$, Player $\B$ can leave $L_2$ via $L_1 \neq \emptyset$, avoid color $4$, and win. There $L_2$ and $L_4$ are not alike. Hence, we do not merge them. 
	
	\item Earlier, some states of color $1$ may have been added to $L_3$ since $L_3$ "overruled" $L_2$ w.r.t indices, but since $L_2$ has just been merged into $L_4$, it now overrules $L_3$, so some states from $L_3$ may have to go back to $L_1$. Therefore we reset the layers below $L_4$ and repeat the above attractor alternation 
	all over again, until all the states are eventually in $L_6$ as we shall prove.  
\end{enumerate}
The key property that is growing throughout the above computation and will hold in the final $L_6$ involves layer games: the $L_n$-game is derived from the $u$-slice by abstracting each $L_i$ with $i \neq n$ via one state $k_i^n$ from which the player who dislikes the parity of $n$ chooses any next state in $L_n$, making it harder for the other to win. If $i > n$ then $k_i^n$ is $i$-colored, else $(n-1)$-colored, also making it harder for the other to win. And states in $L_n$ bear their true colors. See for instance Figure~\ref{fig:extracted_game}. The $L_n$-game is only seemingly harder to win: it is actually equivalently hard, but its useful properties are easier to prove.

The key growing property is as follows: between two merges
, the attractor computation from the top layer down to $L_n$ ensures that Player $\A$ has a positional strategy of value at least $u$ in each $L_i$ for even $i \geq n$, and Player $\B$ less than $u$ for odd $i \geq n$. In the very end, there is only one even layer with all states bearing their true colors, and no abstract states, so the last layer game equals the $u$-slice game, for which we have thus computed a positional optimal strategy.

Let us hint at how to show positional optimality in the $L_n$-games when it holds: we break $L_n$ each into one simple parity game built on $\fsf(q)$ per state $q$ in $L_n$, abstracting the other states in $L_n$ into one. Our theorem assumption yields an optimal $\GF$-strategy for Player $\A$ or $\B$ in the simple parity game. Gluing them does the job. 

Now recall the above example of computation, where Player $\A$ for even indices, and Player $\B$ for odd ones, guarantees (assuming neither $\msf{Leak}_{\geq u}$ nor $\msf{Leak}_{< u}$ occurs): 
\begin{itemize}
	\item in $L_6$, that a next state is of color 6 with pp;
	\item in $L_5$, that the next state is surely not of color 6, and with pp a next state is of color 5;
	\item in $L_4$, that a next state is of color 6 with pp or surely the next state is not of color 5, and with pp a next state is of color 4;
	\item in $L_3$, that the next state is surely not of color 6 and that the next state is of color 5 with pp or surely the next state is not of color 4, and with pp a next state is of color 3; etc.
\end{itemize}
One can realize that the furthest the layer is from the maximal color (that is 6), the more complex the requirement is at that layer. That is why the strength of our assumption on the game form induced at some state increases with the difference between maximal true color in the $u$-slice and the true color (at most $n$) of the state, as stated in Remark~\ref{rmk:smaller_color_more_complex_env}. 
We will discuss this further in the next chapter. In particular, we will show that there is an infinite strict hierarchy between these assumptions. 


\subsection{Extracting an environment function from a parity game}
Once restricted to the game $\G^u$, the method that we use to prove Theorem~\ref{thm:transfer_local_global}
consists in iteratively building a pair of (virtual) coloring
and environment functions ensuring a nice property (namely
\emph{faithfulness}, defined below in
Definition~\ref{def:trithful_env_coloring_func}). For the remainder of this section, we illustrate the definitions and lemmas on the game depicted in Figures~\ref{fig:InitialGame} and~\ref{fig:SomeStepsGame}.
\begin{figure*}
	\hspace*{-0.5cm}
	\begin{minipage}{0.48\linewidth}
		\hspace*{-1cm}
		\centering
		\includegraphics[scale=0.4]{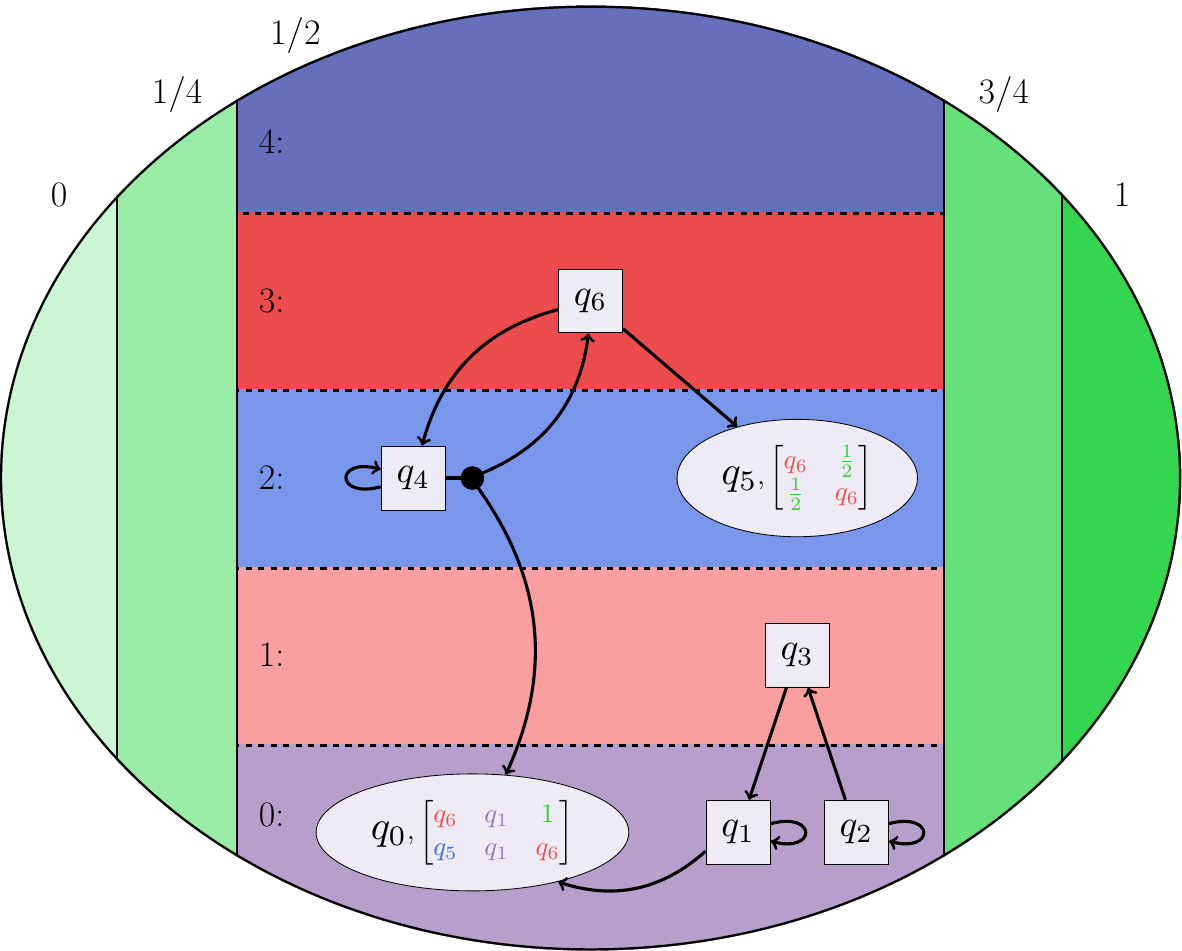}
		\caption{The depiction of a game restricted to the $1/2$-slice $Q_{1/2}$ with the initial coloring function $\colFunc$.}
		\label{fig:InitialGame}  
	\end{minipage}
	\hspace{1cm}
	\begin{minipage}{0.48\linewidth}
		\centering
		\includegraphics[scale=0.4]{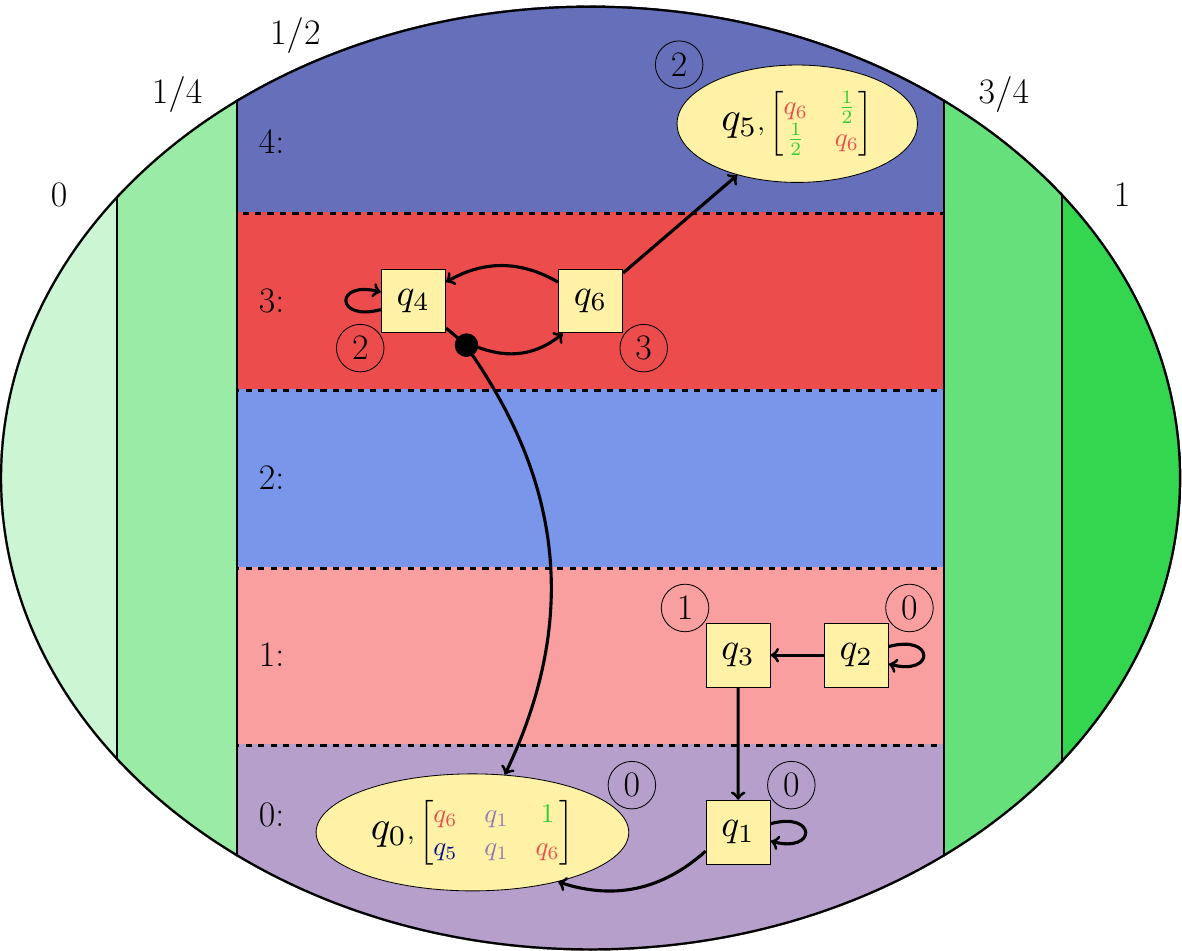}
		\caption{The same game restricted to $Q_{1/2}$ with a different coloring function $\colVirt$.}
		\label{fig:SomeStepsGame}  
	\end{minipage}
\end{figure*}

\begin{example}
	We explain the notations used to depict this game (it is in
	fact the same arena in both Figures~\ref{fig:InitialGame}
	and~\ref{fig:SomeStepsGame}, 
	with different coloring functions -- real or virtual). 
	On the sides in green are the slices $Q_0,Q_{1/4},Q_{3/4}$ and $Q_1$ from left to right. We focus on the central slice $Q_{1/2}$. In $Q_{1/2}$, there are seven states, five of which (the square-shaped ones) are turn-based for Player $\B$, 
	that is, Player $\A$ has only one available action. On the other hand, the two circled-shaped states $\leta$ and $\letf$ are \textquotedblleft truly\textquotedblright{} concurrent in the sense that both players have several actions available. Furthermore, note that only one stochastic outcome (of a local interaction) is drafted as a black dot: from $\lete$, Player $\B$ may either loop on $\lete$ or go with equal probability to $\leta$ and $\letg$. 
	The other arrows lead to a single state and the outcomes of the game forms in $\leta$ or $\letf$ is a single state or a value: $1$ or $1/2$. These formally 
	refer to a (distribution over) stopping states outside of the $1/2$-slice $Q_{1/2}$. 
	The horizontal layers depict the colors of the states. In Figure~\ref{fig:InitialGame}, the coloring function considered is the initial one $\colFunc$ whereas in Figure~\ref{fig:SomeStepsGame} we have depicted a (virtual) coloring function $\colVirt$. For instance, $\colFunc(\letg) = 3$ whereas $\colFunc(\letf) = 2$. Similarly, $\colVirt(\letg) = 3$ whereas $\colVirt(\letf) = 4$. Note that, in Figure~\ref{fig:SomeStepsGame}, the real colors (given by $\colFunc$) are reminded next to some states with circled numbers. 
	Finally, note that $e := e_{1/2} = 4$. 
\end{example}

Given a (virtual) coloring function, we need to 
extract local environments from the parity game $\G$, which 
summarize how the Players see their neighboring states via the
virtual coloring function.
This is (partly) done in Definition~\ref{def:probability_extracted_parity_game}.
\begin{definition}[Probability function extracted from an arena and a
	(virtual) coloring function]
	\label{def:probability_extracted_parity_game}
	Consider 
	a virtual coloring function $\colVirt: Q_u \rightarrow \brd$. 
	For all $n \in \brd$, we let $Q_n := \colVirt^{-1}[n]$. We then define the map $p_{n,\colVirt}: Q \rightarrow Q_n \uplus K^n \uplus V_{Q \setminus Q_u}$ such that, for all $q \in Q$:
	\begin{itemize}
		\item if $q \in Q_n$, $p_{n,\colVirt}(q) := q \in Q$;
		\item if $q \in Q_u \setminus Q_n$, 
		$p_{n,\colVirt}(q) := k^n_{\colVirt(q)} \in K^n$;
		\item if $q \in Q \setminus Q_u$, $p_{n,\colVirt}(q) := 
		\MarVal{\G}(q) \in V_{Q \setminus Q_u}$.
	\end{itemize}
\end{definition}


Given a (virtual) coloring function $\colVirt: Q_u \rightarrow \brd$  and a color $n \in \brd$, we can now extract a smaller parity game from $\G$ where the states with non-trivial game forms are the states in $\colVirt^{-1}[n]$, the states in $Q \setminus Q_u$ are stopping states and 
the arena loops back to $\colVirt^{-1}[n]$ when a state in $Q_u
\setminus \colVirt^{-1}[n]$ is seen. This is done below in Definition~\ref{def:parity_game_from_set_of_gfs}.
\begin{definition}[Parity game extracted from the $u$-slice
	]
	\label{def:parity_game_from_set_of_gfs}
	Consider a virtual coloring function $\colVirt: Q_u \rightarrow \brd$ and a color $n \in \brd$. Let $\msf{C} \in \{ \A,\B \}$ be a Player: $\A$ if $n$ is odd and $\B$ if $n$ is even. The arena $\Aconc_{\colVirt}^{n} = \langle Q',\fsf' \rangle$ along with the coloring function $\colVirt_n: Q' \rightarrow \N$ are such that, denoting $Q_n := \colVirt^{-1}[n]$:
	\begin{itemize}
		\item $Q' := Q_n \cup K^n \cup V_{Q \setminus Q_u}$ where 
		all $x \in V_{Q \setminus Q_u}$ are stopping states with $\msf{val}(x) \gets x$;
		\item for all $q \in Q_n$, $\fsf'(q) := \fsf(q)^{p_{n,\colVirt}} = \langle \Act_\A^q,\Act_\B^q,Q',\outCNF_{p_{n,\colVirt}} \rangle$ where, for all $\sigma_\A \in \Act_\A^q$, $\sigma_\B \in \Act_\B^q$, and $q \in Q$, we have $\outCNF_{p_{n,\colVirt}}(\sigma_\A,\sigma_\B)(q) := \outCNF(\sigma_\A,\sigma_\B)[p_{n,\colVirt}^{-1}[q]]$;
		\item for all $k \in K^n$, we set $\fsf'(k)$ as a Player-$\C$ state whose outcomes are all the states in $Q_n$;
		\item for all $q \in Q_n$, we let $\colVirt_n(q) := \colFunc(q)$ and for all $i \in \brd$, we have $\colVirt_n(k_i^n) := \max(i,n-1)$.
	\end{itemize}
	For $t \in [0,1]$, we define the valuation $v_{n,\colVirt}^{t}: Q' \rightarrow [0,1]$: $v_{n,\colVirt}^{t}[Q_n \cup K^n] := \{ t \}$ and for all $x \in V_{Q \setminus Q_u}$, $v_{n,\colVirt}^{t}(x) := x$.
	
	The game $\Zone_{\colVirt}^{n}$ is then equal to $\Zone_{\colVirt}^{n} = \Games{\Aconc^{n}_{\colVirt}}{\colVirt_n}$
	.
\end{definition}

\begin{remark}
	\label{rmk:explain_game_extracted_layer}
	First, the notation $\Zone^n_{\colVirt}$ comes from the fact that
	the game is extracted for the $n$-colored layer w.r.t. the coloring
	function $\colVirt$. The idea behind
	Definition~\ref{def:parity_game_from_set_of_gfs} is the following:
	the states of interest are those of $Q_n$, that is, those for which
	the virtual color given by $\colVirt$ is $n$. Note however that the
	colors of these states in $\Zone^n_{\colVirt}$ are given by the real
	coloring function $\colFunc$. On the other hand, for all $i \in \brd$, the state $k_i^n$ in $\Zone^n_{\colVirt}$ corresponds to the states in $\G^u$ colored with $i$ w.r.t. $\colVirt$. 
	In the case where $n$ is even, as formally defined later in Definition~\ref{def:env_guarantee_value_and_color}, we will require that any Player-$\A$ positional strategy generated by a given environment has value at least $u$, in the game $\Zone^n_{\colVirt}$, from all states in $Q_n$. However, all states $k^n_i$ for $i \in \brd$ are Player-$\B$'s, who can then choose to loop back to any state in $Q_n$. Therefore, given a Player-$\A$ positional strategy $\s_\A$, if the game cannot exit to any stopping state, for the strategy $\s_\A$ not to have value 0, the game may loop on some $k_i^n$ only at the condition that the highest color seen with positive probability is even.
	In addition, note that the color of the state $k^n_i$ for $i \in \brck{0,n-1}$ is $n-1$ (which is odd). Hence, all other things being equal, the game is harder for Player $\A$ when $n = 4$ than when $n=2$ or $0$. Finally, note that, when $n$ is odd, we will require that any Player-$\B$ positional strategy generated by a given environment has value less than $u$
	.
\end{remark}

\begin{figure}
	\centering
	\resizebox{0.4\linewidth}{!}{\begin{tikzpicture}
	\node[player] (c3) at (0*\lenV,2*\lenH) {\color{red}3} ;
	\node[player2',scale=0.8] (k3) at (0*\lenV,1.25*\lenH) {$k_3^3$} ;
	\node[player] (c4) at (2*\lenV,2*\lenH) {\color{red}4} ;
	\node[player2',scale=0.8] (k4) at (2*\lenV,1.25*\lenH) {$k_4^3$} ;
	\node[draw=black,rectangle,scale=1,minimum size=6mm] (q4) at (0*\lenV,0*\lenH) {$q_{4}$} ;
	\node[player] (c2) at (-0.5*\lenV,0.5*\lenH) {\color{red}2} ;
	\node[fill=black,draw=black,circle,scale=0.1,minimum size=6mm] (nextq4) at (0*\lenV,-0.6*\lenH) {};
	\node[draw=black,rectangle,scale=1,minimum size=6mm] (q6) at (2*\lenV,0*\lenH) {$q_{6}$} ;
	\node[player] (c2) at (2.5*\lenV,0.5*\lenH) {\color{red}3} ;
	\node[player] (c0) at (-0.1*\lenV,-1.35*\lenH) {\color{red}2} ;
	\node[player2',scale=0.8] (k0) at (-0.5*\lenV,-1*\lenH) {$k_0^3$} ;
	\node[player] (c1) at (1.4*\lenV,-1.35*\lenH) {\color{red}2} ;
	\node[player2',scale=0.8] (k1) at (1*\lenV,-1*\lenH) {$k_1^3$} ;
	\node[player] (c2) at (2.9*\lenV,-1.35*\lenH) {\color{red}2} ;
	\node[player2',scale=0.8] (k2) at (2.5*\lenV,-1*\lenH) {$k_2^3$} ;
	
	\path[-latex]  
	(q6) edge (k4)
	(q6) edge[bend right] (q4)
	(q4) edge[-] (nextq4)
	(q4) edge[loop left] (q4)
	(nextq4) edge (k0)
	(nextq4) edge[bend left=10] (q6)
	;
\end{tikzpicture}}
	\caption{The game $\Zone^{3}_{\colVirt}$. For
		readability, exiting arrows from
		$k_0^3,k_1^3,k_2^3,k_3^3$ and $k_4^3$ are not
		depicted: they would all loop back to both $\lete$
		and $\letg$.}
	\label{fig:extracted_game}
\end{figure}
\begin{figure}
	\centering
	\resizebox{0.6\linewidth}{!}{\begin{tikzpicture}
	\node[player] (cs0) at (5.2*\lenV,1.4*\lenH) {\color{red}0} ;
	\node[player1ell] (s0) at (3*\lenV,2*\lenH) {$q_{\msf{init}}$,{\Large$\begin{bmatrix}
			k_3 & k_0 & 1\\
			k_4 & k_0 & k_3
			\end{bmatrix}$}} ;
	\node[player] (c0) at (0.5*\lenV,3.7*\lenH) {\color{red}0} ;
	\node[player1'] (k0) at (0.5*\lenV,3*\lenH) {$k_0$} ;
	\node[player] (c1) at (1.6*\lenV,4.2*\lenH) {\color{red}1} ;
	\node[player1'] (k1) at (1.6*\lenV,3.5*\lenH) {$k_1$} ;
	\node[player] (c2) at (3*\lenV,4.5*\lenH) {\color{red}2} ;
	\node[player1'] (k2) at (3*\lenV,3.8*\lenH) {$k_2$} ;
	\node[player] (c3) at (4.3*\lenV,4.2*\lenH) {\color{red}3} ;
	\node[player1'] (k3) at (4.3*\lenV,3.5*\lenH) {$k_3$} ;
	\node[player] (c4) at (5.5*\lenV,3.7*\lenH) {\color{red}4} ;
	\node[player1'] (k4) at (5.5*\lenV,3*\lenH) {$k_4$} ;
	\node[player1',dashed] (s2) at (6*\lenV,2*\lenH) {\color{black}$1$} ;
	
	\path[-latex]  	
	(s0) edge (k0)
	(k0) edge[bend right=10] (s0)
	(s0) edge[bend right=10] (k3)
	(k3) edge[bend right=10] (s0)
	(s0) edge[bend right=10] (k4)
	(k4) edge[bend right=0] (s0)
	(s0) edge (s2)
	;
\end{tikzpicture}}
	\caption{The game $\G_{q_0,\colVirt}^0$ with $\colVirt$ the coloring function depicted in Figure~\ref{fig:SomeStepsGame}.}
	\label{fig:extracted_environment}
\end{figure}

\begin{example}
	\label{ex:extracted_layer}
	The game $\Zone^{3}_{\colVirt}$ is partly depicted in 
	Figure~\ref{fig:extracted_game} (the virtual coloring function
	$\colVirt$ being the one depicted in
	Figure~\ref{fig:SomeStepsGame}). The colors of the states are
	depicted in red
	. 
	Although the arrows are not depicted, from all states $k_0^3,k_1^3,k_2^3,k_3^3$ and $k_4^3$ Player $\A$ can decide to which state among $\{ q_4,q_6 \}$ to loop back (since $n = 3$ is odd). 
	In an even-colored layer, it would have been Player $\B$ to decide.
\end{example}


Given a virtual coloring function, we also associate a local environment with each state.

\begin{definition}[Local environment induced by a virtual coloring function
	and a color]
	\label{def:local_env_induced_coloring_func}
	Consider a state $q \in Q_u$ and a coloring function
	$\colVirt: Q_u \rightarrow \llbracket 0,e \rrbracket$. We define $p_{q,\colVirt}: Q \rightarrow \{ \qinit \} \cup K_e \cup [0,1]$ similarly to how we define $p_{n,\colVirt}$ in Definition~\ref{def:probability_extracted_parity_game}. That is, for all $q' \in Q$, we have: 
	\begin{itemize}
		\item if $q' = q$, $p_{q,\colVirt} := \qinit$;
		\item if $q' \in Q_u \setminus \{ q \}$, $p_{q,\colVirt} := k_{\colVirt(q')} \in K_e$;
		\item if $q' \in Q \setminus Q_u$, $p_{q,\colVirt}(q) := \MarVal{\G}(q) \in [0,1]$.
	\end{itemize}
	
	Then, for all $n \in \llbracket 0,e \rrbracket$, the environment
	$E^n_{q,\colVirt}$ associated with state $q$ w.r.t. $\colVirt$
	and $n$ is such that
	$E^n_{q,\colVirt} := \langle \max(c_n,\colVirt(q)),e 
	,p_{q,\colVirt} \rangle$ where $c_n = n+1$ if $n$ is odd
	and $c_n := n-1$ if $n$ is even. We say that the coloring function
	$\colVirt$ is \emph{associated} with the environment
	$E^n_{q,\colVirt}$. 
	
	The corresponding (local) game
	$\G_{(\fsf(q),E^n_{q,\colVirt})}$ (see
	Definition~\ref{def:parity_game_from_gf}) is denoted
	$\G_{q,\colVirt}^n$. For all $x \in [0,1]$, we set
	$v_{q,\colVirt}^x := v^x_{(\fsf(q),E^0_{q,\colVirt})}$ (see
	Definition~\ref{def:parity_game_from_gf}).
\end{definition}
The definition of $c_n$ may seem ad hoc. We give an explanation of this definition below in the next subsection.
\begin{example}
	\label{ex:local_env}
	The game $\G^n_{q_5,\colFunc}$ is depicted on the right of	 Figure~\ref{fig:GFEnv} for $n=0,1,2$. However, 
	if $n=3$, 
	the color of $q_\msf{init}$ would be 4, and if $n=4$, it would be 3
	. The game $\G^n_{q_0,\colVirt}$ is depicted in Figure~\ref{fig:extracted_environment} for $n=0$. However, if $n=1$, the color of $q_\msf{init}$ would be 2, if $n=2$, the color would be 1, if $n=3$, the color would be 4 and if $n=4$ the color would be 3.
\end{example}

\subsection{Local Operator}
\label{subsec:local_operator}
We want to define a way to update a (virtual) coloring function
$\colVirt$. This will be done via a local operator mapping a given
state $q$ to the best color $k$ for which Player $\A$ can achieve the
value $u$ in the corresponding local parity game
$\G^k_{q,\colVirt}$. Note that \textquotedblleft
best\textquotedblright{} is to be understood considering an ordering
compatible with the parity objective. Specifically, taking the
point-of-view of Player $\A$, any even number is better than any odd
number, and when they increase, odd numbers get worse whereas 
even numbers get better
. This induces a new ordering
.
\begin{definition}[Parity order]
	We define a total strict order relation $\prec_{\msf{par}}$ on $\N$ such that, for all $m,n \in \N$, we have $m \prec_{\msf{par}} n$ if: $m$ is odd and $n$ is even; or $m > n$ and $m$ and $n$ are odd; or $m < n$ and $m$ and $n$ are even.
	
\end{definition}

\begin{definition}[Local operator]
	\label{def:local_operator}
	Consider a state $q \in Q_u$ and a (possibly virtual) coloring
	function $\colVirt: Q_u \rightarrow \brd$. The color
	$\newC(q,\colVirt) \in \N$ induced by $\colVirt$ at $q$ is defined by:
	\begin{displaymath}
	\newC(q,\colVirt) := \max_{\prec_{\msf{par}}} \{ n \in \brd \mid \MarVal{\G_{q,\colVirt}^n}(q_{\msf{init}}) \geq u \}
	\end{displaymath}
\end{definition}
The meaning of a new virtual color $n$ assigned to a state $q$ via
$\newC$ is the following: in the game
$\G^u$ with the coloring function $\colVirt$, from state $q$ and in at most one
step, the highest color w.r.t. $\colVirt$ seen with positive
probability when both players play optimally is $n$ (and no
stopping state is seen).
%

\label{eplain:c_n}
Let us now explain the choice of $c_n$ in
Definition~\ref{def:local_env_induced_coloring_func}. In a local
environment parameterized by $n$, the integer $n$ induces a shifted
parity objective for Player $\A$: her objective is that the maximal
color seen infinitely often is at least $n$
w.r.t. $\prec_{\msf{par}}$; in particular $n=0$ induces the usual
parity objective. The value $c_n$ encodes that winning
condition. For instance, if $n=2$, assuming $\colVirt(q)=0$ for
simplicity, then $c_n = 1$, which implies that seeing $0$ infinitely
often is not enough, but seeing $2$ infinitely often is enough to
win. Similarly, if $n = 1$, $c_n = 2$, which implies that seeing $1$ infinitely
often is now enough to win, but seeing $3$ infinitely often is still losing.

\begin{remark}
	\label{rmk:explain_new_col}
	Assume for instance that $\newC(q,\colVirt) = 2$ for some state $q \in Q$ and some (virtual) coloring function $\colVirt$\footnote{In
		particular, it must be the case that $\colVirt(q) \le 2$, see
		Proposition~\ref{prop:update_color_at_least_preivous_color}.}. In particular, Player $\A$ has a $\GF$-strategy $\sigma_\A$ optimal w.r.t. $(\formNF,E^2_{q,\colVirt})$, which yields a value at least $u$, and Player $\B$ has a $\GF$-strategy
	$\sigma_\B$ optimal w.r.t. $(\formNF,E^4_{q,\colVirt})$, which yields
	a value less than $u$. Remember that all games
	$\G^k_{q,\colVirt}$ (where $k$ ranges over $\brd$) share the same
	structure and that only the color of $q_\msf{init}$ changes.
	Consider what happens in the game $\G^0_{q,\colVirt}$ (where the state
	$q_\msf{init}$ is colored by $\colVirt(q)$) when playing strategies
	$\sigma_\A$ and $\sigma_\B$. Recalling
	Remark~\ref{rmk:explanation_optim_gf_strat}, the game cannot exit to
	any stopping state since the expected value of the stopping states reached
	would be both at least $u$ (since $\sigma_\A$ is optimal in
	$\G^2_{q,\colVirt}$) and less than $u$ (since $\sigma_\B$ is
	optimal in $\G^4_{q,\colVirt}$); hence this cannot
	happen. Furthermore, if some color of value at least
	$3$ is seen with positive probability, then the highest such color
	must be both even (since $\sigma_\A$ is optimal in
	$\G^2_{q,\colVirt}$) and odd (since $\sigma_\B$ is optimal in
	$\G^4_{q,\colVirt}$). In fact, the highest color seen with positive
	probability in $\G^0_{q,\colVirt}$ under $\sigma_\A$ and $\sigma_\B$ is $2$ and neither of the players can do better
	(w.r.t. the ordering
	$\prec_\msf{par}$). 
	From this, we can infer the
	semantics 
	of a virtual color $n$ assigned to a state $q$ via operator $\newC$
	(i.e. a color given by a virtual coloring function): in the game
	$\G^u$ with coloring function $\colVirt$, from state $q$ and in at most one
	step, the highest color w.r.t. $\colVirt$ 
	seen with positive probability is $n$ 
	(and no stopping state can be seen).
	%
\end{remark}

Let us exemplify this operator on an example.
\begin{example}
	\label{ex:compute_new_col}
	First, consider Figure~\ref{fig:GFEnv} and let us compute $\newC(q_5,\colFunc)$. We can realize that, regardless of the color of state $q_\msf{init}$, Player $\A$ can (positionally) play both rows with positive probability and ensure reaching (almost-surely) the stopping state $1/2$. In fact, for all $n \in \brck{0,4}$, we have $\MarVal{\G_{q_5,\colFunc}^n}(q_{\msf{init}}) = 1/2$. Hence, $\newC(q_5,\colFunc) = 4$.		
	
	Consider now Figure~\ref{fig:extracted_environment} and let us compute $\newC(q_0,\colVirt)$. As mentioned in Example~\ref{ex:local_env}, the game $\G^4_{q_0,\colVirt}$ corresponds to the game depicted in Figure~\ref{fig:extracted_environment} except that $q_\msf{init}$ is colored with 3. One can realize that, with this choice (of coloring of the state $q_\mathsf{init}$), if the highest color $i \in \brck{0,4}$ such that $k_i$ is seen infinitely often is such that $i \prec_\msf{par} 4$, then Player $\A$ loses. The value of this game is 0 as Player $\B$ can ensure looping on $k_0$ and $q_\msf{init}$ (by playing, positionally and deterministically, the middle column) thus ensuring that the highest color seen infinitely often is 3. Thus, $\newC(q_0,\colVirt) \prec_\msf{par} 4$. In the game $\G^2_{q_0,\colVirt}$, $q_\msf{init}$ is colored with 1. Again, with this choice (of coloring of the state $q_\mathsf{init}$), if the highest color $i \in \brck{0,4}$ such that $k_i$ is seen infinitely often is such that $i \prec_\msf{par} 2$, then Player $\A$ loses. The value of this game is also 0 as Player $\B$ can still play the middle column ensuring that the highest color seen infinitely often is 1. Thus, $\newC(q_0,\colVirt) \prec_\msf{par} 2$. Consider now the game $\G^0_{q_0,\colVirt}$, the one depicted in Figure~\ref{fig:extracted_environment}. The value of the state $q_\msf{init}$ is now 1. Indeed, if Player $\A$ plays the two rows with equal probability, one can see that this strategy parity dominates (see Definition~\ref{def:dominate_strongly_dominate_gurantee}) the valuation $v_{q_0,\colVirt}^1$ (recall Definition~\ref{def:local_env_induced_coloring_func})
	. Indeed, the BSCCs compatible with this strategy are $\{ q_\msf{init},k_3,k_4 \}$ and $\{ q_\msf{init},k_0 \}$ 
	and they are even-colored. Hence, by Theorem~\ref{thm:strongly_domnates_imply_guarantee}, $\MarVal{\G_{q_1,\colFunc}^0}(q_{\msf{init}}) = 1 \geq 1/2$ and $\newC(q_0,\colVirt) \succeq_\msf{par} 0$. That is, $\newC(q_0,\colVirt) = 0$.
\end{example}

Some of the properties enjoyed by the local operator $\newC$ are given in the appendix in Page~\pageref{subsubsec:local_operator}.

\subsection{Faithful coloring function}
\label{subsec:faithful_coloring_function}
To prove Theorem~\ref{thm:transfer_local_global}, we iteratively build
a (virtual) coloring function and a local
environment. 
We want to define the desirable property that the pair of coloring and
environment functions should satisfy that will be preserved step by
step.
%
First, we need to define the notion of an environment function witnessing a color.
\begin{definition}[Environment witnessing a color]
	\label{def:env_guarantee_value_and_color}
	Consider a coloring function $\colVirt: Q_u \rightarrow
	\brd$, a color $n \in \brd$ and an environment function
	$\msf{Ev}: Q_n \rightarrow \msf{Env}(Q)$ with
	$Q_n := \colVirt^{-1}[n]$. 
	
	Assume that $n$ is even. We say that the pair
	$( \colVirt,\msf{Ev})$
	\emph{witnesses} 
	the color
	$n$ 
	if for all $q \in Q_n$, $\msf{Sz}_\A(\msf{Ev}(q)) \leq e - \colFunc(q)$
	and all positional Player-$\A$ strategies $\s_\A \in \msf{S}_\A^\Aconc$ 
	generated by $\msf{Ev}$ (recall
	Definition~\ref{def:strat_induced_by_environment}) in the game
	$\Zone^{n}_{\colVirt}$ 
	parity dominate the valuation $v_{n,\colVirt}^u$ (recall
	Definition~\ref{def:parity_game_from_set_of_gfs}). 
	
	Assume that $n$ is odd. We say that the pair $( \colVirt,\msf{Ev})$ witnesses the color $n$, if for all $q \in Q_n$, $\msf{Sz}_\B(\msf{Ev}(q)) \leq o - \colFunc(q)$ and for all positional Player-$\B$ strategies $\s_\B \in \msf{S}_\B^\Aconc$ generated by $\msf{Ev}$ in the game $\Zone^{n}_{\colVirt}$, there is some $u' < u$, such that $\s_\B$ parity dominates the valuation
	$v_{n,\colVirt}^{u'}$.
	%
\end{definition}
\begin{remark}
	The condition on the size of the environments considered along with the assumptions of Theorem~\ref{thm:transfer_local_global} ensures that the quantification over 
	the strategies generated by the environment function 
	are not over the empty set.
	
	This definition was hinted in
	Remark~\ref{rmk:explain_game_extracted_layer}. Informally, it
	means that, in the (virtual) games given by $\colVirt$, in the
	even-colored layers, Player $\A$ can achieve at least what she
	should be able to achieve in this $u$-slice (i.e. the value of the states is at least $u$). Whereas, in the odd-colored layers, Player $\B$ can prevent Player $\A$ from achieving this.
\end{remark}

We can now define the notion of faithful pair of coloring and environment functions. 
\begin{definition}[Faithful pair of coloring and environment functions]
	\label{def:trithful_env_coloring_func}
	Consider a coloring function $\colVirt: Q_u \rightarrow \brd$, some $n \in \brck{0,e+1}$ and a partial environment function $\msf{Ev}: Q_u \rightarrow \msf{Env}(Q)$ defined on $\colVirt^{-1}[\brck{n,e}]$.
	We say that $(\colVirt,\msf{Ev})$ is \emph{faithful down to} $n$ if:
	\begin{itemize}
		\item for all $k \in \brck{n,e}$, the pair $(\colVirt,\msf{Ev})$ witnesses color $k$
		;
		\item for all $q \in Q_u$, if $\colVirt(q) < n$, then
		$\colFunc(q) = \colVirt(q)$ and $\newC(q,\colVirt) < n$;
	\end{itemize}
	If $n = 0$, we say that the pair $(\colVirt,\msf{Ev})$ is \emph{completely faithful}.
\end{definition}

Only the first condition for faithfulness is really of interest to
us. For instance, this first condition 
suffices to show the crucial 
proposition below. However, the second condition is used in the
proofs. (It is also helpful as it guides us in how to build a
completely faithful pair, as discussed below.) Note that, in the
proofs, we use an even stronger notion of faithfulness with a third
condition. However, we do not present it here in order not to
complexify too much the approach, and the two first conditions are
sufficient for the 
(informal) explanation of Theorem~\ref{thm:transfer_local_global}. It
can however be found in Appendix~\ref{subsubsec:faithfulness}.

The benefit of faithful environments and coloring functions lies in
the proposition below: if all states are mapped w.r.t. the coloring
function to $e$
, then the environment function guarantees the value $u$ in the whole $u$-slice
$Q_u$. 
\begin{proposition}
	\label{lem:completely_faithful_all_even_ok}
	For a coloring function
	$\colVirt: Q_u \rightarrow \llbracket 0,e \rrbracket$ and an
	environment function $\msf{Ev}: Q_u \rightarrow \msf{Env}(Q)$, assume that $(\colVirt,\msf{Ev})$ is completely faithful and that
	$\colFunc_u[Q_u] = \{ e \}$. Then, all Player-$\A$ positional
	strategies generated by the environment function $\msf{Ev}$ parity
	dominate the valuation $\MarVal{\G}
	$ in the game
	$\G^u$
	.
\end{proposition}
\begin{proof*}
	This is direct from the definitions. Indeed, as $(\colVirt,\msf{Ev})$ is completely faithful, it follows that $(\colVirt,\msf{Ev})$ witnesses the color $e$ (see Definition~\ref{def:trithful_env_coloring_func}). 
	That is, all Player-$\A$ positional strategies $\s_\A$ generated by $\msf{Ev}$ in the game $\Zone^{e}_{\colVirt}$ parity dominate the valuation $v_{e,\colVirt}^u$ (see Definition~\ref{def:env_guarantee_value_and_color}). Since $\colVirt[Q_u] = \{ e \}$, 
	both games $\Zone^{e}_{\colVirt}$ and $\G^u$ are identical (see Definitions~\ref{def:game_restricted_value_slice} and \ref{def:parity_game_from_set_of_gfs})
	. Similarly, the valuation $v_{e,\colVirt}^u$ is equal to the valuation $\MarVal{\G}$ in the game $\G^u$ (also see Definition~\ref{def:env_guarantee_value_and_color}). 
\end{proof*}

\subsection{Computing a completely faithful pair}
\label{subsec:computing_faithful}
Given Lemma~\ref{lem:completely_faithful_all_even_ok}, our goal
is to come up with a pair of an environment function and a coloring
function completely faithful such that all states are colored with
$e$.
Let us first consider how to obtain a completely faithful pair from
the initial coloring function and the empty environment function
(i.e. no state is mapped to an environment). Note that this initial
pair of coloring and environment functions is faithful down to
$e+1$. Hence, our goal is, given a pair $(\colVirt,\msf{Ev})$
faithful down to some $n \in \brck{1,e+1}$, to build a new pair that
is faithful down to $n-1$. To do so, let us be guided by the second
property for faithfulness: to be faithful down to $n-1$, no state
$q \in Q_u$ such that $\colVirt(q) \leq n-2$ should be such that
$\newC(q,\colVirt) = n-1$. Hence, the idea is, for all such states
$q \in Q_u$, to change their colors to $n-1$ until no state
$q \in Q_u$ with $\colVirt(q) \leq n-2$ satisfies\footnote{This can seen as computing the \textquotedblleft probabilistic attractor with leaks towards the stopping states\textquotedblright{} that we mentioned above.}
$\newC(q,\colVirt) = n-1$. The environment associated to each such
state $q$ newly colored by $n-1$ will be given by the coloring
function $\colVirt$ for which $\newC(q,\colVirt) = n-1$ for the
first time (crucially, this is done before the color of $q$ is updated
to
$n-1$). 
The procedure we have described is formally given in the
Appendix as Algorithm~\ref{fig:UpdateColEnv}.
Interestingly, the update done in the algorithm
preserves the faithfulness of environment and coloring functions.
\begin{lemma}[Proof Appendix~\ref{subsubsec:proof_lem_algo_preserves_faithfulness}, Page~\pageref{subsubsec:proof_lem_algo_preserves_faithfulness}]
	\label{lem:algo_preserves_faithfulness}
	Consider a coloring function
	$\colVirt: Q_u \rightarrow \llbracket 0,e \rrbracket$,
	$n \in \brck{1,e+1}$, and a partial environment function
	$\msf{Ev}: Q_u \rightarrow \msf{Env}(Q)$ defined on
	$\colVirt^{-1}[\brck{n,e}]$. Assume that $(\colVirt,\msf{Ev})$ is
	faithful down to $n$. Let
	$(\colVirt',\msf{Ev}') \gets \msf{UpdateColEnv}(n-1,
	\colVirt,\msf{Ev})$ be the pair computed by
	Algorithm~\ref{fig:UpdateColEnv} for index $n-1$. (Only states $q$ such that $\colVirt' (q) = n-1$ may have
	changed their colors and be newly mapped to an environment.)
	Then, $(\colVirt',\msf{Ev}')$ is faithful
	down to $n-1$.
\end{lemma}
Before giving a proof sketch of this lemma, let us illustrate it on an example.
\begin{example}
	\label{ex:turthful_one_step}
	Let us illustrate this algorithm on Figures~\ref{fig:InitialGame}
	and~\ref{fig:SomeStepsGame}. The first step is to build a pair that
	is faithful down to $e = 4$. As mentioned above in
	Example~\ref{ex:compute_new_col}, we have $\newC(q_5,\colFunc) =
	4$. Hence, the color of this state is changed to $4$ (we obtain a
	virtual coloring function $\colVirt^{q_5}$) and we set
	$\msf{Ev}(q_5) := E^4_{q_5,\colFunc}$. Note that a Player-$\A$
	$\GF$-strategy $\sigma_\A$ is optimal in this environment if and
	only if it plays both rows with positive probability. Furthermore,
	note that, in the extracted game $\Zone^{4}_{\colVirt^{q_5}}$, a
	Player-$\A$ positional strategy playing such a $\GF$-strategy
	$\sigma_\A$ in $q_5$ parity dominates the valuation
	$v_{Q_4,\colFunc^{q_5}}^{1/2}$. Hence, the pair
	$(\colVirt^{q_5}, \msf{Ev})$ is faithful down to $4$.
	
	Consider now the layer $3$. First, the state $q_6$ already has color
	$3$, so it only remains to set its environment:
	$\msf{Ev}(q_6) := E^{3}_{q_6,\colVirt^{q_5}}$.  We then realize that
	$\newC(q_4,\colVirt^{q_5}) = 3$. Indeed, $q_4$ is colored with $2$
	and may go with equal probability to a state colored with $0$ and to
	a state colored with $3$. The color of this state is therefore
	changed, thus obtaining a new virtual coloring function
	$\colVirt^{q_5,q_6,q_4}$. We set its environment:
	$\msf{Ev}(q_4) := E^{3}_{q_4,\colVirt^{q_5}}$. One can realize that the
	pair $(\colVirt^{q_5,q_6,q_4}, \msf{Ev})$ witnesses the color $3$
	: a positional
	Player-$\B$ strategy generated by this environment would be so that
	(i) from $q_6$, it goes to $q_4$ with probability $1$ (to avoid
	$k_4$ that is colored with $4$) and (ii) from $q_5$, it goes to
	$q_6$ with positive probability (to see the color $3$ with positive
	probability). Such a strategy has value $0$ in the game
	$\Zone^3_{\colVirt^{q_5}} = \Zone^3_{\colVirt}$ from
	Figure~\ref{fig:extracted_game}, hence the pair
	$(\colVirt^{q_5,q_6,q_4}, \msf{Ev})$ witnesses the color $3$.
	
	We illustrate on this step why the environment needs to be set
	before changing the new color and not after. That is, we explain why
	it would not be correct to set
	$\msf{Ev}(q_4) := E^{3}_{q_4,\colFunc^{q_5,q_6,q_4}}$ instead of
	what we do above.
	In this environment, the state $q_4$ has color $3$. Hence, looping
	with probability $1$ on $q_4$ is an optimal $\GF$-strategy for
	Player $\B$ w.r.t.
	$(\fsf(q_4),\msf{Ev}(q_4))$. Then, the
	corresponding pair of coloring and environment functions would not
	witness the color $3$. Indeed, a Player $\B$ strategy that loops
	with probability $1$ on $q_4$ is generated by this environment, and
	it has value $1 \geq u$ (because the real color of this state is
	$2$, and not $3$).
	
	This process is then repeated down to $0$. In
	Figure~\ref{fig:SomeStepsGame}, the depicted coloring function (with
	the appropriate environment function that is not shown in
	Figure~\ref{fig:SomeStepsGame}) are in fact completely faithful (this
	is what outputs Algorithm~\ref{fig:UpdateColEnv} on the
	coloring function of Figure~\ref{fig:InitialGame}). 
\end{example}

We give a proof sketch of Lemma~\ref{lem:algo_preserves_faithfulness},
which explains the ideas for the first phase of the procedure for
computing a first completely faithful pair, before
Algorithm~\ref{fig:IncLeast} is called -- which we will discuss right after.
\begin{proof*}[Proof sketch]
	We want to prove that the pair $(\colVirt',\msf{Ev}')$ witnesses the
	color $n-1$ (the other condition for faithfulness is ensured by the
	construction).  We consider the case where $n-1$ is even, the
	other case is similar (but one needs to take the point-of-view of
	Player $\B$). Consider a Player-$\A$ positional strategy $\s_\A$
	generated by the environment function $\msf{Ev}'$ in the game
	$\Zone^{n-1}_{{\colVirt'}}$.
	Let $Q_{n-1} := \colVirt'^{-1}[n-1]$ and let
	$v := v_{n-1,{\colVirt'}}^u$. For every $q \in Q_{n-1}$, let
	$Y_q := (\fsf(q),\msf{Ev}'(q))$ be the
	local environment at state $q$ and let us denote $\msf{Ev}'(q)$ by $\msf{Ev}'(q) = \langle c_q,e,p_q \rangle$
	.  From the characterization of
	Lemma~\ref{lem:opt_in_gf_equiv_opt_in_parity} (item (ii.1)),
	by carefully analyzing the links between the local games $\G_{Y_q}$
	for all $q \in Q_{n-1}$
	and the game $\Zone^{n-1}_{{\colVirt'}}$, we deduce that the
	strategy $\s_\A$ dominates the valuation $v$.
	
	It remains to show that all BSCCs (that are not reduced to a stopping
	state and are) compatible with $s_\A$ are even-colored.  Consider
	such a BSCC $H$ and a Player-$\B$ deterministic positional strategy
	$\s_\B$ which induces $H$.  For every state $q \in H$, since no stopping
	state occurs in $H$, it must be that the probability to reach a
	stopping state is $0$. That is, it amounts to have
	$\outM_{\Games{\fsf(q)}{\mathbbm{1}_{Q \setminus Q_u}}}(\sigma_\A,b) = 0$
	.
	%
	%
	For every state $q \in Q_{n-1}$, the coloring function $\colVirt_q$
	associated with environment $\msf{Ev}'(q)$ is such that
	$\colVirt_q(q) \leq n-1$.\footnote{This is because all states
		$q \in Q_{n-1}$ satisfy $\colFunc(q) \leq n-1$. This is
		one of the additional conditions for faithfulness that we did
		mention, but that is used in the Appendix in Definition~\ref{def:true_trithful_env_coloring_func}.} Hence, the color $c_q$
	is such that
	$c_q = \max(n-2,\colVirt_q(q)) \leq n-1$. Now, assume that some
	state $k_i$ is in $H$ for some $i > n-1 \ge c_q$. In that case,
	as explained in Remark~\ref{rmk:explanation_optim_gf_strat}, the
	highest $i$ such that $k_i$ is in $H$ must be even. Hence, $H$ is
	even-colored. Assume now that no state $k_i$ in $H$ is such that
	$i > n-1$. In that case, if a state in $H$ has color $n-1$ (like the
	state $q_6$ in Figure~\ref{fig:SomeStepsGame} in the case where
	$n-1 = 3$), then $n-1$ is the highest color in $H$ and $H$ is
	even-colored. Consider the first state $q$ whose color is now $n-1$
	(w.r.t. ${\colVirt'}$) but whose previous color was not $n-1$.
	In that case, we have $c_q = \max(n-2,\colVirt_q(q)) = n-2$ is
	odd. Furthermore, the state $q$ has changed its color because
	$\newC(q,\colVirt_q) =
	n-1$. 
	With Remark~\ref{rmk:explanation_optim_gf_strat}, since $\s_\A(q)$
	is optimal w.r.t. $Y_q$, it follows that there is a positive
	probability to reach, in the game $\G_{Y_q}$ the state $k_{n-1}$. In
	the game $\Zone^{n-1}_{{\colVirt'}}$, this corresponds to a positive
	probability to reach a state $q' \in H$ colored with $n-1$
	w.r.t. $\colVirt_q$ (recall
	Definition~\ref{def:probability_extracted_parity_game}). Since $q$
	is the first state to have changed its color, we can deduce that
	$q'$ 
	already had color $n-1$ w.r.t. $\colVirt$. Furthermore, one can show
	that $q'$ is colored with $n-1$ w.r.t. the real coloring function
	$\colFunc$. Overall, in the game $\Zone^{n-1}_{{\colVirt'}}$, with
	the $\GF$-strategy $\s_\A(q)$, there is a positive probability to
	reach in one step a state $q'$ colored with
	$n-1$. 
	Iteratively, we obtain that, considering the $k$-th state whose
	color is now $n-1$ (i.e. w.r.t. ${\colVirt'}$) but whose initial
	color was not $n-1$, there is a positive probability to reach (in at
	most $k$ steps) a state colored with $n-1$. Hence, the highest color
	appearing in $H$ is $n-1$, which is even. We obtain that $\s_\A$
	parity dominates the valuation $v$.
\end{proof*}

Overall, applying iteratively Algorithm~\ref{fig:UpdateColEnv} on
all colors from $e$ down to $0$ starting with the initial coloring
function induces a completely faithful pair
$(\colVirt,\msf{Ev})$. However, it may be the case that some states
are mapped to an odd number, which does not allow to apply
Lemma~\ref{lem:completely_faithful_all_even_ok}. 
The question is then: from that completely faithful configuration, how
can one make some progress towards a situation where
Lemma~\ref{lem:completely_faithful_all_even_ok} can be applied?
\begin{example}
	\label{ex:progress_when_completely_turhtful}
	Consider the coloring function of Figure~\ref{fig:SomeStepsGame}. As
	mentioned in Example~\ref{ex:turthful_one_step}, with an appropriate
	environment function (that is not shown in
	Figure~\ref{fig:SomeStepsGame}), we can have a pair which is
	completely faithful. 
	To gain some intuition on what should be done next, let us focus
	only on the states $q_1,q_2,q_3$. A simplified version is presented
	in Figure~\ref{fig:simpleGameThreeStates} (with a slight
	modification: instead of going to $q_0$, $q_1$ loops on itself): the
	initial (and true) colors of the states are in circles next to them
	and their color w.r.t. the current (virtual) coloring function (that
	is completely faithful with an appropriate environment function) is
	written in red
	. In this game, Player $\B$ plays alone, but it is obvious that Player
	$\A$ wins surely from $q_2$: indeed, either the game stays
	indefinitely in $q_2$, or it eventually reaches and settles in
	$q_1$.
	
	The current virtual color $1$ assigned to both $q_2$ and $q_3$
	does not properly reflect the fact that if the game reaches $q_3$,
	even though Player $\B$ plays optimally according to the local
	game associated with $q_2$, it will end up looping in $q_1$, which
	will be losing for Player $\B$. In a way, we would like to
	propagate the information that reaching $q_1$ is bad for Player
	$\B$. Since $0$ is the smallest color, there is no harm in
	increasing it to $2$, the game from $q_1$ will be the same: it
	will be won by Player $\A$ by looping. Player $\B$ will now be
	able to know that going to $q_1$ is dangerous for him, which will
	be obtained by applying the previous iterative process.
	
	
	In a more general concurrent game, the next step of the process
	when we have a completely faithful configuration not satisfying
	the assumptions of
	Lemma~\ref{lem:completely_faithful_all_even_ok}
	consists in changing all the states with the least (virtual) color
	$n$ to the color $n+2$. 
	However, note that there is a (very important) second step: the
	colors of all states (virtually) colored with $n+1$ should be reset
	to their initial colors. The reason why can be seen again in
	Figure~\ref{fig:simpleGameThreeStates}. After the color of $q_1$
	becomes $2$, the color of $q_3$ will also become $2$. However, if
	the color of the state $q_2$ is not reset, then it is not going to
	change 
	since Player $\B$ can choose to loop to $q_2$ and see the color $1$
	for ever (in game $\G^0_{q_2,\colVirt}$). That is, from Player $\B$'s perspective, looping
	indefinitely on $q_2$ is winning, which is not what happens in the
	real game (i.e. the coloring function does not faithfully describes
	what happens in the game). The changes made to the coloring function
	$\colVirt$ from Figure~\ref{fig:SomeStepsGame} can be seen in
	Figure~\ref{fig:ProgressCompletelyFaithful}. Note that the process
	of increasing the colors of some states by $2$ can only be done with
	the \emph{least} color (otherwise faithfulness will not be
	preserved).
\end{example}

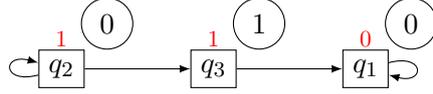
\begin{figure}
	\centering
	\begin{tikzpicture}
	\node[player] (c1) at (2*\lenV,2.4*\lenH) {\color{red}0} ;
	\node[draw=black,circle,scale=1] (oc1) at (2.6*\lenV,2.6*\lenH) {$0$} ;
	\node[draw=black,rectangle,scale=1] (q1) at (2*\lenV,2*\lenH) {$q_1$} ;
	
	\node[player] (c2) at (-2*\lenV,2.4*\lenH) {\color{red}1} ;
	\node[draw=black,circle,scale=1] (oc1) at (-1.4*\lenV,2.6*\lenH) {$0$} ;
	\node[draw=black,rectangle,scale=1] (q2) at (-2*\lenV,2*\lenH) {$q_2$} ;
	
	\node[player] (c3) at (0*\lenV,2.4*\lenH) {\color{red}1} ;
	\node[draw=black,circle,scale=1] (oc1) at (0.6*\lenV,2.6*\lenH) {$1$} ;
	\node[draw=black,rectangle,scale=1] (q3) at (0*\lenV,2*\lenH) {$q_3$} ;
	
	\path[-latex]  	
	(q1) edge[loop right] (q1)
	(q3) edge (q1)
	(q2) edge (q3)
	(q2) edge [loop left] (q2)
	;
\end{tikzpicture}
	\caption{A (deterministic turn-based) game with only three states.}
	\label{fig:simpleGameThreeStates}  
\end{figure}
\begin{figure}
	\centering
	\includegraphics[scale=0.4]{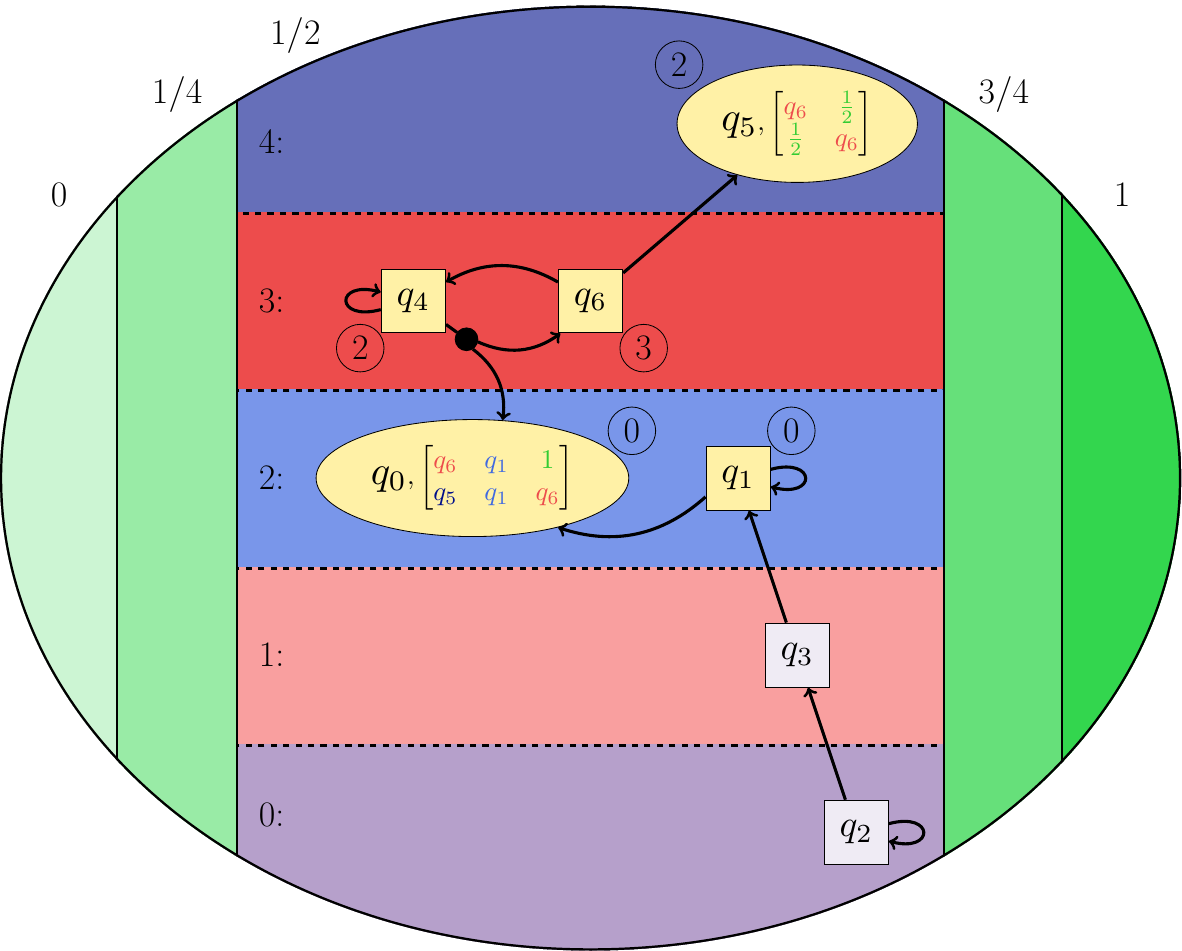}
	\caption{The same arena as in Figures~\ref{fig:InitialGame},\ref{fig:SomeStepsGame} but with a different coloring function.}
	\label{fig:ProgressCompletelyFaithful}  
\end{figure}

The process described in
Example~\ref{ex:progress_when_completely_turhtful} is implemented as
Algorithm~\ref{fig:IncLeast} in the Appendix,
it ensures the lemma below.
\begin{lemma}[Proof Appendix~\ref{subsubsec:proof_lem_algo_increase_least_color},  Page~\pageref{subsubsec:proof_lem_algo_increase_least_color}]
	\label{lem:algo_increase_least_color}
	Let 
	$\colVirt: Q_u \rightarrow \llbracket 0,e \rrbracket$,
	$\msf{Ev}: Q_u \rightarrow \msf{Env}(Q)$ be a coloring and an environment functions. 
	Let $n := \min \colVirt[Q]$. Assume that $n \leq
	e-2$ and the pair $(\colVirt,\msf{Ev})$ is completely faithful. If $(\colVirt',\msf{Ev}') \gets
	\msf{IncLeast}(\colVirt,\msf{Ev})$ is the result of
	increasing the least-colored layer by 2 and resetting the environment of the last but least-colored layer (Algorithm~\ref{fig:IncLeast}), then 
	$(\colVirt',\msf{Ev}')$ is faithful down to $n+2$.
\end{lemma}
\begin{proof*}[Proof sketch]
	Let $Q_n := \colVirt^{-1}[n]$ and
	$Q_{n+2} := \colVirt^{-1}[n+2]$.  The algorithm has three steps:
	first, it increases the least color by $2$;
	then it resets the environments of the $(n+1)$-colored
	states; finally it applies Algorithm~\ref{fig:UpdNewSta} to these reset states. Let us argue that 
	$(\colVirt'',\msf{Ev})$ (obtained
	after the first step) witnesses the color $n+2$.
	


	Consider a Player-$\A$ positional strategy $\s_\A$ generated by the
	environment $\msf{Ev}$ in the game $\Zone^{n+2}_{{\colVirt''}}$. Let
	$v := v^u_{n+1,\colVirt''}$. Similarly to the proof of Lemma
	\pat{\ref{lem:algo_preserves_faithfulness}}, $\s_\A$ dominates the
	valuation $v$.
	Consider a BSCC $H$ compatible with $\s_\A$. If
	$H \cap Q_{n+2} = \emptyset$, then $H$ is even-colored. Indeed,
	$(\colVirt,\msf{Ev})$ witnesses the color $n$. In addition, the
	probability to go to a state $k_{i}^{n+2}$ that is $(n+1)$-colored in the game
	$\Zone^{n+2}_{{\colVirt''}}$ is exactly the probability to
	go to a state $k_{i}^{n}$ that is $(n+1)$-colored 
	in the game $\Zone^{n}_{\colVirt}$ (since $n$ is the least color). Furthermore, $H$ is also even-colored as soon as $H \cap Q_n = \emptyset$ since $(\colVirt,\msf{Ev})$ witnesses the color $n+2$. Now, assume that
	none of these cases occur.
	%
	Then, one can show that: either a state $k_i$ is seen for some
	$i \geq n+2$, and $H$ is even-colored; or, from some states in
	$Q_{n+2}$, there is a positive probability to exit $Q_{n+2}$ and no
	state $k_i$ is seen for $i \geq n+2$. Now, looking at what happens
	in game $\Zone^{n+2}_{{\colVirt}}$, some states $k_i$ are
	seen for $i \leq n+1$, and such states are colored with
	$n+1$. Hence, since $(\colVirt,\msf{Ev})$ witnesses the color
	$n+2$, it must be that the highest color in $H$ is $n+2$, which is
	even. Therefore it is also the case in the game
	$\Zone^{n+2}_{{\colVirt''}}$. In all the cases, $H$ is even-colored.
	%
\end{proof*}

As stated in Lemma~\ref{lem:algo_increase_least_color}, the update of colors described in Example~\ref{ex:progress_when_completely_turhtful} can be done only if, for a completely faithful pair, the least (virtual) color $n$ appearing is at most $e-2$. If $n = e$
, we are actually in the scope of Lemma~\ref{lem:completely_faithful_all_even_ok} since in that case all states have (virtual) color $e$. However, there remains the case where 
we have $n = e-1$. In fact, this case cannot happen.
\begin{lemma}[Proof Appendix~\ref{subsubsec:proof_lem_minimum_faithful_cannot_be_d_minus_one},  Page~\pageref{subsubsec:proof_lem_minimum_faithful_cannot_be_d_minus_one}]
	\label{lem:minimum_faithful_cannot_be_d_minus_one}
	Consider a coloring function
	$\colVirt: Q_u \rightarrow \llbracket 0,e \rrbracket$, an
	environment function
	$\msf{Ev}: Q_u \rightarrow \msf{Env}(Q)$. Assume that
	$(\colVirt,\msf{Ev})$ is completely faithful. Then, for
	$C := \colVirt[Q]$, we have $\min C \neq e-1$.
\end{lemma}
\begin{proof*}[Proof sketch]
	Let $Q_{e-1} := \colVirt^{-1}[e-1]$. Towards a contradiction
	, let $\s_\B$ be a Player-$\B$ positional strategy generated by
	$\msf{Ev}$ in the game $\Zone^{e-1}_{\colVirt}$. It
	parity dominates the valuation $v_{e-1,\colVirt}^{u'}$ for some
	$u' < u$
	. Hence, all BSCCs compatible with $\s_\B$ are odd-colored: they all 
	stay in the layer $Q_{e-1}$. Indeed, since $e-1 = \min C$, 
	exiting $Q_{e-1}$ while staying in $Q_u$ means seeing $Q_e := \colVirt^{-1}[e]$ with $e$ even and the highest color in the game.
	Hence, either the game stays indefinitely in $Q_{e-1}$ and Player
	$\B$ wins almost surely, or there is some positive probability to
	visit stopping states, and in that case their expected values is at most
	$u'$. Hence, in the game $\G^u$, the strategy $\s_\B$ has
	values less than $u$ from the states $Q_{e-1} \subseteq Q_u$, which
	is a contradiction. 
\end{proof*}

Finally, all these pieces are put together in
Algorithm~\ref{fig:ComputeEnv} in the Appendix, whose output is a
completely faithful pair where all states are mapped to $e$. 
The only remaining step is to prove the termination of this algorithm. Let us consider the (virtual) coloring functions as vectors in $\N^{e+1}$ indicating the number of states mapped to each color. 
Then, one can realize that each step of Algorithm~\ref{fig:ComputeEnv} increases this vector for a lexicographic order (i.e. we first compare the number of states mapped to $e$, then the number of states mapped to $e-1$,
etc). Hence, Algorithm~\ref{fig:ComputeEnv} does terminate in
finitely many steps. 
\begin{lemma}[Proof Appendix~\ref{subsubsec:proof_lem_metric_work}, Page~\pageref{subsubsec:proof_lem_metric_work}]
	\label{lem:metric_work}
	Algorithm~\ref{fig:ComputeEnv} computes a completely faithful pair
	of environment and coloring functions mapping each state to $e$ in
	finitely many steps.
\end{lemma}	

We can now proceed to the proof of Theorem~\ref{thm:transfer_local_global}.
\label{proof:mainthm}
\begin{proof*}
	Let us prove Theorem~\ref{thm:transfer_local_global} for Player
	$\A$. Consider some $u \in V_Q \setminus \{ 0 \}$. By
	Lemma~\ref{lem:metric_work}, there is a completely faithful pair of
	environment and coloring functions $(\colVirt_u,\msf{Ev}_\A^u)$
	mapping each state in $Q_u$ to $e_u$. Hence, by
	Lemma~\ref{lem:completely_faithful_all_even_ok}, all
	Player-$\A$ positional strategies generated by the environment
	function $\msf{Ev}_\A^u$ parity dominate the valuation $\MarVal{\G}$
	in the game $\G^u$. Since we assume that all game forms occurring in
	$Q_u$ are positionally maximizable up to $e_u - \colFunc(q)$
	w.r.t. Player $\A$, such positional strategies generated by
	$\msf{Ev}_\A^u$ do exist. Then, considering the environment function
	$\msf{Ev}_\A: Q \rightarrow \msf{Ev}(\distribSet)$ that merges all
	the environment functions
	$(\msf{Ev}_\A^u)_{u \in V_\G \setminus \{ 0 \}}$ together (and that
	is defined arbitrarily on $Q_0$), it follows by
	Lemma~\ref{lem:optimal_in_all_value_slice_optimal_everywhere},
	that all Player-$\A$ positional strategies generated by that
	environment function $\msf{Ev}$ are optimal. (And such strategies
	exist.)
\end{proof*}

\section{Positionally optimizable game forms}
\label{sec:positionally_max_gf}
In this section, we state some properties that are ensured by positionally optimizable game forms. We begin by giving some notations for positionally optimizable game forms (and the corresponding decision problems).
\begin{notation}
	For $n \in \N$, we let $\msf{ParO}(n)$
	be the set of all game forms positionally
	optimizable up to $n$ and we let $\msf{IsO}(n)$
	be the problem of deciding whether a
	game form is in $\msf{ParO}(n)$. Furthermore, we let
	$\msf{ParO} := \cap_{n \in \N} \msf{ParO}(n)$ and we denote by
	$\msf{IsO}$ the problem of deciding whether a game form is in
	$\msf{ParO}$.
\end{notation}

\subsection{Some positionally optimizable game forms} 
Let us first exhibit some positionally maximizable game forms. Given the distance to the highest color in the game,  Theorem~\ref{thm:transfer_local_global} and
Remark~\ref{rmk:parm_necessary} give exactly the game forms that
should be used in parity games to ensure the existence of positional
optimal strategies for both players. The game forms in $\msf{ParO}$
are the ones that can be used in all parity games, regardless of the
number of colors involved, while ensuring the existence of positional
optimal strategies (for both players). Such game forms do exist. For instance, 
, turn-based game forms are positionally optimizable (straightforwardly). In fact, all \emph{determined}~\cite{BBSFSTTCS21} game forms are positionally
optimizable. A deterministic (i.e. where all distributions over are the outcomes are deterministic) game form $\formNF \in \Frm(\outComeNF)$ is determined if it satisfies the following property: for all $\{ 0,1 \}$-valuations of the outcomes $v: \outComeNF \rightarrow \{ 0,1 \}$, in the game in normal form $\Games{\formNF}{v}$, there is either a row of $1$ or a column of $0$. In particular, all turn-based game forms are determined. Furthermore, all game forms with at most two outcomes are in $\msf{ParO}$. This is stated formally in the lemma below.
\begin{proposition}[Proof~\ref{proof:prop_determined_two_var_ok}, Page~\pageref{proof:prop_determined_two_var_ok}]
	\label{prop:determined_two_var_ok}
	Consider a set of outcomes $\outComeNF$ and a game form $\formNF = \langle \Act_\A,\Act_\B,\distribSet,\outCNF \rangle \in \Frm(\distribSet)$. Assume that $\formNF$ is either determined (in particular, it may be turn-based), or such that $\outComeNF \leq 2$. In that case, $\formNF \in \msf{ParO}$.
\end{proposition}

\subsection{Deciding if a game form is positionally optimizable}
Now, we want to show that it is decidable if a game form if positionally optimizable. To do so, we are going to encode this problem in the first order of the reals, that we introduce in the subsection below.
\subsubsection{First-order theory of the reals}
Informally, first order theory of the reals formulas use (existential or universal) quantification over real variables, classical logical connectors and (in)equalities between multi-variate polynomials. We formally define them below.
\begin{definition}[First-order theory of the reals]
	\label{def:first_order_theory_reals}
	In the first order theory of the reals ($\msf{FO}$-$\R$ for short), we consider formulas $\Phi$ of the shape:
	\begin{equation*}
	\Phi = Q_1 \; x_1 \in \R, \ldots, Q_n \; x_n \in \R, \varphi(x_1,\ldots,x_n)
	\end{equation*}
	where $n \in \N$ and
	\begin{itemize}
		\item for all $i \in \brck{1,n}$, $Q_i \in \{ \exists,\forall \}$. That is, it is either an existential or a universal quantifier; 
		\item $\varphi(x_1,\ldots,x_n)$ is a classical formula without quantifiers, with $\wedge$ (and), $\lor$ (or) and $\lnot$ (negation) as logical connectors. The atomic propositions considered in $\varphi(x_1,\ldots,x_n)$ are of the shape:
		\begin{equation*}
		P(x_1,\ldots,x_n) \bowtie 0
		\end{equation*}
		where $\bowtie \; \in \{ =,\neq,\geq,>,\leq,< \}$ and $P: \R^n \rightarrow \R$ is a real multi-variate polynomial with integer coefficients.
	\end{itemize}
\end{definition}
The semantics behind the first order theory of the reals is quite intuitive, though it would take some space to formally define. Instead, let us illustrate it on a couple of examples.
\begin{example}
	Consider the $\msf{FO}$-$\R$ formula below.
	\begin{equation*}
	\forall x \in \R,\; \exists y \in \R,\; x + y \geq 0
	\end{equation*}
	This $\msf{FO}$-$\R$ formula is true since it does hold that, for all $x \in \R$, taking $y := -x \in \R$, we have $x + y = 0$. However, the $\msf{FO}$-$\R$ formula
	\begin{equation*}
	\exists x \in \R,\; \forall y \in \R,\; x + y \geq 0
	\end{equation*}
	is false in $\msf{FO}$-$\R$ since there is no $x \in \R$, such that, for all $y \in \R$, we have $x + y \geq 0$.
\end{example}

Let us now consider the decision problem associated with $\msf{FO}$-$\R$ formulas.
\begin{definition}
	We denote by $\msf{True}_{\msf{FO}-\R}$ the set of all first-order theory of the reals formulas that evaluate to true. The decision problem $\msf{Is}_{\msf{FO}-\R}$ is as follows:
	\begin{itemize}
		\item Input: an $\msf{FO}$-$\R$ formula $\Phi$;
		\item Output: yes if and only if the formula $\Phi$ is in $\msf{True}_{\msf{FO}-\R}$.
	\end{itemize}
\end{definition}

\begin{theorem}
	\label{thm:decidable_theory}
	The decision problem $\msf{Is}_{\msf{Fo}-\R}$ is decidable.
\end{theorem}
This result was first shown in \cite{tarski1998decision}. It was then improved in \cite[Theorem 1.1]{DBLP:journals/jsc/Renegar92b} 
it as it was shown that $\msf{Is}_{\msf{FO}-\R}$ could be decided in doubly-exponential time via a quantifier elimination procedure. 

\subsubsection{Encoding in the first order theory of the reals}
Let us come back to the problem of interest of us: deciding if a game form is positionally optimizable. To tackle this issue, we first realize that a game form is positionally optimizable if and only if there are optimal $\GF$-strategies w.r.t. all \emph{relevant} environments, defined below.
\begin{definition}[Relevant environments]
	\label{def:relevant_env}
	For a set of outcomes $\outComeNF$, an environment $E = \langle c,e,p \rangle \in \msf{Env}(\outComeNF)$ is \emph{relevant} if $c \in \{ 0,1 \}$, $p^{-1}[\{ c-1 \}] = p^{-1}[\{ \qinit \}] = \emptyset$ and, for all $i \in \brck{c,e}$, there is $o \in \outComeNF$ such that $p(o) = k_i$. The \emph{size} of a relevant environment $E$ is equal to $\msf{Sz}(E) := e - c$.
\end{definition}

As mentioned above, we may only consider relevant environment to decide if a game form is positionally maximizable. 
\begin{proposition}[Proof Appendix~\ref{appen:proof_prop_relevant_env_sufficient}, Page~\pageref{appen:proof_prop_relevant_env_sufficient}]
	\label{prop:relevant_env_sufficient}
	Consider a set of outcomes $\outComeNF$ and a game form
	$\formNF \in \Frm(\outComeNF)$. Consider a Player $\C \in \{ \A,\B \}$. For all $n \in \N$, the game form $\formNF$ is positionally maximizable w.r.t. Player $\C$ up to $n$ if and only if, for all relevant environments $E$ with $\msf{Sz}(E) \leq n-1$, there is an optimal $\GF$-strategy in $\formNF$ for Player $\C$ w.r.t. $(\formNF,E)$.
\end{proposition}
The proof of this proposition is quite long and actually uses a result that is not published yet, but that is in the PhD thesis of one of the authors, which will soon be available online. 

The benefit of considering only relevant environment is that, given any $l \in \N$, deciding if a game form is positionally optimizable up to $l$ can be done by considering environments where all the outcomes are mapped to integers at most $l$. Then, the fact that a game form is positionally optimizable can be encoded in the first-order theory of the reals.
\begin{proposition}
	\label{prop:positionally_optim_decidable}
	Consider a deterministic game form $\formNF$. The fact that $\formNF \in \msf{ParO}$ (resp. $\formNF \in \msf{ParO}(n)$, for some $n \in \N$) can be encoded, with formula of size polynomial in $|\formNF|$ (resp. and $n$), in $\msf{FO}$-$\R$. 
\end{proposition}
\begin{proof}
	Consider some $l \in \N$. We encode the fact that a standard finite deterministic game form $\formNF = \langle \Act_\A,\Act_\B,\outComeNF,\outCNF \rangle$ is positionally optimizable up to $l$. Without loss of generality, since the game form $\formNF$ is deterministic, we assume that $\outComeNF = \outCNF(\Act_\A,\Act_\B)$. To do so, we use the characterization of Proposition~\ref{prop:relevant_env_sufficient} and we express in $\msf{FO}$-$\R$ the fact that, for all relevant environments $E = \langle c,e,p \rangle$ of size at most $l-1$, both players have an optimal $\GF$-strategy. We use the characterization given by Lemma~\ref{lem:opt_in_gf_equiv_opt_in_parity}. We assume that $\Act_\A = \llbracket 1,n \rrbracket$, $\Act_\B = \llbracket 1,m \rrbracket$ and we consider the formula $\Phi_{\formNF}^{\msf{ParO}(l)}$ below, that we will explain line by line in the following.
	\begin{align*}
	\hspace*{-0.1cm}
	\Phi_{\formNF}^{\msf{ParO}(l)} := \forall c, & \; \msf{ZeroOrOne}(c) \; \wedge \\
	\forall \alpha = (\alpha_{o})_{o \in \outComeNF}, & \; \msf{ZeroOrOne}(\alpha) 
	\; \wedge \\ 
	\forall v = (v_{o})_{o \in \outComeNF}, & \; \msf{RealBetweenZeroOne}(v) \; \wedge \\ 
	\forall k = (k_{o})_{o \in \outComeNF}, & \; \msf{IntBetweenZeroL}(k,c) \; \wedge \\ 
	\exists u, & \; (0 \leq u \leq 1) \; \wedge \\ 
	\exists S_\A = S_\A^1,\ldots,S_\A^n,\; & \\
	\exists \sigma_\A = \sigma_\A^1,\ldots,\sigma_\A^n,\; & \mathsf{IsIndicator}(S_\A) \wedge 
	\mathsf{IsStrategy}_\A(\sigma_\A) \wedge \mathsf{IsSupp}(\sigma_\A,S_\A) \wedge \mathsf{Val}_\A(\sigma_\A,\alpha,v,u) \wedge \\
	& \bigwedge_{1 \leq j \leq m} 
	\msf{MaxIntegerEven}(S_\A,k,j,c) \; \wedge \\
	\exists S_\B = S_\B^1,\ldots,S_\B^m,\; & \\
	\exists \sigma_\B = \sigma_\B^1,\ldots,\sigma_\B^n,\; & \mathsf{IsIndicator}(S_\B) \wedge 
	\mathsf{IsStrategy}_\B(\sigma_\B) \wedge \mathsf{IsSupp}(\sigma_\B,S_\B) \wedge \mathsf{Val}_\B(\sigma_\B,\alpha,v,u) \wedge \\
	& \bigwedge_{1 \leq i \leq n} 
	\msf{MaxIntegerOdd}(S_\B,k,i,c)
	\end{align*}
	The formula $\Phi_{\formNF}^{\msf{ParO}(l)}$ does not fit exactly the formalism of Definition~\ref{def:first_order_theory_reals}, since all the quantifiers are not at the beginning of the formula. However, the semantics is unchanged if these quantifiers are moved at the beginning of the formula. We presented $\Phi_{\formNF}^{\msf{ParO}(l)}$ in that way for readability. 
	
	The first four lines encode the relevant environment $E$. Specifically:
	\begin{equation*}
	\msf{ZeroOrOne}(c) := (c = 0) \lor (c = 0) \;
	\end{equation*}
	Since the environment $E$ is relevant, $c$ is equal to either 0 or 1. 
	Furthermore:
	\begin{equation*}
	\msf{ZeroOrOne}(\alpha) := \bigwedge_{o \in \outComeNF} ((\alpha_o = 0) \lor (\alpha_o = 1))
	\end{equation*}
	where, for all $o \in \outComeNF$, $\alpha_o = 1$ means that $p(o) \in [0,1]$ and $\alpha_o = 0$ means that $p(o) \in K_l$. Furthermore:
	\begin{equation*}
	\msf{RealBetweenZeroOne}(v) := \bigwedge_{o \in \outComeNF} (0 \leq v_o \leq 1)
	\end{equation*}
	where, for all $o \in \outComeNF$, $v_o \in [0,1]$ corresponds to the value of $o$ w.r.t. $p$ (assuming $p(o) \in [0,1]$). Furthermore: 
	\begin{equation*}
	\msf{IntBetweenZeroL}(k,c) := \bigwedge_{o \in \outComeNF} (\bigvee_{0 \leq p \leq l} k_o = p) \wedge  ((c = 0) \Rightarrow (k_o \leq l-1) ) 
	\end{equation*}
	where, for all $o \in \outComeNF$, $k_o \in \brck{0,l}$ is the index of $o$ given by $p$ (assuming that $p(o) \in K_l$). Furthermore, since the size of $E$ is at most $l-1$, if $c = 0$, then the maximum of the colors should be at most $l-1$. 
	
	Then, $u$ is the value of the game $\G_{(\formNF,E)}$ from the state $\qinit$. This is ensured by the remainder of the formula: it exhibits a $\GF$-strategy per player whose corresponding positional strategy parity dominates the valuation $v_Y^u$ for $Y := (\formNF,E)$. The predicates are used for both players, we give them only for Player $\A$, the case of Player $\B$ being analogous.
	\begin{equation*}
	\mathsf{IsIndicator}(S_\A) := \bigwedge_{1 \leq i \leq n} ((S_\A^i = 0) \lor (S_\A^i = 1))
	\end{equation*}
	where $S_\A$ encodes the support of the $\GF$-strategy $\sigma_\A$. The fact that it is indeed a Player-$\A$ $\GF$-strategy is ensured by the predicates below: 
	\begin{equation*}
	\mathsf{IsStrategy}_\A(\sigma_\A) := \bigwedge_{1 \leq i \leq n} ((0 \leq \sigma_\A^i \leq 1) \wedge (\sum_{i = 1}^n \sigma_\A^i = 1)) 
	\end{equation*}
	In addition, we have:
	\begin{equation*}
	\mathsf{IsSupp}(\sigma_\A,S_\A) := \bigwedge_{1 \leq i \leq n} ((\sigma_\A^i > 0) \Leftrightarrow (S_\A^i = 1))
	\end{equation*}
	This predicate ensures that $S_\A$ does indeed correspond to the support of the $\GF$-strategy $\sigma_\A$. Furthermore:
	\begin{equation*}
	\mathsf{Val}_\A(\sigma_\A,\alpha,v,u) := \bigwedge_{1 \leq j \leq k} (\sum_{i = 1}^n \sigma_\A^i \cdot (\alpha_{\outCNF(i,j)} \cdot v_{\outCNF(i,j)} + (1 - \alpha_{\outCNF(i,j)}) \cdot u) \geq u)
	\end{equation*}
	This predicates encodes the fact that the Player-$\A$ $\GF$-strategy $\sigma_\A$ dominates the valuation $v_Y^u$. Indeed, note that for all $o \in \outComeNF$, if $\alpha_o = 1$, then $p(o) \in [0,1]$, and therefore $v_Y^u(o) = p(o) = v_o$. Similarly, if $\alpha_o = 0$, then $p(o) \in K_l$, and therefore $v_Y^u(o) = u$. Finally, when $l$ is even, denoting $l = 2 \cdot x$ with $x \in \N$:
	\begin{align*}
	\msf{MaxIntegerEven}(S_\A,k,j,c) := & \bigvee_{1 \leq i \leq n} (S_\A^i = 1 \wedge \alpha_{\outCNF(i,j)} = 1) \lor \\
	& \bigvee_{0 \leq e \leq x} (2 e \geq c) \wedge (( \bigvee_{1 \leq i \leq n} (S_\A^i = 1) \wedge (k_{\outCNF(i,j)} = 2 e)) \; \wedge \\
	& (\bigwedge_{1 \leq i \leq n} (S_\A^i = 1) \Rightarrow (k_{\outCNF(i,j)} \leq 2 \cdot e)) )
	\end{align*}
	It is similar if $l$ is odd, up to small changes. This predicate encodes item $(ii.2)$ of Lemma~\ref{lem:opt_in_gf_equiv_opt_in_parity}: for all columns $j \in \Act_\B$, either there is positive probability to see an outcome mapped to a value in $[0,1]$ (i.e. $\outM{\Games{\formNF}{\mathbbm{1}_{p^{-1}[0,1]}}}(\sigma_\A,j) > 0$) or the maximum of the colors in the support of the strategy $\sigma_\A$ is even (i.e. $\max (\mathsf{Color}(\formNF,p,\sigma_\A,j) \cup \{ c \})$ is even). 
	%
	
	By Lemma~\ref{lem:opt_in_gf_equiv_opt_in_parity} and Proposition~\ref{prop:relevant_env_sufficient}, we have that $\formNF \in \msf{ParO}(l)$ if and only if $\Phi_{\formNF}^{\msf{ParO}(l)} \in \msf{True}_{\msf{FO}-\R}$. Furthermore, the size of $\Phi_{\formNF}^{\msf{ParO}(l)}$ is polynomial in the size of $\formNF$ (and in $l$). 
	
	Finally, since given an environment, an outcome is mapped to at most one integer, and since in relevant environments of size $n \in \N$, the outcomes are mapped to at least $n-1$ different integers, we have $\formNF \in \msf{ParO} \Leftrightarrow \formNF \in \msf{ParO}(n+2)$ for $n := |\outComeNF|$. The proposition follows.
\end{proof}

\subsection{Hierarchy}
For all $n \in \N$, the game forms in $\msf{ParO}(n)$ are the ones to
be used --- to ensure the existence of positional optimal strategies
for both players --- in parity games at states where the gap between
the color of the state and the maximum color in the game is at most
$n$. Straightforwardly, $\msf{ParO}(n) \subseteq \msf{ParO}(n+1)$. In fact,
this inclusion is strict for all $n \in \N$. This defines an infinite
hierarchy of game forms.
\begin{proposition}
	\label{prop:inclustion_parm}
	For all $n \in \N$, the set of game forms positionally optimizable up to $n$ (i.e. $\msf{ParO}(n)$) strictly contains the set of game forms optimizable up to $n+1$ (i.e. $\msf{ParO}(n+1)$). That is, we have $\msf{ParO}(n) \supsetneq \msf{ParO}(n+1)$.
\end{proposition}
For all $n \geq 1$, we exhibit a standard finite game form $\formNF_n$ that is in $\msf{ParO}(n-1)$ but not in $\msf{ParO}(n)$. This is done in Definition~\ref{def:game_form_distinct} below, where we describe a game form that is $\msf{ParO}(n)$ but not in $\msf{ParO}(n-1)$ in the case where $n$ is even. 
\begin{definition}
	\label{def:game_form_distinct}
	Consider some even integer $n \geq 2$. We consider the set of outcomes $\outComeNF_n := \{ x_0,x_1,\ldots,x_{n-1},y,z \}$ and we consider the game form $\formNF_n = \langle \Act_\A^n,\Act_\B^n,\outComeNF_n,\outCNF_n \rangle_\msf{s}$ depicted in Figure~\ref{fig:GameFormDetail}. 
	Let us describe the set of actions of the players.
	We set:
	\begin{itemize}
		\item $\Act_\A^n := \{ a_\msf{t},a_\msf{b},a_1,\ldots,a_{n-1},a_{1,0},\ldots,a_{n-1,n-2},a_{\msf{Ex}} \}$; and
		\item $\Act_\B^n := \{ b_\msf{l},b_\msf{r},b_0,\ldots,b_{n-2},b_{2,1},\ldots,b_{n-2,n-3},b_{\msf{Ex}} \}$.
	\end{itemize}
	The actions $a_\msf{t},a_\msf{b}$ (resp. $b_\msf{l},b_\msf{r}$) refer to the two topmost rows (resp. leftmost columns): $a_\msf{t}$ (resp. $b_\msf{l}$) leads to $\frac{x_0+\ldots+x_{n-1}}{n},\frac{3y+z}{4}$ whereas $a_\msf{b}$ (resp. $b_\msf{r}$) leads to $\frac{3y+z}{4},\frac{y+3z}{4}$. Then, for all odd $i \leq n-1$ (resp. even $j \leq n-2$) $a_i$ and $a_{i,i-1}$ (resp. $b_j$ and $b_{j,j-1}$ --- only if $j \geq 2$) correspond to the rows leading to $x_i$ and $\frac{x_i + x_{i-1}}{2}$ (resp. columns leading to $x_j$ and $\frac{x_j + x_{j-1}}{2}$) respectively. Finally, action $a_\msf{Ex}$ (resp. $b_\msf{Ex}$) correspond to the bottommost row (resp. rightmost column). 
\end{definition}

\begin{figure}[!t]
	\hspace*{-0.8cm}
	\includegraphics[scale=1.3]{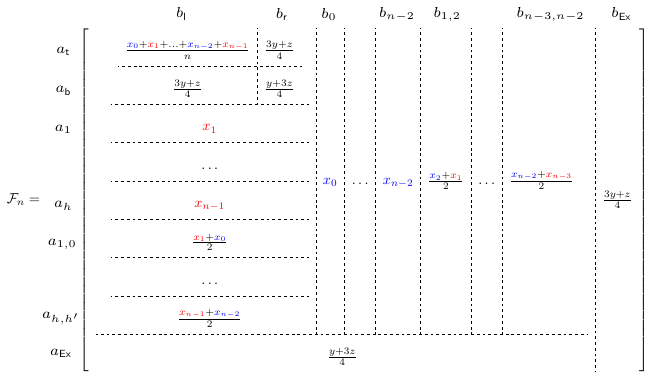}
	\caption{The game form $\formNF_n$. Due to a lack of space, $h$ refers to $n-1$ and $h'$ to $n-2$.}
	\label{fig:GameFormDetail}  
\end{figure}

The game forms described in Definition~\ref{def:game_form_distinct} satisfies the lemma below.
\begin{lemma}[Proof~\ref{proof:lem_gf_distinct}, Page~\pageref{proof:lem_gf_distinct}]
	\label{lem:gf_distinct}
	Consider some even $n \geq 2$. The game form of $\formNF_n$ Definition~\ref{def:game_form_distinct} is:
	\begin{itemize}
		\item positionally maximizable w.r.t. Player $\B$;
		\item positionally maximizable w.r.t. Player $\A$ up to $n-1$ but not up to $n$.
	\end{itemize}
	That is: $\formNF_n \in \msf{ParO}(n) \setminus \msf{ParO}(n-1)$.
\end{lemma}
The proof of this lemma is very tedious since it requires to consider a lot of different possibilities in terms of which outcome is mapped to which value or color. We give here an informal proof of this lemma in the case where $n = 2$
. The game form $\formNF_2$ is depicted in Figure~\ref{fig:gf_n_eq_2}.
\begin{figure}
	\begin{minipage}[t]{0.48\linewidth}
		\hspace*{-0.3cm}
		\centering
		\includegraphics[scale=0.65]{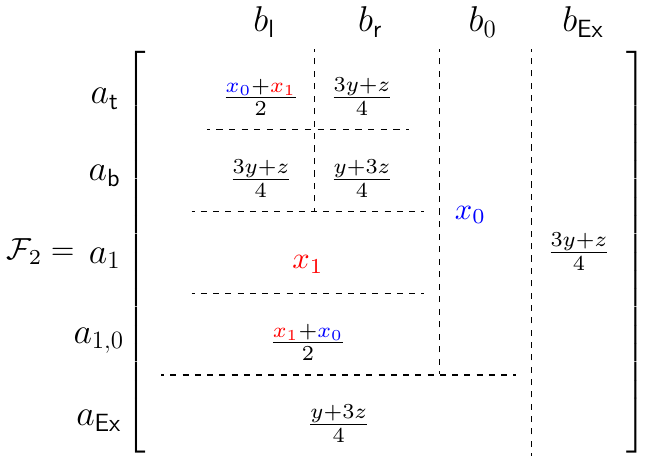}
		\caption{The game form $\formNF_2$.}
		\label{fig:gf_n_eq_2}  
	\end{minipage}
	\hspace*{0.1cm}
	\begin{minipage}[t]{0.48\linewidth}
		\centering
		\includegraphics[scale=1]{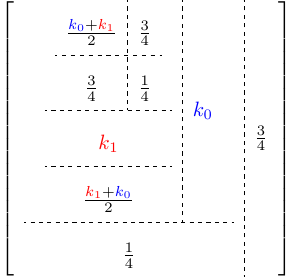}
		\caption{The game form $\formNF_2$ in a specific environment.}
		\label{fig:gf_n_eq_2_replace}  
	\end{minipage}
\end{figure}
\begin{proof}[Proof sketch]
	Consider a relevant environment $E = \langle c,e,p \rangle \in \msf{Env}(\outComeNF)$ and the corresponding parity game $\G_{(\formNF_2,E)}$ (recall Definition~\ref{def:parity_game_from_gf}). We let $u := \MarVal{\G_{(\formNF_2,E)}}(\qinit)$. For all $t \in \outComeNF_2$, if $p(t) \notin [0,1]$, we let $n_t \in \brck{0,e}$ denote the integer that $t$ is mapped to w.r.t. $p$, i.e. the integer ensuring $p(t) = k_{n_t}$.
	%
	We want to show that, in any case, Player $\B$ has a $\GF$-strategy that is optimal w.r.t. $(\formNF_2,E)$ and that this also holds for Player $\A$ as long as $c = e$ (i.e. if $\msf{Sz}(E) = 0 = n-2$). However, there is a relevant environment $E$ with $\msf{Sz}(E) = 1$ such that Player $\A$ has no $\GF$-strategy optimal w.r.t. $(\formNF_2,E)$. 
	There are several cases, we detail some of them.
	\begin{itemize}
		\item Assume that $p(y),p(z) \notin [0,1]$. This implies $u \in \{ 0,1 \}$ and playing action $a_\msf{Ex}$ for Player $\A$ and action $b_\msf{Ex}$ for Player $\B$ is optimal w.r.t. $(\formNF_2,E)$.
		\item Assume now that $p(y) \in [0,1]$ and $p(z) \notin [0,1]$. Then, we have $u = p(y)$. The reason is that, by playing action $a_\msf{Ex}$, Player $\A$ ensures that almost surely a stopping state of value $p(y)$ is reached. Furthermore, Player $\B$ can ensure the same thing by playing action $b_\msf{Ex}$. Then, playing action $a_\msf{Ex}$ for Player $\A$ and action $b_\msf{Ex}$ for Player $\B$ is optimal w.r.t. $(\formNF_2,E)$. This is similar if $p(y) \notin [0,1]$ and $p(z) \in [0,1]$.
		\item Assume now $p(y),p(z) \in [0,1]$. If $p(y) < p(z)$, then we have $u = \frac{3 p(y) + p(z)}{4}$, and as before playing action $a_\msf{Ex}$ for Player $\A$ and action $b_\msf{Ex}$ for Player $\B$ is optimal w.r.t. $(\formNF_2,E)$. 
		\item Assume now that $p(y),p(z) \in [0,1]$ and $p(y) \geq p(z)$. 
		If $p(x_0) \in [0,1]$ or $p(x_1) \in [0,1]$, we can also exhibit an optimal $\GF$-strategy for both players. 
		\item Let us now assume that $p(y),p(z) \in [0,1]$ with $p(y) \geq p(z)$. For simplicity, we assume that $p(y) = 1$ and $p(z) = 0$. In that case, we have:
		\begin{equation*}
		\frac{1}{4} = \frac{3 p(z) + p(y)}{4} \leq u \leq \frac{p(z) + 3 p(y)}{4} = \frac{3}{4}
		\end{equation*}
		Assume also that $p(x_0),p(x_1) \notin [0,1]$. Assume first that $n_{x_0}$ is odd. In that case, $u = \frac{1}{4}$. A Player-$\B$ strategy playing positionally action $b_0$ achieves this value in the game $\G_{(\formNF_2,E)}$
		. Thus, the Player-$\A$ $\GF$-strategy $a_{\msf{Ex}}$ is optimal w.r.t. $(\formNF_2,E)$. 
		
		Assume now that $n_{x_0}$ is even. Then, we have $u = \frac{3}{4}$: Indeed, for all $\varepsilon > 0$, a Player-$\A$ positional strategy playing action $a_\msf{t}$ with probability $1 - \varepsilon$ and action $a_\msf{b}$ with probability $\varepsilon$ has value at least $\frac{3}{4} - \varepsilon$ in the game $\G_{(\formNF_2,E)}$. Therefore, playing action $b_{\msf{Ex}}$ is always optimal for Player $\B$. It will also be the case for Player $\A$ if $\msf{Sz}(E) = 0$, i.e. if $c = e$. Indeed, in that case, we have either $n_{x_0} = n_{x_1} = 0$, in which case $u = \frac{3}{4}$ and the Player-$\A$ $\GF$-strategy $a_{1}$ is optimal w.r.t. $(\formNF_2,E)$; or $n_{x_0} = n_{x_1} = 1$, in which case we are in the scope of the previous case of this item since $n_{x_0}$ is odd. However, if $\msf{Sz}(E) = 1$, there may be an issue for Player $\A$, as we show below.
	\end{itemize}
	
	Consider now a relevant environment $E = \langle c,e,p \rangle$ of size 1 such that $c := 0$, 
	$p(y) := 1$, $p(z) := 0$, $p(x_0) = k_0$ and $p(x_1) = k_1$. We do have $\msf{Sz}(E) = 1$. What we obtain is depicted in Figure~\ref{fig:gf_n_eq_2_replace}. In that case, we have $u := \MarVal{\G_{\formNF_2,E}}(\qinit) = \frac{3}{4}$, as argued above in the last item. Consider any Player-$\A$ $\GF$-strategy $\sigma_\A \in \Sigma_\A(\formNF_2)$. If it plays action $a_\msf{b}$ or action $a_\msf{Ex}$ with positive probability, then the strategy $\s_\A^{(\formNF_2,E)}(\sigma_\A)$ does not dominate the valuation $v^u_{(\formNF_2,E)}$. However, if it does not play these actions with positive probability, the strategy $\s_\A^{(\formNF_2,E)}(\sigma_\A)$ has value 0 in $\MarVal{\G_{(\formNF_2,E)}}$ from $\qinit$. Indeed, by positionally  playing action $b_\msf{l}$, Player $\B$ ensures that: 1) surely, no stopping state is seen, and 2) almost surely, the state $k_1$ (of color 1) is seen infinitely often, while $c = 0$.
\end{proof}

Let us now consider the proof of Proposition~\ref{prop:inclustion_parm}. 
\begin{proof}
	Let $n \geq 1$. The inclusion $\msf{ParO}(n-1) \supseteq \msf{ParO}(n)$ is direct. Lemma~\ref{lem:gf_distinct} ensures that these two sets are not equal when $n$ is even. A similar lemma could be stated in the case where $n$ is odd.
\end{proof}

\bibliographystyle{plain}
\bibliography{biblio}

\appendices

\newpage
\onecolumn
\section{Explanations on Table~\ref{tab:intro} from the introduction}
\label{sec:explain_table}
Let us describe exactly what the rows and columns of the table refer to. First, the rows correspond to the objectives considered in a finite concurrent arena: 
\begin{itemize}
	\item Safety: Player $\A$ wants to avoid a target (i.e. a set of states);
	\item Reach: stands for reachability, Player $\A$ wants to reach a target;
	\item Büchi: Player $\A$ wants to reach a target infinitely often, it corresponds to a parity objective with colors $1$ and $2$;
	\item co-B: stands for co-Büchi, Player $\A$ wants to reach a target only finitely often, it corresponds to a parity objective with colors $0$ and $1$;
	\item Parity: arbitrary parity objective.
\end{itemize}
Consider now the columns. 
\begin{itemize}
	\item The first column on the left, $\exists$ opt strat?, tells if there always exist optimal strategies in finite games for these objectives. This holds for the safety objective, but it does not for reachability, therefore it does not either for Büchi, co-Büchi and parity as they are more involved than the reachability objective.
	\item The second column, opt, tells the simpler kind of strategies that allow to play optimally, assuming it is possible. In any cell, pos stands for positional whereas $\infty$ stands for infinite memory, i.e. strategies that cannot be described with a finite automaton.
	\item The third column, $\varepsilon$-opt, tells the simpler kind of strategies that allow to play $\varepsilon$-optimally. Note that such strategies always exist by definition of the value of the game.
	\item Finally, the forth column, SubG-opt, tells the simpler kind of strategies that allow to play subgame optimally, assuming it is possible. A strategy is subgame optimal if all of its residual strategies are optimal. (See e.g. \cite{BBSFoSSaCs23} for more details.) Note that being subgame is stronger than being optimal, however both notions coincide on positional strategies. Note that the need for infinite-memory strategies for arbitrary parity objectives already occurs with 3 colors.
\end{itemize}

\section{Algorithms}
\begin{figure}
	\begin{algorithmic}[1]
		\Procedure{CreateEnv}{$k,q,\colVirt$}
		\If{$k$ is even}
		\State return $E_{q,\colVirt}^k$
		\EndIf
		\If{$k$ is odd}
		\State $m \gets \max(k-2,0)$
		\State return $E_{q,\colVirt}^m$
		\EndIf
		\EndProcedure
	\end{algorithmic}
	\caption{\textsf{Build an environment from a state, a color and a coloring function}}
	\label{fig:CreateEnv}
\end{figure}

\begin{figure}
	\begin{algorithmic}[1]
		\Procedure{UpdCurSta}{$k,\colVirt,\msf{Ev}$}
		\For{$q \in Q_u$}
		\If{$\colVirt(q) = k$}
		\State $\msf{Ev}[q] \gets \msf{CreateEnv}(k,q,\colVirt)$
		\EndIf
		\EndFor
		\State \textsf{return $\msf{Ev}$}
		\EndProcedure
	\end{algorithmic}
	\caption{\textsf{Updates the Environment for the states of the current color}}
	\label{fig:UpdCurSta}
\end{figure}

\begin{figure}
	\begin{algorithmic}[1]
		\Procedure{UpdNewSta}{$k,\colVirt,\msf{Ev}'$}
		\State $\msf{change} \gets \msf{True}$
		\While{\textsf{change}}
		\State $\msf{change} \gets \msf{False}$
		\For{$q \in Q_u$}
		\If{
			$\newC(q,\colVirt) = k$}
		\State $\msf{Ev}'[q] \gets \msf{CreateEnv}(k,q,\colVirt)$; $\colVirt[q] \gets k$
		\State $\msf{change} \gets \msf{True}$; \textsf{break};
		\EndIf
		\EndFor
		\EndWhile
		\State \textsf{return $(\colVirt,\msf{Ev}')$}
		\EndProcedure
	\end{algorithmic}
	\caption{\textsf{Updates the Environment and Coloring Functions for new states of a specific color}}
	\label{fig:UpdNewSta}
\end{figure}

\begin{figure}
	\begin{algorithmic}[1]
		\Procedure{UpdateColEnv}{$k,\colVirt,\msf{Ev}'$}
		\State $\msf{Ev}' \gets \msf{UpdCurSta}(k,\colVirt,\msf{Ev}')$
		\State $\msf{UpdNewSta}(k,\colVirt,\msf{Ev}')$
		\EndProcedure
	\end{algorithmic}
	\caption{\textsf{Updates the Environment and Coloring Functions for a specific color}}
	\label{fig:UpdateColEnv}
\end{figure}

\begin{figure}
	\begin{algorithmic}[1]
		\Procedure{IncLeast}{$\msf{col}',\msf{Ev}'$}
		\State $c_{\msf{min}} \gets \min \colVirt$
		\For{$q \in Q_u$}
		\If{$\colVirt(q) = c_\msf{min}$}
		\State $\colVirt(q) \gets c_\msf{min}+2$;
		\EndIf
		\If{$\colVirt(q) = c_\msf{min}+1$}
		\State $\colVirt(q) \gets \colFunc(q)$;
		\State  $\msf{Ev}'[q] \gets \msf{NoEnv}$;
		\EndIf
		\EndFor
		\State $\msf{UpdNewSta}(c_\msf{min}+2,\colVirt,\msf{Ev}')$
		\EndProcedure
	\end{algorithmic}
	\caption{\textsf{Increases the least color, resets the last but least layer}}
	\label{fig:IncLeast}
\end{figure}


\begin{figure}
	\begin{algorithmic}[1]
		\Procedure{ComputeEnv}{$u,\colFunc$}
		\State $\colVirt \gets \colFunc$
		\State $\mathsf{Ev} \gets \mathsf{EmptyEnv}$
		\For{$k = e \textsf{ down to } 0$}
		\State $(\colVirt,\msf{Ev}) \gets \msf{UpdateColEnv}(k,\colVirt,\msf{Ev})$
		\EndFor
		\While{$\colVirt[Q_u] \neq \{ e \}$}
		\State $(\colVirt,\msf{Ev}) \gets \msf{IncLeast}(\colVirt,\msf{Ev})$
		\For{$k = e \textsf{ down to } 0$}
		\State $(\colVirt,\msf{Ev}) \gets \msf{UpdateColEnv}(k,\colVirt,\msf{Ev})$
		\EndFor
		\EndWhile
		\State \textsf{return }$(\colVirt,\msf{Ev})$
		\EndProcedure
	\end{algorithmic}
	\caption{\textsf{Computes an environment function guaranteeing $u$ for Player $\A$ in $Q_u$}}
	\label{fig:ComputeEnv}
\end{figure}

\section{Complements on strategies and valuations}

\subsection{How to implement stopping states}
\label{appen:implement_stopping_states}
This can be done in a straightforward manner: given a stopping state of value $u \in [0,1]$, we consider a probability distribution that goes with probability $u$ to  a (self-looping) state $q_\A$ of value 1 and that goes to a (self-looping) state $q_\B$ of value 0 with probability $1 - u$. Furthermore, note that the two corresponding BSCCs $\{ q_\A \}$ and $\{ q_\B \}$ (which exists regardless of the strategies of the players) behave properly w.r.t. the conditions for parity dominating a valuation (again, regardless of the strategies). 

\subsection{Markov chain induced by a pair of strategies}
\label{appen:markov_chain}
\begin{definition}[Markov chain induced by a pair of positional strategies]
	Consider a concurrent arena $\Aconc$ and two positional strategies $\s_\A$ (for Player $\A$) and $\s_\B$ (for Player $\B$). The Markov chain $\mathcal{M}_{\s_\A,\s_\B} = \langle Q,\mathbb{P}^{\s_\A,\s_\B} \rangle$ \emph{induced} by $\s_\A,\s_\B$ is such that, for all $q \in Q$
	:
	\begin{displaymath}
	\forall q' \in Q,\; \mathbb{P}^{\s_\A,\s_\B}(q)(q') := \outM_{\Games{\fsf(q)}{\mathbbm{1}_{q'}}}(\s_\A(q),\s_\B(q))
	\end{displaymath}
\end{definition}

\subsection{Proof of Theorem~\ref{thm:strongly_domnates_imply_guarantee}}
\label{appen:proof_thm_even_in_bscc_ok}
This theorem is straightforwardly deduced from already existing work, specifically, \cite{BBSCSL22} and \cite{BBSArXivICALP}. In particular, in this paper, results are stated by using the notion of MDPs and End Component. Informally, MDPs are games where only one Player plays. In particular, one obtains an MDP from a concurrent game by fixing the (positional) strategy of a player. An End Component $H$ (EC for short) in an MDP is a set of states $Q_H \subseteq Q$ such that one can ensure to stay in $Q'$ while the underlying graph of $H$ if strongly connected. In particular, BSCCs are special cases of end components.
\begin{proof*}
	The first part of the theorem is given by \cite[Proposition 18]{BBSCSL22}. As for the second part, we apply \cite[Lemma 17]{BBSArXivICALP}. Indeed, the first condition is ensured by assumption. As for the second, consider a End Component $H$ with $v_H > 0$. This can be seen as a finite MDP --- where only Player $\B$ plays --- with a parity objective. From \cite{DBLP:conf/soda/ChatterjeeJH04}, Player $\B$ has an optimal strategy $\s_\B$ that is positional and deterministic. Then, we obtain a Markov chain $H_{\s_\B}$ where all BSCCs are compatible with the strategy $\s_\A$ (and the minimum of their values w.r.t. the valuation $v$ is $v_H > 0$). Hence, by assumption, all these BSCCs are even-colored, that is they have value 1. It follows that all states in $H_{\s_\B}$ have value 1. Hence, the second condition of \cite[Lemma 17]{BBSArXivICALP} is ensured, and this lemma can be applied. 
\end{proof*}

\subsection{Proof of Lemma~\ref{lem:opt_in_gf_equiv_opt_in_parity}}
\label{appen:proof_prop_opt_in_gf_equiv_opt_in_parity}

\begin{proof*}
	The first equivalence comes from \cite[Lemma 16]{BBSFSTTCS22} and the fact that positional deterministic strategies are enough (for Player $\B$) to play optimally in finite MDPs with parity objectives \cite{DBLP:conf/soda/ChatterjeeJH04}. Now assume that $u > 0$. Assume that the strategy $\s_{Y}^\A(\sigma_\A)$ parity dominates 
	%
	%
	the valuation $v_Y^u$ in the game $\G_{Y}$. In particular, it dominates this valuation, i.e. item $(ii.1)$ of Lemma~\ref{lem:opt_in_gf_equiv_opt_in_parity} is ensured. Consider now some action $b \in \Act_\B$ such that $\outM_{\Games{\formNF}{p(\cdot)[p_{[0,1]}]}}(\sigma_\A,\s_\B(q_0)) = 0$ and a Player-$\B$ positional and deterministic strategy $\s_\B$ such that $\s_\B(q_0) := b$. 
	Since we have $\outM_{\Games{\formNF}{p(\cdot)[p_{[0,1]}]}}(\sigma_\A,\s_\B(q_0)) = 0$, no stopping state can be reached under $\s_{Y}^\A(\sigma_\A)$ and $\s_\B$. Consider the Markov chain $\mathcal{M}_{\s_{Y}^\A(\sigma_\A),\s_\B}$. Besides stopping states, it is reduced to a BSCC $H$ whose states are $q_{\msf{init}}$ and all states $k_i$ reachable with that action $b$. That is, $H = \{ q_{\msf{init}} \} \cup \{ k_i \mid i \in \brd,\; \outM_{\Games{\formNF}{p(\cdot)(i)}}(\sigma_\A,b) > 0 \}$. Furthermore, $q_{\msf{init}}$ is colored with $c$. The value of this BSCC $H$ is in $\{ 0,1 \}$ and it is equal to $1$ if and only if $H$ is even-colored (since every state in a BSCC is almost-surely seen infinitely often). Since $u > 0$, the value of $H$ cannot be 0, thus $H$ is even-colored and the maximum of the colors seen with that action is even. This exactly corresponds to item $(ii.2)$ of Lemma~\ref{lem:opt_in_gf_equiv_opt_in_parity}.
	
	Assume now that the $\GF$-strategy $\sigma_\A$ satisfies both items $(ii.1)$ and $(ii.2)$. Consider a positional deterministic Player-$\B$ strategy $\s_\B$ in the game $\G_Y$. Let $b := \s_\B(q_0) \in \Act_B$. Consider a BSCC $H$ in the induced Markov chain $\mathcal{M}_{\s_{Y}^\A(\sigma_\A),\s_\B}$ that is not reduced to a stopping state. In particular, this implies that $\outM_{\Games{\formNF}{p(\cdot)[p_{[0,1]}]}}(\sigma_\A,b) = 0$. Then, as previously, we have $H = \{ q_{\msf{init}} \} \cup \{ k_i \mid i \in \brd,\; \outM_{\Games{\formNF}{p(\cdot)(i)}}(\sigma_\A,b) > 0 \}$. The fact that $\sigma_\A$ satisfies item $(ii.2)$ ensures that the BSCC $H$ is even-colored. It follows that the strategy $\s_{Y}^\A(\sigma_\A)$ parity dominates the valuation $v_Y^u$.

\end{proof*}


\subsection{Proof of Lemma~\ref{lem:optimal_in_all_value_slice_optimal_everywhere}}
\label{appen:proof_prop_optimal_in_all_value_slice_optimal_everywhere}
\begin{proof*}
	Consider such a strategy $\s_\A$. Let us show that it parity dominates the valuation $\MarVal{\G}$. First, note that it dominates the valuation $\MarVal{\G}$ since, for all $u \in V_\G \setminus \{ 0 \}$,  the strategy $\s_\A^u$ dominates $\MarVal{\G}$ in $\G^u$ (recall that the stopping states in $\G^u$ have the values of the original states in $\G$). Consider now a BSCC $H$ compatible with $\s_\A$ such that $\min \MarVal{\G}[H] > 0$. By Theorem~\ref{thm:strongly_domnates_imply_guarantee}, there is a value $u_H \in (0,1]$ such that $\MarVal{\G}[H] = \{ u_H \}$. That is, $H \subseteq Q_{u_H}$. It follows that $H$ is compatible with $\s_\A^{u_H}$. Since $\s_\A^{u_H}$ parity dominates the valuation $\MarVal{\G}$ in $\G^{u_H}$, we can deduce that $H$ is even-colored. Overall, the strategy $\s_\A$ parity dominates the valuation $\MarVal{\G}$. By Theorem~\ref{thm:strongly_domnates_imply_guarantee}, the strategy $\s_\A$ guarantees it (i.e. it is optimal).
\end{proof*}

\section{Proof of Theorem~\ref{thm:transfer_local_global}}	
\label{appen:property_of_the_update_function}
In this section, we give all the technical details used to prove Theorem~\ref{thm:transfer_local_global}. The proof of this theorem is given in Page~\pageref{proof:mainthm}, provided that Lemmas~\ref{lem:algo_preserves_faithfulness},~\ref{lem:algo_increase_least_color},~\ref{lem:minimum_faithful_cannot_be_d_minus_one} and~\ref{lem:metric_work} hold. We already stated and argued why these lemmas hold, however we have not given a detailed proof of them. This is what we do in this section. However, we first state and prove results that we use in this section.

We first state and prove properties on the update of colors $\newC$. This is done in Page~\pageref{subsubsec:local_operator}. Then, we give the complete definition of faithfulness, see Page~\pageref{subsubsec:faithfulness}. We also state and prove three lemmas, that relate local and global behaviors, see Page~\pageref{subsubsec:three_lemmas}. We finally prove the lemmas mentioned above: we prove Lemma~\ref{lem:algo_preserves_faithfulness} in Page~\pageref{subsubsec:proof_lem_algo_preserves_faithfulness}, Lemma~\ref{lem:algo_increase_least_color} in Page~\pageref{subsubsec:proof_lem_algo_increase_least_color}, Lemma~\ref{lem:minimum_faithful_cannot_be_d_minus_one} in Page~\pageref{subsubsec:proof_lem_minimum_faithful_cannot_be_d_minus_one} and Lemma~\ref{lem:metric_work} in Page~\pageref{subsubsec:proof_lem_metric_work}.

\subsection{
	Properties ensured by the local operator $\newC$}
\label{subsubsec:local_operator}
We introduce a useful notation which will allow us to rewrite Lemma~\ref{lem:opt_in_gf_equiv_opt_in_parity} in our context. 
\begin{definition}
	\label{def:notation_max_color_local_BSCC}
	Consider a state $q \in Q$, a virtual coloring function $\colVirt: Q_u \rightarrow \brd$, a color $n \in \brd$ and two $\GF$-strategies $(\sigma_\A,\sigma_\B) \in \Sigma_\A(\fsf(q)) \times \Sigma_\B(\fsf(q))$.
	We denote by $\Col(q,\colVirt,\sigma_\A,\sigma_B)$ the set of colors reachable in one step with positive probability w.r.t. $(\sigma_\A,\sigma_\B)$ in the local game $\G_{q,\colVirt}^n$ (regardless of the color $n \in \brd$ considered). That is:
	\begin{displaymath}
	\Col(q,\colVirt,\sigma_\A,\sigma_B) := \{ i \in \brd \mid  \outM_{\Games{\fsf(q)}{\mathbbm{1}_{p_{q,\colVirt}^{-1}[k_i]}}}(\sigma_\A,\sigma_\B) > 0 \}
	\end{displaymath}
	In particular, $\Col(q,\colVirt,\sigma_\A,b) = \msf{Color}(\fsf(q),p_{q,\colVirt},\sigma_\A,b)$ for all $b \in \Act_\B$ (notation from Lemma~\ref{lem:opt_in_gf_equiv_opt_in_parity}).
	
	Then, the set of colors $\msf{ColBSCC}(q,\colVirt,n,\sigma_\A,\sigma_B)$ 
	is defined by:
	\begin{displaymath}
	\msf{ColBSCC}(q,\colVirt,n,\sigma_\A,\sigma_B) := \Col(q,\colVirt,\sigma_\A,\sigma_B) \cup \{ \max(\colVirt(q),c_n) \}
	\end{displaymath}
	with $c_n := n-1$ is $n$ is even and $c_n := n+1$ is $n$ is odd.
\end{definition}

We obtain a corollary of Lemma~\ref{lem:opt_in_gf_equiv_opt_in_parity} (which only consists in writing Lemma~\ref{lem:opt_in_gf_equiv_opt_in_parity} in the context of a local game $\G_{q,\colVirt}^n$).
\begin{corollary}
	\label{coro:opt_in_gf_equiv_opt_in_parity}
	Consider a state $q \in Q$, a virtual coloring function $\colVirt: Q_u \rightarrow \brd$, a color $n \in \brd$ and a local strategy $\sigma_\A \in \Sigma_\A(\fsf(q))$. 
	A Player-$\A$ $\GF$-strategy $\sigma_\A \in \Sigma_\A(\formNF)$ is
	\emph{optimal} w.r.t. $Y := (\fsf(q),E_{q,\colVirt}^n)$ if and only if, letting
	$u := \MarVal{\G_Y}(q_\msf{init})$, either (i) $u = 0$, or (ii)
	the positional Player-$\A$ strategy $\s^Y_\A(\sigma_\A)$ parity
	dominates the valuation $v_Y^u$. 
	
	Furthermore (ii) is equivalent to: 
	\begin{itemize}
		\item (1): the Player-$\A$ positional strategy $\s^Y_\A(\sigma_\A)$ dominates the valuation $v_Y^u$; and
		\item (2): for all $b \in \Act_\B$, if $\outM_{\Games{\formNF}{\mathbbm{1}_{p^{-1}[0,1]}}}(\sigma_\A,b) = 0$, then $\max \msf{ColBSCC}(q,\colVirt,n,\sigma_\A,\sigma_B)$ is even. 
	\end{itemize}
	
	This is symmetrical for Player $\B$.
\end{corollary}

Let us now state a proposition we will use to prove that the update of colors cannot decrease the colors of the states.
\begin{proposition}
	\label{prop:dominating_strat_in_gf_other_def}
	%
	Let $q \in Q_u$ 
	and some virtual coloring function $\colVirt: Q_u \rightarrow \brd$, some color $n \in \brd$ and a positive value $z \in (0,1]$. Let $p := p_{q,\colVirt}$ and $Y := (\fsf(q),E_{q,\colVirt}^n)$. For all Player-$\A$ $\GF$-strategies $\sigma_\A \in \Sigma_\A(\fsf(q))$
	, the positional Player-$\A$ strategy $\s_Y^\A(\sigma_\A)$ dominates the valuation $v_Y^z$ if and only if for all $b \in \Act_\B$: if $\outM_{\Games{\fsf(q)}{\mathbbm{1}_{Q \setminus Q_u}}}(\sigma_\A,b) > 0$, then: 
	\begin{equation*}
	\outM_{\Games{\fsf(q)}{\mu^z}}(\sigma_\A,b) \geq 0
	\end{equation*}
	where $\mu^z: Q \rightarrow [0,1]$ is such that, for all $q \in Q$:
	\begin{equation*}
	\mu^z(q) := \begin{cases}
	0 & \text{ if }q \in Q_u \\
	\MarVal{\G}(q) - z & \text{ otherwise} 
	\end{cases}
	\end{equation*}
\end{proposition}
\begin{proof*}
	In the game $\G_Y$, the Player-$\A$ strategy $\s_Y^\A(\sigma_\A)$ dominates the valuation $v_Y^z$ --- property we denote $(1)$ --- if and only if:
	\begin{align*}
	z = v_Y^z(q_\msf{init}) & \leq \va{\Games{\fsf(q)}{v_Y^z}}(\s_Y^\A(\sigma_\A)(q_\msf{init})) \\
	& = \va{\Games{\fsf(q)}{v_Y^z}}(\sigma_\A)
	\end{align*}
	In addition, $\va{\Games{\fsf(q)}{v_Y^z}}(\sigma_\A) = \min_{b \in \Act_\B} \outM_{\Games{\fsf(q)}{v_Y^z}}(\sigma_\A,b)$. Furthermore, for all $b \in \Act_\B$, we have that
	\begin{align*}
	\outM_{\Games{\fsf(q)}{v_Y^z}}(\sigma_\A,b) & = \outM_{\Games{\fsf(q)}{\mathbbm{1}_{p^{-1}[\{ q_\msf{init} \} \cup K_e]}}}(\sigma_\A,b) \cdot z \\
	& + \sum_{x \in V_{Q \setminus Q_u}} \outM_{\Games{\fsf(q)}{\mathbbm{1}_{p^{-1}[x]}}}(\sigma_\A,b) \cdot x
	\end{align*}
	It follows that, if $\outM_{\Games{\fsf(q)}{\mathbbm{1}_{Q \setminus Q_u}}}(\sigma_\A,b) = 0$, we have $\outM_{\Games{\fsf(q)}{v_Y^z}}(\sigma_\A,b) = z$ since $\cup_{x \in V_{Q \setminus Q_u}} p^{-1}[x] = Q \setminus Q_u$. However, if $\outM_{\Games{\fsf(q)}{\mathbbm{1}_{Q \setminus Q_u}}}(\sigma_\A,b) > 0$, we have that 
	\begin{align*}
	z \leq \outM_{\Games{\fsf(q)}{v_Y^z}}(\sigma_\A,b) \Leftrightarrow 0 & \leq
	\sum_{x \in V_{Q \setminus Q_u}} \outM_{\Games{\fsf(q)}{\mathbbm{1}_{p^{-1}[x]}}}(\sigma_\A,b) \cdot ( x - z) \\
	0 & \leq
	\outM_{\Games{\fsf(q)}{\sum_{x \in V_{Q \setminus Q_u}} \mathbbm{1}_{p^{-1}[x]} \cdot ( x - z)}}(\sigma_\A,b) \\
	0 & \leq \outM_{\Games{\fsf(q)}{\mu^z}}(\sigma_\A,b)
	\end{align*}
	The result follows.
\end{proof*}

{\bf The update of colors does not decrease the color.}
\label{appen:prop_update_color_at_least_preivous_color}
Let us show that the local operator $\newC$ does not decrease (w.r.t. the usual order $<$ on natural numbers) the previous color of the state $q$ given by a (virtual) coloring function $\colVirt$.
\begin{proposition}
	\label{prop:update_color_at_least_preivous_color}
	Consider a state $q \in Q_u$ and a coloring function $\colVirt: Q_u \rightarrow \llbracket 0,e \rrbracket$. We have $\newC(q,\colVirt) \geq \colVirt(q)$.
\end{proposition}

\begin{proof*}
	We let 
	$p := p_{q,\colVirt}$ 
	and, for all $n \in \brd$, we let $Y_n := (\fsf(q),E_{q,\colVirt}^n)$. 
	Note that $p_{[0,1]} \subseteq V_{Q \setminus Q_u}$. 
	
	There are two cases: either $\colVirt(q) = e$ or $\colVirt(q) < e$. First, assume that $\colVirt(q) = e$. Let us show that in that case $\newC(q,\colVirt) = e$. Assume towards a contradiction that $\MarVal{\G_{Y_e}}(q_{\msf{init}}) = u' < u$ for some $u' \in [0,1]$. Consider a Player-$\B$ $\GF$-strategy $\sigma_\B$ that is optimal w.r.t. $Y_e
	$. 
	For all $a \in \Act_\A$, we have $
	\outM_{\Games{\formNF}{\mathbbm{1}_{Q \setminus Q_u}}}(a,\sigma_\B) > 0$. Indeed, otherwise, in the game $\G_{q,\colVirt}^e$ where Player $\B$ plays the strategy defined by $\sigma_\B$, Player $\A$ could loop indefinitely on $q_\msf{init}$ thus ensuring winning with probability 1 (since the color of the state $q_\msf{init}$ is $e$ --- as $\colVirt(q) = e$ --- which is both the highest color appearing in the game and even). We let:
	\begin{equation*}
		p_{\msf{Exit}} := \min_{a \in \Act_\A} \outM_{\Games{\formNF}{\mathbbm{1}_{Q \setminus Q_u}}}(a,\sigma_\B) > 0
	\end{equation*}
	
	Now, consider some $a \in \Act_\A$. By Proposition~\ref{prop:dominating_strat_in_gf_other_def} for Player $\B$, we have:
	$\outM_{\Games{\formNF}{\mu^{u'}}}(a,\sigma_\B) \leq 0$ where $\mu^{u'}: Q \rightarrow [0,1]$ comes from Proposition~\ref{prop:dominating_strat_in_gf_other_def}. Hence, we have:
	\begin{displaymath}
	\sum_{x \in V_{Q \setminus Q_u}} \outM_{\Games{\formNF}{\mathbbm{1}_{Q_x}}}(a,\sigma_\B) \cdot x \leq \sum_{x \in V_{Q \setminus Q_u}} \outM_{\Games{\formNF}{\mathbbm{1}_{Q_x}}}(a,\sigma_\B) \cdot u'
	\end{displaymath}
	Letting $v := \MarVal{\G}$ be the value vector in the game $\G$, we have:
	\begin{align*}
	\sum_{q' \in Q} \prob{q,q'}{}{}(a,\sigma_\B) \cdot v(q') & = \sum_{q' \in Q} \outM_{\Games{\formNF}{\mathbbm{1}_{q'}}}(a,\sigma_\B) \cdot v(q')  \\
	& = \outM_{\Games{\formNF}{\mathbbm{1}_{Q_u}}}(a,\sigma_\B) \cdot u + \sum_{x \in V_{Q \setminus Q_u}} \outM_{\Games{\formNF}{\mathbbm{1}_{Q_x}}}(a,\sigma_\B) \cdot x  \\
	& \leq \outM_{\Games{\formNF}{\mathbbm{1}_{Q_u}}}(a,\sigma_\B) \cdot u + \sum_{x \in V_{Q \setminus Q_u}} \outM_{\Games{\formNF}{\mathbbm{1}_{Q_x}}}(a,\sigma_\B) \cdot u'  \\
	& = \outM_{\Games{\formNF}{\mathbbm{1}_{Q_u}}}(a,\sigma_\B) \cdot u + \outM_{\Games{\formNF}{\mathbbm{1}_{Q \setminus Q_u}}}(a,\sigma_\B) \cdot u'  \\
	& = u - \outM_{\Games{\formNF}{\mathbbm{1}_{Q \setminus Q_u}}}(a,\sigma_\B) 
	\cdot (u - u') \\
	& \leq u - p_{\msf{Exit}}\cdot (u - u') < u = v(q)
	\end{align*}
	Hence, letting $\delta := p_{\msf{Exit}}\cdot (u - u')/2 > 0$ and considering, in the original game $\G$, a Player-$\B$ strategy $\s_\B \in \msf{S}_\B^\Aconc$ such that $\s_\B(q) := \sigma_\B$ and for all $q' \in Q$, we have $\s_\B(q \cdot q')$ a $\delta$-optimal Player-$\B$ strategy from $q'$, it follows that, for all Player-$\A$ strategy $\s_\A$ in the game $\G$, we have: 
	\begin{align*}
	\prob{\Aconc,q}{\s_\A,\s_\B}{}[W] & = \sum_{q' \in Q} \prob{q,q'}{}{}(\s_\A(q),\sigma_\B) \cdot \prob{\Aconc,q'}{\s_\A^q,\s_\B^q}{}[W] \\
	& \leq \sum_{q' \in Q} \prob{q,q'}{}{}(\s_\A(q),\sigma_\B) \cdot (v(q') + \delta) \leq v(q) - 2 \delta + \delta = v(q) - \delta
	\end{align*}
	where $W := (\colFunc^{\omega})^{-1}[W(\colFunc)] \subseteq Q^\omega$. 
	Thus, the value from $q$ is less than $u = v(q)$
	. Hence the contradiction. In fact, $\MarVal{\G_{q,\colVirt}^e}(q_{\msf{init}}) \geq u$ and $\newC(q,\colVirt) = e$.
	
	Consider now the case where $\colVirt(q) < e$. Assume towards a contradiction that $\newC(q,\colVirt) < \colVirt(q)$. Let $n := \newC(q,\colVirt)$. 
	
	Assume that $n$ is even. 
	By assumption, we have $\MarVal{\G_{Y_n}}(q_{\msf{init}}) \geq u$. Consider a Player-$\A$ $\GF$-strategy $\sigma_\A$ that is optimal w.r.t. $Y_n$. First, the positional Player-$\A$ strategy $\s_\A := \s_{Y_n}^A(\sigma_\A) = \s_{Y_{n+2}}^A(\sigma_\A)$ defined by $\sigma_\A$ dominates the valuation $v_{Y_n}^u = v_{Y_{n+2}}^u$ in the game $\G_{Y_n}$ and it also does in the game $\G_{Y_{n+2}}$. Consider now an action $b \in \Act_\B$ and assume that
	$\outM_{\Games{\formNF}{\mathbbm{1}_{p^{-1}[0,1]}}}(\sigma_\A,b) = 0$. 	 
	By Corollary~\ref{coro:opt_in_gf_equiv_opt_in_parity}, it follows that, in the game $\G_{Y_n}$, we have $\max \msf{ColBSCC}(q,\colVirt,n,\sigma_\A,b)$ even with $\msf{ColBSCC}(q,\colVirt,n,\sigma_\A,b) = \msf{Col}(q,\colVirt,\sigma_\A,b) \cup \{ e_n \}
	$ with $e_n = \max(\colVirt(q),c_n)$ and $c_n = n-1$ since $n$ is even. Since $n < \colVirt(q)$, we have $e_n = \max(\colVirt(q),c_n) = \colVirt(q)$. Furthermore, $e_{n+2} = \max(\colVirt(q),c_{n+2})$ with $c_{n+2} = n+1$ since $n$ is even. Since $n < \colVirt(q)$, we also have $e_{n+2} = \colVirt(q) = e_n$. That is, $\msf{ColBSCC}(q,\colVirt,n,\sigma_\A,b) = \msf{ColBSCC}(q,\colVirt,n+2,\sigma_\A,b)$. Hence, $\msf{ColBSCC}(q,\colVirt,n+2,\sigma_\A,b)$ is also even. By Corollary~\ref{coro:opt_in_gf_equiv_opt_in_parity}, 
	the Player-$\A$ $\GF$-strategy $\sigma_\A$ strongly dominates the valuation $v_{Y_{n+2}}^u$. Hence, by Proposition~\ref{thm:strongly_domnates_imply_guarantee}, the Player-$\A$ positional strategy $\s_{Y_{n+2}}^\A(\sigma_\A)$ guarantees the valuation $v_{Y_{n+2}}^u$. Hence the contradiction since this implies $\newC(q,\colVirt) \geq n+2$.
	
	When $n$ is odd, the reasoning is symmetrical, by taking the point-of-view of Player $\B$: we compare what happens in the games $\G_{Y_m}$ and $\G_{Y_m}$ for $m := \ScPar(n)$.
\end{proof*}

{\bf The update of colors is not affected by small changes of colors.}
\label{appen:prop_update_color_does_not_change_if_small_color_change}	
Let us now tackle another property ensured by the local operator $\newC$. 
Assume that the new color of a state $q$ is $n \in \brd$ w.r.t. a coloring function $\colVirt$. Now, consider another coloring function $\colVirt'$ that coincide with $\colVirt$ on colors at least $n$, and may differ for smaller colors. In that case, the new color of $q$ will still be $n$ w.r.t. the coloring function $\colVirt'$. This is (almost) what we prove here. First, 
we introduce the notion of equivalent and prevailing coloring functions.
\begin{definition}[Equivalent and Prevailing coloring functions]
	\label{def:k_similar_coloring_function}
	Consider two coloring functions $\colVirt,\colVirt': Q_u \rightarrow \brck{0,e}$ and some color $n \in \brd$. The coloring functions $\colVirt,\colVirt'$ are \emph{equivalent down to }$n$ if, for all $k \in \brck{n,e}$, 
	we have 
	$\colVirt^{-1}[k] = \colVirt'^{-1}[k]$.
	
	Furthermore, $\colVirt'$ is said to be $(n-1)$-\emph{prevailing compared to} $\colVirt$ if $\colVirt$ and $\colVirt'$ are equivalent down to $n$ and $\colVirt^{-1}[n-1] \subseteq \colVirt'^{-1}[n-1]$.
\end{definition}

Let us first state a lemma that we will use to prove the proposition (of interest for us) that we state below.
\begin{lemma}
	\label{lem:similarity_color}
	Consider a state $q \in Q_u$ and some color $n 
	\in \llbracket 0,e \rrbracket$. Let $\colVirt,\colVirt': Q_u \rightarrow \llbracket 0,e \rrbracket$ be two virtual coloring functions with $\colVirt'$ $n$-prevailing compared to $\colVirt$. Consider a pair of $\GF$-strategies $(\sigma_\A,\sigma_\B) \in \Sigma_\A(\fsf(q)) \times \Sigma_\B(\fsf(q))$. We have:
	\begin{align*}
	& \max \ColB(q,\colVirt,n,\sigma_\A,\sigma_\B) \geq n \Rightarrow \\
	& \max \ColB(q,\colVirt,n,\sigma_\A,\sigma_\B) = \max \ColB(q,\colVirt',n,\sigma_\A,\sigma_\B)
	\end{align*}
\end{lemma}
\begin{proof*}
	Letting $c_n := n-1$ if $n$ is even and $c_n := n+1$ if $n$ is odd, we have (recall Definition~\ref{def:notation_max_color_local_BSCC}):
	\begin{displaymath}
	\ColB(q,\colVirt,n,\sigma_\A,\sigma_\B) = \Col(q,\colVirt,\sigma_\A,\sigma_\B) \cup \{ \max(\colVirt(q),c_n) \}
	\end{displaymath}
	and
	\begin{displaymath}
	\ColB(q,\colVirt',n,\sigma_\A,\sigma_\B) = \Col(q,\colVirt',\sigma_\A,\sigma_\B) \cup \{ \max(\colVirt'(q),c_n) \}
	\end{displaymath}
	
	By assumption
	, we have that, for all $i \in \brck{n+1,e}$:
	\begin{displaymath}
	\colVirt^{-1}[i] = \colVirt'^{-1}[i]
	\end{displaymath}
	Thus:
	\begin{displaymath}
	\outM_{\Games{\fsf(q)}{\mathbbm{1}_{p_{q,\colVirt}^{-1}[k_i]}}}(\sigma_\A,\sigma_\B) = \outM_{\Games{\fsf(q)}{\mathbbm{1}_{p_{q,\colVirt}^{-1}[k_i]}}}(\sigma_\A,\sigma_\B)
	\end{displaymath}
	Furthermore:
	\begin{displaymath}
	\colVirt^{-1}[n] \subseteq \colVirt'^{-1}[n]
	\end{displaymath}
	Thus:
	\begin{displaymath}
	\outM_{\Games{\fsf(q)}{\mathbbm{1}_{p_{q,\colVirt}^{-1}[k_n]}}}(\sigma_\A,\sigma_\B) > 0 \Rightarrow \outM_{\Games{\fsf(q)}{\mathbbm{1}_{p_{q,\colVirt'}^{-1}[k_n]}}}(\sigma_\A,\sigma_\B) > 0
	\end{displaymath}
	
	We therefore obtain that for any $k \in \brck{n,e}$:
	\begin{itemize}
		\item[(i)] If $\max \Col(q,\colVirt,\sigma_\A,\sigma_\B) \leq k$, then $\max \Col(q,\colVirt',\sigma_\A,\sigma_\B) \leq k$;
		\item[(ii)] If $\max \Col(q,\colVirt,\sigma_\A,\sigma_\B) = k$, then $\max \Col(q,\colVirt',\sigma_\A,\sigma_\B) = k$.
	\end{itemize}
	Furthermore:
	\begin{itemize}
		\item[(i')] If $\colVirt(q) \leq k$, then $\colVirt'(q) \leq k$;
		\item[(ii')] If $\colVirt(q) = k$, then $\colVirt'(q) = k$.
	\end{itemize}
	We have three cases, letting $n \leq j := \max \ColB(q,\colVirt,n,\sigma_\A,\sigma_\B) = \Col(q,\colVirt,\sigma_\A,\sigma_\B) \cup \{ \max(\colVirt(q),c_n) \}$:
	\begin{itemize}
		\item If $j = c_n$, then $\max \Col(q,\colVirt,\sigma_\A,\sigma_\B),\colVirt(q) \leq j$. Hence, by (i) and (i'), we have $\max \Col(q,\colVirt',\sigma_\A,\sigma_\B),\colVirt'(q) \leq j$. \\ That is, $\max \ColB(q,\colVirt',n,\sigma_\A,\sigma_\B) = j$.
		\item If $j = \colVirt(q)$ then $\max \Col(q,\colVirt,\sigma_\A,\sigma_\B),c_n \leq j$. Hence, by (ii'), we have $j = \colVirt'(q)$ and, by (i), we have $\max \Col(q,\colVirt',\sigma_\A,\sigma_\B) \leq j$. \\ That is, $\max \ColB(q,\colVirt',n,\sigma_\A,\sigma_\B) = j$.
		\item If $j = \max \Col(q,\colVirt,\sigma_\A,\sigma_\B)$ then $\colVirt(q),c_n \leq j$. Hence, by (ii), we have $j = \max \Col(q,\colVirt',\sigma_\A,\sigma_\B)$ and, by (i'), we have $\colVirt'(q) \leq j$. That is, $\max \ColB(q,\colVirt',n,\sigma_\A,\sigma_\B) = j$.
	\end{itemize}
	The lemma follows.
\end{proof*}

\begin{proposition}
	\label{prop:update_color_does_not_change_if_small_color_change}
	Consider a state $q \in Q_u$ and a coloring function $\colVirt: Q_u \rightarrow \llbracket 0,e \rrbracket$. Let $n := \newC(q,\colVirt) \in \llbracket 0,e \rrbracket$. Assume that another coloring function $\colVirt': Q_u \rightarrow \brd$ is $n$-prevailing compared to $\colVirt$. In that case, $\newC(q,\colVirt') = n$.
\end{proposition}	
\begin{proof*}
	Let us consider two such coloring functions $\colVirt$ and $\colVirt'$. We have $n = \newC(q,\colVirt) \in \llbracket 0,e \rrbracket$. Consider a Player-$\A$ $\GF$-strategy $\sigma_\A$ that is optimal w.r.t. $Y$ for $Y := (\fsf(d),E_{q,\colVirt}^n)$. Let $Y' := (\fsf(d),E_{q,\colVirt'}^n)$, $p := p_{q,\colVirt}$ and $p' := p_{q,\colVirt'}$. Let us show that $n \preceq_\msf{par} \newC(q,\colVirt')$. This straightforwardly holds if $n = e-1$. Assume now that $n \neq e-1$. The Player-$\A$ strategy $\s_{Y}^\A(\sigma_\A)$ dominates the valuations $v_{Y}^u = v_{Y'}^u$. 
	
	Consider an action $b \in \Act_\B$ and assume that $\outM_{\Games{\fsf(q)}{\mathbbm{1}_{p^{-1}[0,1]}}}(\sigma_\A,b) = 0$
	. 
	By Corollary~\ref{coro:opt_in_gf_equiv_opt_in_parity}, we have $j := \max \msf{ColBSCC}(q,\colVirt,n,\sigma_\A,b)$ is even with $\max \msf{ColBSCC}(q,\colVirt,n,\sigma_\A,b) \geq \max(\colVirt(q),c_n)$ for $c_n = n-1$ if $n$ is even and $c_n = n+1$ otherwise. Since $n-1$ is odd, it follows that, in any case, $\max \msf{ColBSCC}(q,\colVirt,n,\sigma_\A,b) \geq n$. Then, by Lemma~\ref{lem:similarity_color}, we have $\max \msf{ColBSCC}(q,\colVirt',n,\sigma_\A,b) = \max \msf{ColBSCC}(q,\colVirt,n,\sigma_\A,b)$, which is even. 
	As this holds for all $b \in \Act_\B$, the $\GF$-strategy $\sigma_\A$ strongly dominates the valuation $v_{Y'}^u$, hence, by Theorem~\ref{thm:strongly_domnates_imply_guarantee}, the value of the state $q_\msf{init}$ in the game $\G_{Y'}$ is at least $u$: $\MarVal{\G_{Y'}}[q_\msf{init}] \geq u$. Hence, $n \preceq_\msf{par} \newC(q,\colVirt')$. 
	
	
	Let us now show that $\newC(q,\colVirt') \preceq_\msf{par} n$, which straightforwardly holds if $n = e$. Hence, assume that $n \neq e$ and let $m := \ScPar(n) \in \brd$. The proof is very similar to the previous case. Let $Z := (\fsf(d),E_{q,\colVirt}^m)$. Let also $Z' := (\fsf(d),E_{q,\colVirt'}^m)$. The value of the state $q_\msf{init}$ in the game $\G_{Z}$ is at most $u'$ for some $u' < u$: $\MarVal{\G_{Z}}[q_\msf{init}] \leq u' < u$. Consider a Player-$\B$ $\GF$-strategy $\sigma_\B$ that is optimal w.r.t. $Z$. The Player-$\B$ strategy $\s_{Z}^\B(\sigma_\B)$ dominates the valuation $v_{Z}^{u'} = v_{Z'}^{u'}$. Consider an action $a \in \Act_\A$ and assume that $\outM_{\Games{\fsf(q)}{\mathbbm{1}_{p^{-1}[0,1]}}}(a,\sigma_\B) = 0
	$
	. We have $\max \msf{ColBSCC}(q,\colVirt,m,a,\sigma_\B)$ odd with $j := \max \msf{ColBSCC}(q,\colVirt,m,a,\sigma_\B) \geq \max(\colVirt(q),c_m)$ for $c_m = m-1$ if $m$ is even and $c_m = m+1$ otherwise. 
	Let us show that $j \geq n$.
	\begin{itemize}
		\item If $m$ is odd (and $n = m+2$), we have $c_m = m+1 = n-1$ which is even. Hence, $j \geq n$. 
		\item If $m = 0$ (and $n = 1$), since $j$ is odd, it must be that $j \geq n$.
		\item If $m \geq 2$ is even (and therefore $n = m-2$), we have $c_m = m-1 = n+1$. Hence, we have $j \geq n+1$. 
	\end{itemize}
	We can apply Lemma~\ref{lem:similarity_color} to obtain that $\max \msf{ColBSCC}(q,\colVirt',m,a,\sigma_\B) = \max \msf{ColBSCC}(q,\colVirt,m,a,\sigma_\B)$, which is odd. As this holds for all $a \in \Act_\A$, we have that the $\GF$-strategy $\sigma_\B$ strongly dominating the valuation $v_{Z'}^{u'}$, hence, by Theorem~\ref{thm:strongly_domnates_imply_guarantee}, the value of the state $q_\msf{init}$ in the game $\G_{Z'}$ is at most $u'$: $\MarVal{\G_{Z'}}[q_\msf{init}] \leq u'$. Hence, $\newC(q,\colVirt') \prec_\msf{par} m$. 
	
	In any case, we have $n = \newC(q,\colVirt')$.
\end{proof*}

%

\subsection{Complete definition of faithfulness}
\label{subsubsec:faithfulness}
We want to formally give the definition of faithfulness that we will use in the proof. First, we define the notion of coherent coloring and environment functions. Informally, for each state $q \in Q_u$: 
the colors given by the coloring function correspond to the environment provided at each state (the environments being defined from coloring functions as in Definition~\ref{def:local_env_induced_coloring_func}).
\begin{definition}[Coherent coloring and environment functions]
	\label{def:corresponding_coloring_environment_functions}
	Consider a virtual coloring function $\colVirt: Q_u \rightarrow \brd$, some color $n \in \brd$ and let $Q_n := \colVirt^{-1}[n]$. Consider an environment function $\msf{Ev}: Q_n \rightarrow \msf{Env}(Q)$ and, for all $q \in Q_n$, let $\colVirt_q: Q_u \rightarrow \brd$ be the coloring function associated with the environment $\msf{Ev}(q)$. Let $q \in Q_n$. We say that $(\colVirt,\msf{Ev})$ is \emph{coherent} at state $q$ if, letting $n_q := \newC(q,\colVirt_q)$:
	\begin{itemize}
		\item $\colFunc(q),\colVirt_q(q),n_q
		\leq n$ and the coloring function $\colVirt$ is $n$-prevailing compared to the coloring function $\colVirt_q$
		;
		\item $n \equiv n_q\mod 2$ and $\msf{Ev}(q) = \msf{CreateEnv}(
		n_q,q,\colVirt_q)$ where $\msf{CreateEnv}$ corresponds to Algorithm~\ref{fig:CreateEnv}.
	\end{itemize}
	If this holds for all $q \in Q_n$, $(\colVirt,\msf{Ev})$ is \emph{coherent} at color $n$.
\end{definition}

\begin{definition}[Faithful pair of coloring and environment functions]
	\label{def:true_trithful_env_coloring_func}
	Consider a virtual coloring function $\colVirt: Q_u \rightarrow \brd$, some $n \in \llbracket 0,e+1 \rrbracket$ and a partial environment function $\msf{Ev}: Q_u \rightarrow \msf{Env}(Q)$ defined on $\colVirt^{-1}[\brck{n,e}]$. We say that $(\colVirt,\msf{Ev})$ is \emph{faithful down to} $n$ if:
	\begin{itemize}
		\item[1f.] for all $k \in \brck{n,e}$, the pair $(\colVirt,\msf{Ev})$ witnesses the color $k$;
		\item[2f.] for all $k \in \brck{n,e}$, the pair $(\colVirt,\msf{Ev})$ is coherent at color $k$;
		\item[3f.] for all $q \in Q_u$, if $\colVirt(q) < n$, then we have $\colFunc(q) = \colVirt(q)$ and $\newC(q,\colVirt) < n$;
	\end{itemize}
	When $n = 0$, we say that the pair $(\colVirt,\msf{Ev})$ is \emph{completely faithful}.
\end{definition}

\subsection{Three central lemmas}
\label{subsubsec:three_lemmas}	
In the following, we state three lemmas that we will use in remainder of this chapter. The first lemma relates probability distributions in a local game and in a global game. This is particularly useful as it allows to use the assumptions made on the local strategies to obtain various properties on global games where such local strategies are used.

The second lemma states that any Player-$\A$ positional strategy generated by an environment such that the new color (w.r.t. that environment) is even dominates a specific valuation. This is similar for Player $\B$.

Finally, the third lemma states that any Player-$\A$ local strategy defined by local strategies that are optimal w.r.t. to an even color ensure that in a BSCC $H$ compatible with such a strategy, if some high enough color occurs in $H$
, then $H$ is even-colored. 
This is similar for Player $\B$. 
\begin{lemma}
	\label{lem:similar_local_behavior_local_and_global_game}
	Consider a virtual coloring function $\colVirt: Q_u \rightarrow \llbracket 0,e \rrbracket$, some color $n \in \brd$ and a partial environment function $\msf{Ev}: Q_u \rightarrow \msf{Env}(Q)$ defined over $Q_{n} := \colVirt^{-1}[\{ n \}]$. Let $q \in Q_n$ and $\colVirt_q: Q_u \rightarrow \brd$ be the coloring function associated with environment $\msf{Ev}(q)$
	. 
	We also let $p := p_{q,\colVirt_q}$. Then, for all 
	Player-$\A$ and Player-$\B$ strategies $\s_\A$ and $\s_\B$ in the arena $\Aconc_{\colVirt}^{n}$, we have:
	\begin{equation}
	\label{eqn:prob_in_V_Q_Q_u_equal}
	\forall x \in V_{Q \setminus Q_u},\; \mathbb{P}_{\Aconc_{\colVirt}^{n},q}^{\s_\A,\s_\B}[x] = \outM_{\Games{\fsf(q)}{\mathbbm{1}_{p^{-1}[x]}}}(\s_\A(q),\s_\B(q))
	\end{equation}
	and 
	\begin{equation}
	\label{eqn:prob_in_Q_u_K_n_equal} \mathbb{P}_{\Aconc_{\colVirt}^{n},q}^{\s_\A,\s_\B}[Q_n \cup K^n] = \outM_{\Games{\fsf(q)}{\mathbbm{1}_{p^{-1}[\{q_{\msf{init}}\} \cup K_e]}}}(\s_\A(q),\s_\B(q))
	\end{equation}
	Furthermore, if $\colVirt$ is equivalent down to $n+1$ to $\colVirt_q$:
	\begin{equation}
	\label{eqn:prob_ki_if_equiv_equal}
	\forall i \in \llbracket n+1,e \rrbracket,\; \mathbb{P}_{\Aconc_{\colVirt}^{n},q}^{\s_\A,\s_\B}[k_i^n] = \outM_{\Games{\fsf(q)}{\mathbbm{1}_{p^{-1}[k_i]}}}(\s_\A(q),\s_\B(q))
	\end{equation}
	and if, in addition, $\colVirt$ is $n$-prevailing compared to $\colVirt_q$:
	\begin{equation}
	\label{eqn:prob_n_similar_equal}
	\mathbb{P}_{\Aconc_{\colVirt}^{n},q}^{\s_\A,\s_\B}[\colVirt_q^{-1}[n] \setminus \{ q \}] = \outM_{\Games{\fsf(q)}{\mathbbm{1}_{p^{-1}[k_n]}}}(\s_\A(q),\s_\B(q))
	\end{equation}
\end{lemma}
\begin{proof*}
	By 
	Definition~\ref{def:probability_extracted_parity_game}, we have:
	\begin{itemize}
		\item for all $q' \in Q_n$, $p_{n,\colVirt}^{-1}[q'] = \{ q' \}$;
		\item for all $i \in \brd$: $p_{n,\colVirt}^{-1}[k_i^n] = \colVirt^{-1}[i] \setminus Q_n$;
		\item for all $x \in V_{Q \setminus Q_u}$, $p_{n,\colVirt}^{-1}[x] = Q_x$.
	\end{itemize}
	Furthermore, from Definition~\ref{def:local_env_induced_coloring_func}, we have:
	\begin{itemize}
		\item $p^{-1}[\qinit] = \{ q \}$;
		\item for all $i \in \brd$: $p^{-1}[k_i] = \colVirt^{-1}[i] \setminus \{ q \}$;
		\item for all $x \in V_{Q \setminus Q_u}$, $p^{-1}[x] = Q_x = p_{n,\colVirt}^{-1}[x]$.
	\end{itemize}
	
	
	By Definition~\ref{def:parity_game_from_set_of_gfs}, for all $x \in V_{Q \setminus Q_u}$. We have:
	\begin{align*}
	\Prb_{\Aconc_{\colVirt}^{n},q}^{\s_\A,\s_\B}(x) & = \outM_{\Games{\fsf(q)}{\mathbbm{1}_{p_{n,\colVirt}^{-1}[x]}}}(\s_\A(q),\s_\B(q)) \\
	& = \outM_{\Games{\fsf(q)}{\mathbbm{1}_{p^{-1}[x]}}}(\s_\A(q),\s_\B(q))
	\end{align*}
	This proves Equation~(\ref{eqn:prob_in_V_Q_Q_u_equal}).
	
	Furthermore, using this Equation for the second equality, we have:
	\begin{align*}
	\mathbb{P}_{\Aconc_{\colVirt}^{n},q}^{\s_\A,\s_\B}[Q_n \cup K^n] & = 1 - \mathbb{P}_{\Aconc_{\colVirt}^{n},q}^{\s_\A,\s_\B}[V_{Q \setminus Q_u}] \\
	& = 1 - \outM_{\Games{\fsf(q)}{\mathbbm{1}_{p^{-1}[V_{Q \setminus Q_u}]}}}(\s_\A(q),\s_\B(q)) \\
	& = \outM_{\Games{\fsf(q)}{1 - \mathbbm{1}_{p^{-1}[V_{Q \setminus Q_u}]}}}(\s_\A(q),\s_\B(q)) \\
	& = \outM_{\Games{\fsf(q)}{\mathbbm{1}_{p^{-1}[\{ q_{\msf{init}} \} \cup K_e]}}}(\s_\A(q),\s_\B(q))
	\end{align*}
	This proves Equation~(\ref{eqn:prob_in_Q_u_K_n_equal}).
	
	Assume now that $\colVirt$ is equivalent down to $n+1$ to $\colVirt_q$. Let $i \in \brck{n+1,e}$. We have $\colVirt^{-1}[i] = \colVirt_q^{-1}[i]$ (recall Definition~\ref{def:k_similar_coloring_function}). Furthermore, for all $q' \in Q_n$, we have $\colVirt(q') = n \neq i$. Hence, $\colVirt^{-1}[i] \setminus Q_n = \colVirt^{-1}[i] = \colVirt^{-1}[i] \setminus \{ q \}$. 
	Hence we have, again by Definition~\ref{def:parity_game_from_set_of_gfs}, for all $i \in \brck{n+1,e}$:
	\begin{align*}
	\mathbb{P}_{\Aconc_{\colVirt}^{n},q}^{\s_\A,\s_\B}(k_i^n) & = \outM_{\Games{\fsf(q)}{\mathbbm{1}_{p_{n,\colVirt}^{-1}[k_i^n]}}}(\s_\A(q),\s_\B(q))  \\
	& = \outM_{\Games{\fsf(q)}{\colVirt^{-1}[i]}}(\s_\A(q),\s_\B(q))  \\
	& = \outM_{\Games{\fsf(q)}{\mathbbm{1}_{p^{-1}[k_i]}}}(\s_\A(q),\s_\B(q)) 
	\end{align*}
	
	This gives Equation~(\ref{eqn:prob_ki_if_equiv_equal}). Assume in addition that $\colVirt$ is $n$-prevailing compared to $\colVirt_q$ 	(recall Definition~\ref{def:k_similar_coloring_function}). That is, $\colVirt_q^{-1}[n] \subseteq \colVirt^{-1}[n] = Q_n$ 
	. 
	Hence, $
	\colVirt_q^{-1}[n] \setminus \{ q \} \subseteq Q_n$. 
	It follows that:
	\begin{align*}
	\mathbb{P}_{\Aconc_{\colVirt}^{n},q}^{\s_\A,\s_\B}[\colVirt_q^{-1}[n] \setminus \{ q \}] & = \sum_{q' \in \colVirt_q^{-1}[n]  \setminus \{ q \}} \outM_{\Games{\fsf(q)}{\mathbbm{1}_{p_{n,\colVirt}^{-1}[q']}}}(\s_\A(q),\s_\B(q)) \\
	& = \sum_{q' \in \colVirt_q^{-1}[n]  \setminus \{ q \}} \outM_{\Games{\fsf(q)}{\mathbbm{1}_{q'}}}(\s_\A(q),\s_\B(q)) \\
	& = \outM_{\Games{\fsf(q)}{\mathbbm{1}_{\colVirt_q^{-1}[n] \setminus \{ q \}}}}(\s_\A(q),\s_\B(q))
	\\
	& = \outM_{\Games{\fsf(q)}{\mathbbm{1}_{p^{-1}[k_n]}}}(\s_\A(q),\s_\B(q)) 
	\end{align*}
	We obtain Equation~(\ref{eqn:prob_n_similar_equal}).
\end{proof*}

\begin{lemma}
	\label{lem:update_with_algo_dominating}
	Consider a virtual coloring function $\colVirt: Q_u \rightarrow \brd$, some color $n \in \brd$ and let $Q_n := \colVirt^{-1}[n]$. Consider an environment function $\msf{Ev}: Q_n \rightarrow \msf{Env}(Q)$. 
	Assume that $(\colVirt,\msf{Ev})$ is coherent at color $n$. In the arena $\Aconc_{\colVirt}^{n}$:
	\begin{itemize}
		\item if 
		$n$ is even, then all positional Player-$\A$ strategies generated by the environment $\msf{Ev}$ dominate the valuation $v_{n,\colVirt}^u$ in the arena $\Aconc^{n}_{\colVirt}$ (see Definition~\ref{def:env_guarantee_value_and_color});
		\item if 
		$n$ is odd, then there is some $y < u$ such that, for all $z \geq y$, all positional Player-$\B$ strategies generated by the environment $\msf{Ev}$ dominate the valuation $v_{n,\colVirt}^z$ in the arena $\Aconc^{n}_{\colVirt}$ (see Definition~\ref{def:env_guarantee_value_and_color}). 
	\end{itemize}
\end{lemma}
\begin{proof*}
	
	For all states $q \in Q_n$, 
	we let $\colVirt_q: Q_u \rightarrow \brd$ be the coloring function associated with the environment $\msf{Ev}(q)$ and we let $p := p_{q,\colVirt_q}$. 
	Let $q \in Q_n$. For all Player-$\A$ and $\B$ strategies $\s_\A,\s_\B$ in the arena $\Aconc_{\colVirt}^{n}$, for all $z \in [0,1]$
	, we have:
	\begin{align*}
	\outM_{\Games{\fsf(q)}{v_{n,\colVirt}^z}}(\s_\A(q),\s_\B(q)) 
	& = z \cdot \Prb_{\Aconc_{\colVirt}^n,q}^{\s_\A,\s_\B}[Q_n \cup K] + \sum_{x \in V_{Q \setminus Q_u}} x \cdot \Prb_{\Aconc_{\colVirt}^n,q}^{\s_\A,\s_\B}(x)
	\end{align*}
	
	Assume that $n$ is even and consider a Player-$\A$ positional strategy $\s_\A$ generated by the environment function $\msf{Ev}$. Letting $Y_q := (\fsf(q),\msf{Ev}(q))$, in the game $\G_{Y_q}$, the positional Player-$\A$ strategy $\s_{Y_q}^\A(\s_\A(q))$ dominates the valuation $v_{Y_q}^{u'}$ for $u' := \MarVal{\G_{Y_q}}[q_{\msf{init}}] \geq u$ (recall Lemma~\ref{lem:opt_in_gf_equiv_opt_in_parity}) since $(\colVirt,\msf{Ev})$ is coherent at color $n$. 
	Hence:
	\begin{align*}
	u' & \leq \outM_{\Games{\fsf(q)}{v_{Y_q}^{u'}}}(\s_\A(q),\s_\B(q)) \\ 
	& = \outM_{\Games{\fsf(q)}{\mathbbm{1}_{p^{-1}[\{ q_\msf{init} \} \cup K_e]}}}(\s_\A(q),\s_\B(q))\cdot u' \\
	& + \sum_{x \in V_{Q \setminus Q_u}} \outM_{\Games{\fsf(q)}{\mathbbm{1}_{p^{-1}[x]}}}(\s_\A(q),\s_\B(q)) \cdot x \\
	& = u' \cdot \mathbb{P}_{\Aconc_{\colVirt}^{n},q}^{\s_\A,\s_\B}[Q_n \cup \colSet^n] + \sum_{x \in V_{Q \setminus Q_u}} x \cdot \mathbb{P}_{\Aconc_{\colVirt}^{n},q}^{\s_\A,\s_\B}[x]
	\end{align*}
	Then, we have (since $u' \geq u$):
	\begin{align*}
	u = u' + (u-u') & \leq u' \cdot \mathbb{P}_{\Aconc_{\colVirt}^{n},q}^{\s_\A,\s_\B}[Q_n \cup K^n] + \sum_{x \in V_{Q \setminus Q_u}} x \cdot \mathbb{P}_{\Aconc_{\colVirt}^{n},q}^{\s_\A,\s_\B}[x] + (u - u') \\
	& \leq u' \cdot \mathbb{P}^{\Aconc_{\colVirt}^{n},q}_{\s_\A,\s_\B}[Q_n \cup K^n] + \sum_{x \in V_{Q \setminus Q_u}} x \cdot \mathbb{P}_{\Aconc_{\colVirt}^{n},q}^{\s_\A,\s_\B}[x] + (u - u') \cdot \mathbb{P}_{\Aconc_{\colVirt}^{n},q}^{\s_\A,\s_\B}[Q_n \cup K^n] \\
	& = u \cdot \mathbb{P}_{\Aconc_{\colVirt}^{n},q}^{\s_\A,\s_\B}[Q_n \cup K^n] + \sum_{x \in V_{Q \setminus Q_u}} x \cdot \mathbb{P}_{\Aconc_{\colVirt}^{n},q}^{\s_\A,\s_\B}[x] \\
	& = \outM_{\Games{\fsf(q)}{v_{n,\colVirt}^u}}(\s_\A(q),\s_\B(q))
	\end{align*}
	Since this holds for all Player-$\B$ positional strategies strategies $\s_\B$ and for all $q \in Q_n$, it follows that the Player-$\A$ strategy $\s_\A$ dominates the valuation $v_{n,\colVirt}^u$ in the arena $\Aconc_{\colVirt}^{n}$. 
	
	In the case where $n$ is odd, the proof is analogous, from Player-$\B$'s point-of-view. 
\end{proof*}

\begin{lemma}
	\label{lem:good_local_highest_color_even}
	Consider a virtual coloring function $\colVirt: Q_u \rightarrow \brd$, some $n \in \brd$ and $Q_{n} := \colVirt^{-1}[\{ n \}]$. Consider an environment function $\msf{Ev}: Q_n \rightarrow \msf{Env}(Q)$ and assume that $(\colVirt,\msf{Ev})$ is coherent at the color $n$. 
	Then, 
	denoting $K^{\geq n} := \{ k_i^n \mid i \in \brck{n,e} \}$:
	\begin{itemize}
		\item if $n$ is even
		, then in the arena $\Aconc_{\colVirt}^{n}$, all positional Player-$\A$ strategies generated by the environment $\msf{Ev}$ ensure that for all BSCCs $H$ compatible with $\s_\A$, if $K^{\geq n}$ occurs in $H$, then $H$ is even-colored;
		\item 
		if $n$ is odd
		, then in the arena $\Aconc_{\colVirt}^{n}$, all positional Player-$\B$ strategies generated by the environment $\msf{Ev}$ ensure that for all BSCCs $H$ compatible with $\s_\B$, if $K^{\geq n}$ occurs in $H$, then $H$ is odd-colored.
	\end{itemize}
\end{lemma}
\begin{proof*}
	Assume that $n$ is even and consider a Player-$\A$ positional strategy $\s_\A$ generated by the environment $\msf{Ev}$ in the game $\Zone^{n}_{\colVirt}
	$ (recall Definition~\ref{def:parity_game_from_set_of_gfs}). Consider a Player-$\B$ positional deterministic strategy $\s_\B$ and consider a BSCC $H$ compatible with $\s_\A$ and $\s_\B$. 
	
	For all $q \in Q_{n}$, let $\colVirt_q: Q_u \rightarrow \brd$ be the coloring function associated with the environment $\msf{Ev}(q)$. Let $q \in Q_{n}$ and 
	$n_q := \newC(q,\colVirt_q)$. Since $(\colVirt,\msf{Ev})$ is coherent at $n$, $n_q$ is even and less than or equal to $n$. Recall that $\G_{q,\colVirt}^{n_q} = \G_{Y_q}$, for $Y_q := (\fsf(q),E^{n_q}_{q,\colVirt_q})$ (see Definition~\ref{def:local_env_induced_coloring_func}). Recall also that $E^{n_q}_{q,\colVirt_q} = \langle \max(c_{n_q},\colVirt_q(q)),e,p_{\{ q \},\colVirt_q} \rangle
	$ (also see Definition~\ref{def:local_env_induced_coloring_func}) with $c_{n_q} = n_q-1$ as $n_q$ is even.  
	Let $p^q := p_{q,\colVirt_q}$. Note that $p^q_{[0,1]} \subseteq V_{Q \setminus Q_u}$ (recall Definition~\ref{def:local_env_induced_coloring_func}). 
	By Equation $(\ref{eqn:prob_in_V_Q_Q_u_equal})$ from Lemma~\ref{lem:similar_local_behavior_local_and_global_game}, if $q$ is in $H$, then $\outM_{\Games{\fsf(q)}{\mathbbm{1}_{Q \setminus Q_u}}}(\s_\A(q),\s_\B(q)) = 0$. In that case, since the Player-$\A$ $\GF$-strategy $\s_\A(q)$ is optimal w.r.t. $Y_q$ (recall Lemma~\ref{lem:opt_in_gf_equiv_opt_in_parity}), we have the following.: 
	\begin{equation}
	\label{eqn:relation_at_state_q}
	\max (\mathsf{Color}(\fsf(q),p^q,\s_\A(q),\s_\B(q)),n_q-1,\colVirt_q(q)) \text{ is even }
	\end{equation}
	Since $(\colVirt,\msf{Ev})$ is coherent (see Definition~\ref{def:corresponding_coloring_environment_functions}) at $n$, we have $\colFunc(q) \leq n$ for all $q \in Q_n$. Hence, our goal is to show that:
	\begin{equation*}
	\label{eqn:max_colors_even}
	K^{\geq n} \text{ occurs in } H \Rightarrow \max M \text{ is even }
	\end{equation*}
	with $M := \{ i \in \llbracket n,e \rrbracket \mid k_i^n \text{ occurs in } H \}$.
	
	To prove this, we use the following characterization: for all colors $i \in \llbracket n+1,e \rrbracket$, by Equation $(\ref{eqn:prob_ki_if_equiv_equal})$ in Lemma~\ref{lem:similar_local_behavior_local_and_global_game} since $\colVirt$ and $\colVirt_q$ are equivalent down to $n+1$, we have the following equivalence:
	\begin{equation}
	\label{eqn:equiv_in_bscc}
	k_i^n \text{ occurs in } H \Leftrightarrow \exists q \in H,\; i \in \mathsf{Color}(\fsf(q),p^q,\s_\A(q),\s_\B(q))
	\end{equation}
	
	Furthermore, for all $q \in H$, by assumption (the pair $(\colVirt,\msf{Ev})$ being coherent), $\colVirt_q(q) \leq n$. 
	Hence, since $n_q \leq n$, Equation~(\ref{eqn:relation_at_state_q}) gives that, for all $q \in H$:
	\begin{equation}
	\label{eqn:if_Iq_at_least_n_then_max_even}
	\max \mathsf{Color}(\fsf(q),p^q,\s_\A(q),\s_\B(q)) \geq n+1 \Rightarrow \max \mathsf{Color}(\fsf(q),p^q,\s_\A(q),\s_\B(q)) \text{ is even }
	\end{equation}
	
	Furthermore, if $\max M = n$, then $n$ is the highest color appearing in $H$ (since $\colVirt$ and $\msf{Ev}$ are coherent at $n$) and $H$ is then even-colored. Otherwise:
	\begin{align*}
	K^{\geq n} \text{ occurs in } H & \Leftrightarrow M \neq \emptyset \\
	& \Rightarrow M \neq \emptyset \wedge k_m^n \text{ occurs in }H  \text{ for }m := \max M \\
	& \Rightarrow M \neq \emptyset \wedge \exists q_m \in H,\; m \in \mathsf{Color}(\fsf(q_m),p,\s_\A(q_m),\s_\B(q_m)) \\
	& \text{ by Equation~(\ref{eqn:equiv_in_bscc}) } \\
	& \Rightarrow M \neq \emptyset \wedge \exists q_m \in H,\; \max \mathsf{Color}(\fsf(q_m),p^{q_m},\s_\A(q_m),\s_\B(q_m)) = m \\
	& \text{ since } m = \max M \text{ and by Equation~(\ref{eqn:equiv_in_bscc}) }\\
	& \Rightarrow M \neq \emptyset \wedge m = \max M \text{ is even } \\
	& \text{ since } m \geq n+1 \text{ and by Equation~(\ref{eqn:if_Iq_at_least_n_then_max_even}) }
	\end{align*}
	We obtain the desired result. 
	
	The case of $n$ odd is analogous by reversing the instances of 'odd' and 'even'. However, the arguments for obtaining the analogue of Equation~(\ref{eqn:relation_at_state_q}) is slightly different. Indeed, letting $Y_q := (\fsf(q),E^{\ScPar(n_q)}_{\{ q \},\colVirt_q})$, the Player-$\B$ $\GF$-strategy $\s_\B(q)$ is optimal w.r.t. $Y_q$ (recall Lemma~\ref{lem:opt_in_gf_equiv_opt_in_parity} and the definition of $\msf{CreateEnv}$ (i.e. Algorithm~\ref{fig:CreateEnv})). 
	Recall that $E^{\ScPar(n_q)}_{\{ q \},\colVirt_q} = \langle \max(c_{\ScPar(n_q)},\colVirt_q(q)),e,\colVirt_q \rangle$ where $c_{\ScPar(n_q)} = -1$ if $n_q = 1$ and $c_{\ScPar(n_q)} = n_q-1$ otherwise (i.e. if $n_q \geq 3$). Overall, from Lemma~\ref{lem:opt_in_gf_equiv_opt_in_parity}, we do obtain: 
	\begin{equation}
	\label{eqn:relation_at_state_q_B}
	\max (\mathsf{Color}(\fsf(q),p,\s_\A(q),\s_\B(q)),n_q-1,\colVirt_q(q)) \text{ is odd }
	\end{equation}
	with $n_q-1$ even.
\end{proof*}

\subsection{Proof of Lemma~\ref{lem:algo_preserves_faithfulness}, Page~\pageref{lem:algo_preserves_faithfulness}}
\label{subsubsec:proof_lem_algo_preserves_faithfulness}	
Let us define a notion which is slightly weaker then being faithful: being non-deceiving.
\begin{definition}[Non-deceiving environment and coloring functions]
	\label{def:non-deceiving_env_coloring_func}
	Consider a virtual coloring function $\colVirt: Q_u \rightarrow \llbracket 0,e \rrbracket$, a partial environment function $\msf{Ev}: Q_u \rightarrow \msf{Env}(Q)$ and some color $n \in \llbracket 0,e+1 \rrbracket$. We say that the pair $(\colVirt,\msf{Ev})$ is \emph{non-deceiving down to} $n$ if:
	\begin{itemize}
		\item[1n-d.] for all $k \in \brck{n,e}$, the pair $(\colVirt,\msf{Ev})$ witnesses the color $k$;
		\item[2n-d.] for all $k \in \brck{n,e}$, the pair $(\colVirt,\msf{Ev})$ is coherent at color $k$;
		\item[3n-d.] for all $q \in Q_u$, if $\colVirt(q) < n$, then we have $\colFunc(q) = \colVirt(q)$ and $\newC(q,\colVirt) \leq n$.
	\end{itemize}
	The difference with being faithful lies in the fact that some state $q \in Q_u$ with $\colVirt(q) < n$ could be such that $\newC(q,\colVirt) = n$ (which is not possible with a faithful pair).
	When $n = 0$, being non-deceiving down to $n$ is equivalent to being faithful down to $n$ (i.e. to being completely faithful). 
\end{definition}

Then, Algorithm~\ref{fig:UpdateColEnv} is composed of Algorithm~\ref{fig:UpdCurSta} and Algorithm~\ref{fig:UpdNewSta}. In fact, Algorithm~\ref{fig:UpdCurSta} transforms a faithful pair into a non-deceiving one (one level below). And Algorithm~\ref{fig:UpdNewSta} transforms a non-deceiving pair into a faithful one (at the same level). We state one lemma per algorithm formally stating the specifications of these algorithms.
\begin{lemma}
	\label{lem:algo_turthful_to_non_deceiving}
	Consider a virtual coloring function $\colVirt: Q_u \rightarrow \llbracket 0,e \rrbracket$, some color $n \in \llbracket 1,e+1 \rrbracket$ and a partial environment function $\msf{Ev}: Q_u \rightarrow \msf{Env}(Q)$ defined on $\colVirt^{-1}[\brck{n,e}]$
	. Assume that $(\colVirt,\msf{Ev})$ is faithful down to $n$. Let $\msf{Ev}' \gets \msf{UpdCurSta}(n-1,\colVirt,\msf{Ev})$. Then, $(\colVirt,\msf{Ev}')$ is non-deceiving down to $n-1$.
\end{lemma}
\begin{proof*}
	For all $q \in Q_u$, such that $\colVirt(q) \geq n-1$, we denote by $\colVirt_q$ the coloring function corresponding to the environment $\msf{Ev}'(q)$.
	\begin{itemize}
		\item[2n-d.] Let us show that $(\colVirt,\msf{Ev}')$ is coherent at $n-1$. We let $Q_{n-1} := \colVirt^{-1}[n-1]$ and $q \in Q_{n-1}$. Since $(\colVirt,\msf{Ev})$ is faithful down to $n$, we have $\colVirt(q) = \colFunc(q) = n-1$. Furthermore, $\colVirt_q = \colVirt$. Hence, $\colVirt_q(q) = \colFunc(q) = n-1$. In addition, since the pair $(\colVirt,\msf{Ev})$ is faithful down to $n$, we have $\newC(q,\colVirt) < n$. By Proposition~\ref{prop:update_color_at_least_preivous_color}, we have $n-1 = \colVirt(q) \leq \newC(q,\colVirt) < n$. Hence, $\colVirt(q) = \newC(q,\colVirt) = n-1$. By definition of Algorithm~\ref{fig:CreateEnv}, we have $\msf{Ev}'(q) = \msf{CreateEnv}(n-1,q,\colVirt)$. As this holds for all $q \in Q_{n-1}$, $(\colVirt,\msf{Ev})$ is coherent at $n-1$.
		\item[3n-d.] This condition straightforwardly holds since the coloring function has not changed and the pair $(\colVirt,\msf{Ev})$ is faithful down to $n$.
		\item[1n-d.] This condition straightforwardly holds for $k \geq n$ since the pair $(\colVirt,\msf{Ev})$ is faithful down to $n$. Note that since $(\colVirt,\msf{Ev}')$ is coherent at $n-1$, both Lemma~\ref{lem:update_with_algo_dominating} and Lemma~\ref{lem:good_local_highest_color_even} can be applied to $(\colVirt,\msf{Ev}')$ at  $n-1$. 
		
		Assume that $n-1$ is even. Consider a Player-$\A$ positional strategy $\s_\A$ generated by the environment $\msf{Ev}$ in the game $\Zone^{n-1}_{\colVirt}$. By Lemma~\ref{lem:update_with_algo_dominating}, this strategy dominates the valuation $v_{n-1,\colVirt}^u$ in the arena $\Aconc_{\colVirt}^{n-1}$. Consider a Player-$\B$ positional deterministic strategy $\s_\B$ and consider a BSCC $H$ in $\Aconc_{\colVirt}^{n-1}$ compatible with $\s_\A$ and $\s_\B$. By Lemma~\ref{lem:good_local_highest_color_even}, if $K^{\geq n-1}$ occurs in $H$, then $H$ is even-colored. 
		Assume now that $K^{\geq n-1}$ does not occur in $H$. In that case, there are some states in $H$ in $Q_{n-1} \subseteq \colFunc^{-1}[n-1]$ and possibly some states in $\{ k_i^{n-1} \mid i \in \llbracket 0,n-2 \rrbracket \}$ occur in $H$ with $\colVirt_{n-1}(k_i^{n-1}) = n-2$ for all $i \in \llbracket 0,n-2 \rrbracket$ and $\colVirt_{n-1}[Q_{n-1}] = \{ n-1 \}$. Hence, the highest color appearing in $H$ is $n-1$ and therefore $H$ is even-colored. 
		Thus, $\s_\A$ parity dominates the valuation $v_{n-1,\colVirt}^u$.
		
		The case where $n-1$ is odd is analogous, from Player-$\B$'s point-of-view. 
	\end{itemize}

\end{proof*}

\begin{lemma}
	\label{lem:algo_turthful_non_deceiving_to_faithful}
	Consider a virtual coloring function $\colVirt: Q_u \rightarrow \brd$, some color $n \in \brd$, and a partial environment function $\msf{Ev}: Q_u \rightarrow \msf{Env}(Q)$ defined on $\brck{n,e}$. Assume that the pair $(\colVirt,\msf{Ev})$ is non-deceiving down to $n$. Let $(\colVirt',\msf{Ev}') \gets \msf{UpdNewSta}(n,\colVirt,\msf{Ev})$. Then, $(\colVirt',\msf{Ev}')$ is faithful down to $n$.
\end{lemma}
\begin{proof*}
	For all $q \in Q_u$ such that $\colVirt(q) \geq n$, we denote by $\colVirt_q$ the coloring function associated with the environment $\msf{Ev}'(q)$. Let $Q_n := (\colVirt')^{-1}[n]$. 
	
	\begin{itemize}
		\item[2f.] For all $k \in \brck{n+1,e}$, the pair $(\colVirt',\msf{Ev}')$ is coherent at $k$ since the pair $(\colVirt,\msf{Ev})$ is non-deceiving down to $n$. Now, let $q \in Q_n$. If $\colVirt(q) = n$, then straightforwardly, $(\colVirt',\msf{Ev}')$ is coherent at state $q$ since $(\colVirt,\msf{Ev})$ is non-deceiving down to $n$, hence $\colVirt$ is $n$-prevailing compared to $\colVirt_q$ and $\colVirt'$ is also $n$-prevailing compared to $\colVirt$. Consider now a state $q \in Q_n$ such that $\colVirt(q) < n$. 
		By definition of Algorithm~\ref{fig:UpdNewSta}, we have $\colVirt(q) = \colVirt_q(q) \leq \newC(q,\colVirt_q) = n$ (by Proposition~\ref{prop:update_color_at_least_preivous_color}). That is, $\colVirt_q(q) \leq n-1$. Since $(\colVirt,\msf{Ev})$ is non-deceiving down to $n$, we have $\colFunc(q) = \colVirt(q) = \colVirt_q(q) \leq n-1$. Furthermore, by definition of Algorithm~\ref{fig:UpdNewSta}, the colors of the states in $\colVirt^{-1}[\brck{n,e}]$ 
		is not changed and more and more states are colored with $n$. Hence, $\colVirt'$ is $n$-prevailing compared to $\colVirt_q$. In addition, $n = \newC(q,\colVirt_q)$ and $\msf{Ev}(q) = \msf{CreateEnv}(n,q,\colVirt_q)$. 
		It follows that the pair $(\colVirt',\msf{Ev}')$ is coherent at state $q$, and this holds for all $q \in Q_n$
		.
		\item[3f.] 
		By definition of Algorithm~\ref{fig:UpdNewSta}, for all $k < n$ and state $q \in Q_u$ such that $\colVirt'(q) = k$, we have $\newC(q,\colVirt') \neq n$ and $\colVirt'(q) = \colVirt(q) = \colFunc(q)$ (since $(\colVirt,\msf{Ev})$ is non-deceiving down to $n$). Furthermore, assume towards a contradiction that a state $q \in Q_u$ is such that $\colVirt'(q) = k = \colVirt(q)$ and $\newC(q,\colVirt') = m > n$. In that case, by Proposition~\ref{prop:update_color_does_not_change_if_small_color_change}, since $\colVirt$ and $\colVirt'$ are equivalent down to $m$, we would have $\newC(q,\colVirt) = m > n$, which is not possible since $(\colVirt,\msf{Ev})$ is non-deceiving down to $n$. In fact, $\newC(q,\colVirt') < n$. 
		\item[1f.] 
		Let $k \geq n+1$. We have $\colVirt^{-1}[k] = (\colVirt')^{-1}[k]$ and for all $q \in \colVirt^{-1}[k]$, we have $\msf{Ev}(q) = \msf{Ev}'(q)$. Hence, the first condition for faithfulness holds for $k$ since $(\colVirt,\msf{Ev})$ is non-deceiving down to $n$. We consider now the case where $k = n$. 
		%
		%
		Assume that $n$ is even. 
		Consider a Player-$\A$ positional strategy $\s_\A$ generated by the environment function $\msf{Ev}$ in the game $\Zone^{n}_{\colVirt'}$. Let us show that this strategy parity dominates the valuation $v_{n,\colVirt'}^u$ in the game $\Zone^{n}_{\colVirt'}$. By Lemma~\ref{lem:update_with_algo_dominating}, this strategy dominates the valuation $v_{n,\colVirt'}^u$. Now, fix a Player-$\B$ positional deterministic strategy $\s_\B$ and consider a BSCC $H$ compatible with $\s_\A$ and $\s_\B$ in $\Zone^{n}_{\colVirt'}$. If $K^{\geq n}$ occurs in $H$, then by Lemma~\ref{lem:good_local_highest_color_even}, the BSCC $H$ is even-colored.
		Assume now that $K^{\geq n}$ does not occur in $H$. Let $X_n^0 := \colVirt^{-1}[n]$ and let $j := |Q_n| - |X_n^0|$. For all $1 \leq i \leq j$, we denote by $q_i$ the $i$-th element of $Q_n$ whose color was changed by Algorithm~\ref{fig:UpdNewSta} and $X_n^i := X_n^0 \cup \cup_{1 \leq k \leq i} \{q_i\}$. In particular, $X_n^j = Q_n$ and:
		\begin{equation}
		\label{eqn:X_n^^i_eq_colFunc_n}
		\forall i \in \brck{1,j},\;
		\colVirt_{q_i}^{-1}[n] = X_n^{i-1}
		\end{equation}
		First, assume that $H \subseteq X_n^0$. In that case, this is a BSCC compatible with the strategy $\s_\A$ in the game $\Zone^{n}_{\colVirt}$ in which this strategy, by assumption, parity dominates the valuation $v_{n,\colVirt}^u$. Hence, this BSCC is even-colored. Second, assume that $H \not \subseteq X_n^0$ and let $C_n := Q_n \cap \colFunc^{-1}[n]$. For all $i \in \brck{0,j}$, let $H_i := H \cap X_n^i$. Let us show by induction on $i \in \brck{0,j}$ the following property $\mathcal{P}(i)$: for all $q \in H_i \setminus C_n$, there is a positive probability to visit $C_n$ from $q$ while only visiting states in $H_i$, that is:
		\begin{displaymath}
		\Prb_{\Aconc_{\colVirt'}^{n},q}^{\s_\A,\s_\B}[H_i^* \cdot C_n] > 0
		\end{displaymath}
		Assume towards a contradiction that $\mathcal{P}(0)$ does not hold. That is, there is some $q \in H_0 \setminus C_n$ such that $\mathbb{P}_{\Aconc_{\colVirt'}^{n},q}^{\s_\A,\s_\B}[H_0^* \cdot C_n] = 0$. Let $Z \subseteq H_0$ be the set of states in $H_0$ defined by $Z := \{ q \} \cup \{ q' \in H_0 \mid \mathbb{P}_{\Aconc_{\colVirt'}^{n},q}^{\s_\A,\s_\B}[H_0^* \cdot q'] > 0 \}$. Note that $Z \cap C_n = \emptyset$. Our goal is now to exhibit a BSCC in $\Zone^{n}_{\colVirt}$ that is compatible with $\s_\A$ and that is odd-colored. This will be a contradiction with the fact that $\s_\A$ parity dominates the valuation $v_{n,\colVirt}^u$ in the game $\Zone^{n}_{\colVirt}$. Consider a Player-$\B$ strategy $\s_\B'$ such that, for all $q' \in X_n^0 = \colVirt^{-1}[n]$, $\s_\B'(q') := \s_\B(q')$ and for all $i \in \llbracket 0,n-1 \rrbracket$, we have $\s_\B(k_i^n)$ playing deterministically to reach $q$
		. By construction, we do have a BSCC $H'$ compatible with $\s_\A$ and $\s_\B'$ in $\Zone^{n}_{\colVirt}$ that is included in $Z \cup \{ k_i^n \mid i \in \llbracket 0,n-1 \rrbracket \}$. Furthermore, since the BSCC $H$ in the game $\Zone^{n}_{\colVirt'}$ is not equal to $H_0$ (because this would imply $H \subseteq X_n^0$), it must be that at least one state in $\{ k_i^n \mid i \in \llbracket 0,n-1 \rrbracket \}$ is in $H'$ (recall that in the game $\Zone^{n}_{\colVirt}$ any edge leading to a state in $Q_u \setminus X_n^0$ in $\G$ leads to a state in $\{ k_i^n \mid i \in \brd \}$). Then, since each state in $\{ k_i^n \mid i \in \llbracket 0,n-1 \rrbracket \}$ is colored with $n-1$ and no state in $Z$ is in $C_n$, it follows that the highest color in $H$ is $n-1$, which is odd. 
		That is, the BSCC $H'$ his odd-colored, which is a contradiction
		. In fact, $\mathcal{P}(0)$ holds. 
		
		
		Assume now that $\mathcal{P}(i)$ holds for some $0 \leq i \leq j -1$. Consider some state $q \in H_{i+1} \setminus C_n$. If $q \in H_i$, then we can apply $\mathcal{P}(i)$. Assume now that $q \notin H_i$, that is $q \notin X^i_n$. Hence, $q = q_{i+1}$. Since $q \notin C_n$, we have $\colFunc(q) \neq n$, and therefore $\colFunc(q) \leq n-1$\footnote{We have shown that $(\colVirt,\msf{Ev})$ is coherent at $n$, hence all states in $Q_n$ are such that $\colFunc(q) \leq n$.}. Furthermore, letting $Y_q := (\fsf(q),E^{n}_{\{ q \},\colVirt_q})$, the Player-$\A$ $\GF$-strategy $\s_\A(q)$ is optimal w.r.t. $Y_q$ (recall Lemma~\ref{lem:opt_in_gf_equiv_opt_in_parity}). Hence, letting $p^q := p_{q,\colVirt_q}$
		, we have:
		\begin{equation}
		\label{eqn:to_conclude_A}
		\max (\mathsf{Color}(\fsf(q),p^q,\s_\A(q),\s_\B(q)),n-1) \text{ is even }
		\end{equation}
		with $n-1$ odd
		. Recall that $\mathsf{Color}(\fsf(q),p^q,\s_\A(q),\s_\B(q)) := \{ i \in \brd \mid \outM_{\Games{\fsf(q)}{\mathbbm{1}_{(p^q)^{-1}[k_i]}}}(\s_\A(q),b) > 0 \}$. In addition, since we assume that $K^{\geq n}$ does not occur in $H$, it follows that we have $\max \mathsf{Color}(\fsf(q),p^q,\s_\A(q),\s_\B(q)) \leq n$ by Equation~(\ref{eqn:prob_ki_if_equiv_equal}) from Lemma~\ref{lem:similar_local_behavior_local_and_global_game}. We can conclude that we have $\max \mathsf{Color}(\fsf(q),p^q,\s_\A(q),\s_\B(q)) = n$. Hence, by definition, we have that 
		$\outM_{\Games{\fsf(q)}{\mathbbm{1}_{(p^q)^{-1}[k_n]}}}(\s_\A(q),b) > 0$. It follows that, by Equation (\ref{eqn:prob_n_similar_equal}) in Lemma~\ref{lem:similar_local_behavior_local_and_global_game} and since $\colVirt'$ is $n$-prevailing compared to $\colVirt_q$, we have $\Prb_{\Aconc_{\colVirt'}^{n},q}^{\s_\A,\s_\B}[\colVirt_q^{-1}[n] \setminus \{ q \}] > 0$. Since $\colVirt_q^{-1}[n] = X^i_n$ (by Equation~\ref{eqn:X_n^^i_eq_colFunc_n}), this is equivalent to: $\Prb_{\Aconc_{\colVirt'}^{n},q}^{\s_\A,\s_\B}[X_n^i] > 0$. That is, from $q$, there is a positive probability to reach a state $q' \in X_n^i$ for which, by $\mathcal{P}(i)$, we have: $\Prb_{\Aconc_{\colVirt'}^{n},q'}^{\s_\A,\s_\B}[H_i^* \cdot C_n] > 0$. It follows that $\Prb_{\Aconc_{\colVirt'}^{n},q}^{\s_\A,\s_\B}[H_{i+1}^* \cdot C_n] > 0$ and $\mathcal{P}(i+1)$ holds. We can conclude that from all states in $H$ there is a positive probability to reach $C_n$, that is $H \cap C_n \neq \emptyset$. Hence, $n$ is the highest color in $H$ and it is even. Thus, $H$ is even-colored.
		
		The case where $n$ is odd is symmetrical, from Player-$\B$'s point-of-view.
	\end{itemize}
\end{proof*}

The proof of Lemma~\ref{lem:algo_preserves_faithfulness} is now straightforward.
\begin{proof*}
	By Lemma~\ref{lem:algo_turthful_to_non_deceiving}, we have $(\colVirt,\msf{Ev}')$ non-deceiving down to $n-1$. Then, by Lemma~\ref{lem:algo_turthful_non_deceiving_to_faithful}, the pair of coloring and environment functions that Algorithm \textsf{UpdNewSta} (i.e. Algorithm~\ref{fig:UpdNewSta}) outputs is 
	faithful down to $n-1$.
\end{proof*}

\subsection{Proof of Lemma~\ref{lem:algo_increase_least_color}, Page~\pageref{lem:algo_increase_least_color}}
\label{subsubsec:proof_lem_algo_increase_least_color}
\begin{proof*}
	We let $C := \colVirt[Q]$, $k := \min C$ and we assume that $k \leq e-2$.
	
	Let us consider $(\colVirt',\msf{Ev}')$ to be the new environment and coloring functions obtained with Algorithm~\ref{fig:IncLeast} before calling Algorithm~\ref{fig:UpdNewSta}. Let us show that $(\colVirt',\msf{Ev}')$ is non-deceiving down to $n := k + 2$. First, note that $\colVirt'$ is $n$-prevailing compared to $\colVirt$. Let $Q_n := (\colVirt')^{-1}[n]$. For all $q \in Q_n$, we let $\colVirt_q$ be the virtual coloring function associated with the environment $\msf{Ev}'(q)$.
	
	\begin{itemize}
		\item[2n-d.] Let $q \in Q_n$. We have either $\colVirt(q) = n-2$ or $\colVirt(q) = n$. In both cases, since $(\colVirt,\msf{Ev})$ is completely faithful and therefore coherent at colors $n-2$ and $n$, it holds that $\colFunc(q) \leq n$, $\colVirt_q(q) \leq n$, $n_q \leq n$ (with $n_q := \newC(q,\colVirt)$). Furthermore, $\colVirt'$ is $n$-prevailing compared to $\colVirt$, which is $n$-prevailing compared to $\colVirt_q$. Finally, $n \equiv n-2 \mod 2$ and therefore this condition for being non-deceiving holds.
		\item[3n-d.] This condition for being non-deceiving also straightforwardly holds. Indeed, assume towards a contradiction that a state $q \in Q_u$ such that $\colFunc(q) = k = \colVirt'(q) < n$ is such that $\newC(q,\colVirt') = m > n$. We have $\colVirt$ that is equivalent down to $m$ to $\colVirt'$. This implies by Lemma~\ref{prop:update_color_does_not_change_if_small_color_change}, that we have $\newC(q,\colVirt) = m > n$, which is not possible since $(\colVirt,\msf{Ev})$ is assumed completely faithful.
		\item[1n-d.] Consider now the first condition for being non-deceiving. It holds for all $k > n$, since $\colVirt^{-1}[k] = (\colVirt')^{-1}[k]$, and in both games $\Zone_{\colVirt}^k$ and $\Zone_{\colVirt'}^k$, the states in $\colVirt^{-1}[\brck{0,k-1}] = \colVirt'^{-1}[\brck{0,k-1}]$ are 
		colored with $k-1$.
		
		Assume that $n$ is even. Consider a Player-$\A$ positional strategy $\s_\A$ generated by the environment function $\msf{Ev}$ in the game $\Zone^{n}_{\colVirt'}$. Let us show that this strategy parity dominates the valuation $v_{n,\colVirt'}^u$ in the game $\Zone^{n}_{\colVirt'}$. By Lemma~\ref{lem:update_with_algo_dominating}, this strategy dominates the valuation $v_{n,\colVirt'}^u$. Now, fix a Player-$\B$ positional deterministic strategy $\s_\B$ and consider a BSCC $H$ compatible with $\s_\A$ and $\s_\B$ in $\Zone^{n}_{\colVirt'}$. If $K^{\geq n}$ occurs in $H$, 
		by Lemma~\ref{lem:good_local_highest_color_even}, $H$ is even-colored. 
		Assume now that $K^{\geq n}$ does not occur in $H$. Let $S_n := (\colVirt)^{-1}[n]$ (before running the algorithm, the set colored by $n$) and $T_n := Q_n \setminus S_n = \colVirt^{-1}[n-2]$ (before running the algorithm, the set colored by $n-2$). Assume that $H \cap S_n = \emptyset$. In that case, this is a BSCC compatible with the strategy $\s_\A$ in the game $\Zone^{n-2}_{\colVirt}$ in which this strategy, by assumption, parity dominates the valuation $v_{n-2,\colVirt}^u$. Furthermore, consider the state $k_{n-1}^n$ 
		in $\Zone^{n}_{\colVirt}$ and the state $k_{n-1}^{n-2}$ in $\Zone^{n-2}_{\colVirt}$. These states are colored by $n-1$ and exactly the same edges in both games $\Zone^{n}_{\colVirt}$ and $\Zone^{n-2}_{\colVirt}$ lead to these states. (These correspond to the states $q \in Q_u$ such that $\colVirt(q) = n-1$.) 
		Hence, since the BSCC $H$ is even-colored in $\Zone^{n-2}_{\colVirt}$, it also is in $\Zone^{n}_{\colVirt}$\footnote{Note that this is why we increase the least color, and not an intermediate color. Indeed, if there were a state colored by the coloring function $\colVirt$ by $n-3$, then the state $k_{n-3}^{n-2}$ would be colored by $n-3$ in $\Zone^{n-2}_{\colVirt}$, whereas its counterpart $k_{n-3}^{n}$ would be colored by $n-1$ in $\Zone^{n}_{\colVirt}$. Hence, we could not deduce anymore that since the BSCC $H$ is even-colored in $\Zone^{n-2}_{\colVirt}$, it also is in $\Zone^{n}_{\colVirt}$.}. Assume now that $H \cap T_n = \emptyset$. In that case, this is a BSCC compatible with the strategy $\s_\A$ in the game $\Zone^{n}_{\colVirt}$ in which this strategy, by assumption, parity dominates the valuation $v_{n,\colVirt}^u$. Hence, this BSCC is even-colored. Finally, assume that $H \cap T_n \neq \emptyset$ and $H \cap S_n \neq \emptyset$. Let $C_n := Q_n \cap \colFunc^{-1}[n]$. Let us show that for all $q \in (H \cap S_n) \setminus C_n$, there is a positive probability to visit $C_n$ from $q$ while only visiting states in $H \cap S_n$, that is:
		\begin{displaymath}
		\Prb_{\Aconc_{\colVirt'}^{n},q}^{\s_\A,\s_\B}[(H \cap S_n)^* \cdot C_n] > 0
		\end{displaymath}
		Assume towards a contradiction that this does not hold. That is, there is some $q \in (H \cap S_n) \setminus C_n$ such that $\Prb_{\Aconc_{\colVirt'}^{n},q}^{\s_\A,\s_\B}[(H \cap S_n)^* \cdot C_n] = 0$. Let $Z \subseteq H \cap S_n$ be the set of states in $H \cap S_n$ such that $Z := \{ q \} \cup \{ q' \in H \cap S_n \mid \Prb_{\Aconc_{\colVirt}^{n},q}^{\s_\A,\s_\B}[(H \cap S_n)^* \cdot q'] > 0 \}$. Note that $Z \cap C_n = \emptyset$. Our goal is now to exhibit a BSCC in $\Zone^{n}_{\colVirt}$ that is compatible with $\s_\A$ and odd-colored. This would be a contradiction with the fact that $\s_\A$ parity dominates the valuation $v_{n,\colVirt}^u$ in the game $\Zone^{n}_{\colVirt}$. Consider a Player-$\B$ strategy $\s_\B'$ such that, for all $q \in S_n$, $\s_\B'(q) := \s_\B(q)$ and for all $i \in \llbracket 0,n-1 \rrbracket$, we have $\s_\B(k_i^n)$ playing in a deterministic way to reach $q$
		. By construction, we do have a BSCC $H'$ compatible with $\s_\A$ and $\s_\B'$ in $\Zone^{n}_{\colVirt}$ that is included in $Z \cup \{ k_i^n \mid i \in \llbracket 0,n-1 \rrbracket \}$. Furthermore, since the BSCC $H$ in the game $\Zone^{n}_{\colVirt'}$ is not equal to $H \cap S_n$, it must be that at least one state in $\{ k_i^n \mid i \in \llbracket 0,n-1 \rrbracket \}$ is in $H'$. Then, since each state in $\{ k_i^n \mid i \in \llbracket 0,n-1 \rrbracket \}$ is colored with $n-1$ and no state in $Z$ is in $C_n$, it follows that the highest color in $H'$ is $n-1$, which is odd. That is, the BSCC $H'$ is odd-colored. This is a contradiction. In fact, from all states $q \in (H \cap S_n) \setminus C_n$, we have $\Prb_{\Aconc_{\colVirt'}^{n},q}^{\s_\A,\s_\B}[(H \cap S_n)^* \cdot C_n] > 0$. Furthermore, the colors of all states in $H$ are at most $n$ (recall what we shown above in 2n-d). Hence, the highest color in $H$ is $n$, which is even (the color of the states in $C_n$). Therefore the BSCC $H$ is even-colored. That is, the Player-$\A$ strategy $\s_\A$ parity dominates the valuation $v_{n,\colVirt'}^u$ in the game $\Zone^{n}_{\colVirt'}$.
		
		The case where $n$ is odd is identical.
	\end{itemize}
	Hence, the pair $(\colVirt',\msf{Ev}')$ is non-deceiving down to $n$. The result then follows from Lemma~\ref{lem:algo_turthful_non_deceiving_to_faithful}.
	%
	%
\end{proof*}

\subsection{Proof of Lemma~\ref{lem:minimum_faithful_cannot_be_d_minus_one}, Page~\pageref{lem:minimum_faithful_cannot_be_d_minus_one}}
\label{subsubsec:proof_lem_minimum_faithful_cannot_be_d_minus_one}
Let us first introduce a way to build games with stopping states. 
\begin{definition}
	\label{def:parity_game_stopping_states}
	Consider a parity game $\G$. For all subsets of states $S \subseteq Q$, we denote by $\G^S$ the game $\G$ restricted to $S$ where all states in $Q \setminus S$ are stopping states with, for all $q \in Q \setminus S$, $q \gets \MarVal{\G}[q]$. (In particular, the set of states in $\G$ and $\G^S$ is the same.) 
	
	For all $x \in [0,1]$, we also let $v_S^x: Q \rightarrow [0,1]$ such that $\restr{v_S^x}{Q \setminus S} := \restr{\MarVal{\G}}{Q \setminus S}$ and $v_S^x[S] := \{ x \}$.
\end{definition}
\begin{lemma}
	\label{lem:same_value_vector_parity_game_stopping_states}
	Consider a parity game $\G$ and a subset of states $S \subseteq Q$. Then: $\MarVal{\G} = \MarVal{\G^S}$. 
\end{lemma}
\begin{proof*}
	The equality straightforwardly holds for all $q \in Q \setminus S$. Consider now a state $q \in S$. Let $\varepsilon > 0$ and let $\s_\A$ be a Player-$\A$ strategy that is $\varepsilon$-optimal from $q$ in the game $\G$. This strategy can be seen as a strategy in the game $\G^S$, the game ending as soon as the set $Q \setminus S$ is reached. Consider any Player-$\B$ strategy $\s_\B$ in the game $\G^S$. Since the games $\G$ and $\G^S$ coincide on $S^*$ (and $S^\omega$) and considering a Player-$\B$ strategy $\s_\B'$ in the game $\G$ that coincides with the strategy $\s_\B$ on $S^*$ and that plays a $\varepsilon$-optimal strategy as soon as a state in $Q \setminus S$ is reached. We have:
	\begin{align*}
	\Prb^{\s_\A,\s_\B}_{\Aconc^S,q}[W(\colFunc)] & = \Prb^{\s_\A,\s_\B}_{\Aconc^S,q}[W(\colFunc) \cap S^\omega] + \sum_{q' \in Q \setminus S} \Prb^{\s_\A,\s_\B}_{\Aconc^S,q}[S^* \cdot q'] \cdot \Prb^{\s_\A,\s_\B}_{\Aconc,q'}[W(\colFunc)] \\
	& = \Prb^{\s_\A,\s_\B'}_{\Aconc,q}[W(\colFunc) \cap S^\omega] + \sum_{q' \in Q \setminus S} \Prb^{\s_\A,\s_\B'}_{\Aconc,q}[S^* \cdot q'] \cdot \MarVal{\G}[q'] \\
	& \geq \Prb^{\s_\A,\s_\B'}_{\Aconc,q}[W(\colFunc) \cap S^\omega] + \sum_{q' \in Q \setminus S} \Prb^{\s_\A,\s_\B'}_{\Aconc,q}[S^* \cdot q'] \cdot (\Prb^{\s_\A,\s_\B'}_{\Aconc,q'}[W(\colFunc)] - \varepsilon) \\
	& \geq \Prb^{\s_\A,\s_\B'}_{\Aconc,q}[W(\colFunc)] - \varepsilon \geq \MarVal{\G}[q'] - 2\varepsilon
	\end{align*}
	As this holds for all Player-$\B$ strategies $\s_\B$, it follows that.$\MarVal{\G^S}[\s_\A](q) \geq \MarVal{\G}(q) - 2\varepsilon$. As this holds for all $\varepsilon > 0$, we have $\MarVal{\G^S}(q) \geq \MarVal{\G^S}[\s_\A](q) \geq \MarVal{\G}(q)$. By symmetry of the players, we have  $\MarVal{\G^S}(q) = \MarVal{\G}(q)$.
\end{proof*}

We can proceed to the proof of Lemma~\ref{lem:minimum_faithful_cannot_be_d_minus_one}. 
\begin{proof*}
	Assume towards a contradiction that $\min C = e-1$. Let $Q_{e-1} := \colVirt^{-1}[e-1]$. Consider a Player-$\B$ strategy $\s_\B$ generated by the environment function $\msf{Ev}$. Because the pair $(\colVirt,\msf{Ev})$ is completely faithful and $e-1$ is odd, it follows that $\s_\B$ parity dominates the valuation $v_{e-1,\colVirt}^{u'}$ in the game $\Zone^{e-1}_{\colVirt}$ for some $u' \leq u$. Consider the game $\G^{Q_{e-1}}$ from Definition~\ref{def:parity_game_stopping_states} (we do not mention any player since the game has a value). 
	Lemma~\ref{lem:same_value_vector_parity_game_stopping_states} gives that all states have the same values in $\G$ and $\G^{Q_{e-1}}$.	Let us show that there is some $x \in [0,1]$ such that $x < u$ and such that the strategy $\s_\B$ parity dominates the valuation $v_{Q_{e-1}}^x$ in the game $\G^{Q_{e-1}}$. First, note that for all $y \in [0,1]$, if $\s_\B$ dominates the valuation $v_{Q_{e-1}}^y$, then it also parity dominates it (since $\s_\B$ parity dominates the valuation $v_{e-1,\colVirt}^{u'}$ in the game $\Zone^{e-1}_{\colVirt}$). 
	
	Now, consider some state $q \in Q_{e-1}$. Let 
	$\colVirt_q$ be the coloring function associated with the environment $\msf{Ev}(q)$. Denoting $Y_q := (\fsf(q),\msf{Ev}(q))$ and $p^q := p_{q,\colVirt_q}$, we have that the Player-$\B$ $\GF$-strategy $\s_\B(q)$ is optimal w.r.t. $Y_q$. Hence, by Lemma~\ref{lem:opt_in_gf_equiv_opt_in_parity} for Player $\B$, we have that for all Player-$\A$ actions $a \in \Act_\A$, if $\outM_{\Games{\formNF}{\mathbbm{1}_{(p^q)^{-1}[p^q_{[0,1]}]}}}(a,\s_\B(q)) = 0$ 
	then 
	$\mathsf{Color}(\fsf(q),p^q,a,\s_\B(q))
	\cup \{ e-2 \}$ is odd with $\mathsf{Color}(\fsf(q),p^q,a,\s_\B(q)) := \{ i \in \brd \mid \outM_{\Games{\fsf(q)}{\mathbbm{1}_{(p^q)^{-1}[k_i]}}}(a,\s_\B(q)) > 0 \}$. Since $e$ is even, it must be that, for all states $q \in Q_{e-1}$ and $a \in \Act_\A^q$:
	\begin{equation}
	\label{eqn:imply_not_e}
	\outM_{\Games{\formNF}{\mathbbm{1}_{(p^q)^{-1}[p^q_{[0,1]}]}}}(a,\s_\B(q)) = 0 \Rightarrow \outM_{\Games{\fsf(q)}{\mathbbm{1}_{(p^q)^{-1}[k_e]}}}(a,\s_\B(q)) = 0
	\end{equation}
	
	For $q \in Q_{e-1}$, let $\msf{NZ}(q) := \{ a \in \Act_\A^q \mid \outM_{\Games{\fsf(q)}{\mathbbm{1}_{(p^q)^{-1}[p^q_{[0,1]}]}}}(a,\s_\B(q)) > 0 \}$. Then, we let:
	\begin{equation}
	\label{eqn:minimum_proba_exit}
	p_m := \min_{q \in Q_{e-1}} \min_{a \in \msf{NZ}(q)} \outM_{\Games{\formNF}{\mathbbm{1}_{(p^q)^{-1}[p^q_{[0,1]}]}}}(a,\s_\B(q)) > 0
	\end{equation}
	and
	\begin{equation}
	\label{eqn:x}
	x := u + p_m \cdot (u' - u) < u
	\end{equation}
	
	Let us show that the strategy $\s_\B$ dominates the valuation $v_{Q_{e-1}}^x$ in the game $\G^{Q_{e-1}}$. Let $q \in Q_{e-1}$. We have $\va{\Games{\fsf(q)}{v}}(\s_\B(q)) = \max_{a \in \Act_\A^q} \outM_{\Games{\fsf(q)}{v}}(a,\s_\B(q))$. Consider some $a \in \Act_\A^q$ and a Player-$\A$ strategy $\s_\A$ with $\s_\A(q) := a$. 
	Denoting $P_{a}[T] := \mathbb{P}_{\Aconc_{\colVirt}^{e-1},q}^{\s_\A,\s_\B}[T]$ for any $T \subseteq Q_{e-1} \cup \{ k^{e-1}_{e} \} \cup V_{Q \setminus Q_u}$, we have:	
	\begin{align*}
	\outM_{\Games{\fsf(q)}{v_{Q_{e-1}}^x}}(a,\s_\B(q)) & = x \cdot \outM_{\Games{\fsf(q)}{\mathbbm{1}_{Q_{e-1}}}}(a,\s_\B(q)) \\
	& + u \cdot \outM_{\Games{\fsf(q)}{\mathbbm{1}_{Q_{e}}}}(a,\s_\B(q)) \\
	& + \sum_{y \in V_{Q \setminus Q_u}} y \cdot \outM_{\Games{\fsf(q)}{\mathbbm{1}_{Q_y}}}(a,\s_\B(q)) \\
	& = x \cdot P_{a}[Q_{e-1}]
	+ u \cdot 
	P_{a}[\{ k_e^{e-1} \}] + \sum_{y \in V_{Q \setminus Q_u}} y \cdot 
	P_{a}[y]
	\end{align*}
	
	In addition, since the strategy $\s_\B$ dominates the valuation $v^{u'}_{e-1,\colVirt}$ in the game $\Zone^{e-1}_{\colVirt}$, we have:
	\begin{align*}
	u' & \geq \outM_{\Games{\fsf(q)}{v^{u'}_{e-1,\colVirt}}}(a,\s_\B(q)) \\
	& = u' \cdot 
	P_a[Q_{e-1} \cup \{ k_e^{e-1} \}]
	+ \sum_{y \in V_{Q \setminus Q_u}} y \cdot 
	P_a[y]
	\end{align*}
	That is:
	\begin{displaymath}
	\sum_{y \in V_{Q \setminus Q_u}} y \cdot P_a
	[y] \leq u' \cdot P_a
	[V_{Q \setminus Q_u}]
	\end{displaymath}
	
	By Equation (\ref{eqn:prob_in_V_Q_Q_u_equal}) in Lemma~\ref{lem:similar_local_behavior_local_and_global_game}, for all $y \in V_{Q \setminus Q_u}$, we have:
	\begin{displaymath}
	P_a[y] = \outM_{\Games{\fsf(q)}{\mathbbm{1}_{(p^q)^{-1}[y]}}}(a,\s_\B(q))
	\end{displaymath}
	Furthermore, by Equation~(\ref{eqn:prob_ki_if_equiv_equal}) in Lemma~\ref{lem:similar_local_behavior_local_and_global_game} which can be applied since $(\colVirt,\msf{Ev})$ is faithful down to $e-1$, it also is coherent at $e-1$, we have:
	\begin{displaymath}
	P_a[\{ k^e_{e-1} \}] = \outM_{\Games{\fsf(q)}{\mathbbm{1}_{(p^q)^{-1}[k_e]}}}(a,\s_\B(q))
	\end{displaymath}
	
	Now, assume that $\outM_{\Games{\fsf(q)}{\mathbbm{1}_{(p^q)^{-1}[V_{Q \setminus Q_u}]}}}(a,\s_\B(q)) = 0$
	. That is, for all $y \in V_{Q \setminus Q_u}$, we have $P_a[y] = 0$. By Equation~(\ref{eqn:imply_not_e}) it follows that $P_a[\{ k^e_{e-1} \}] = 0$. Hence: $v_{Q_{e-1}}^x(q) = x = \outM_{\Games{\fsf(q)}{v_{Q_{e-1}}^x}}(a,\s_\B(q))$. 
	
	Assume now that $\outM_{\Games{\fsf(q)}{\mathbbm{1}_{(p^q)^{-1}[V_{Q \setminus Q_u}]}}}(a,\s_\B(q)) > 0$ (i.e. $a \in \msf{NZ}(q)$). By Equation~(\ref{eqn:minimum_proba_exit}), it implies that $P_a[V_{Q \setminus Q_u}] \geq p_m$. It follows that, by Equation~(\ref{eqn:x}):
	\begin{align*}
	\outM_{\Games{\fsf(q)}{v_{Q_{e-1}}^x}}(a,\s_\B(q)) & = x \cdot P_a[Q_{e-1}] + u \cdot P_a[\{ k_{e-1}^e \}] + \sum_{y \in V_{Q \setminus Q_u}} y \cdot P_a[y] \\
	& \leq u \cdot P_a[Q_{e-1}] + u \cdot P_a[\{ k_{e-1}^e \}] + u' \cdot P_a[V_{Q \setminus Q_u}] \\
	& = u \cdot (1 - P_a[V_{Q \setminus Q_u}]) + u' \cdot P_a[V_{Q \setminus Q_u}] \\
	& = u + P_a[V_{Q \setminus Q_u}] \cdot (u' - u) \leq u + p_m \cdot (u' - u) \\
	& = x = v_{Q_{e-1}}^x(q) 
	\end{align*}
	
	As this holds for all $x \in Q_{e-1}$, it follows that the Player-$\B$ strategy $\s_\B$ dominates the valuation $v^x_{Q_{e-1}}$, and therefore parity dominates it. Hence, 
	for all $q \in Q_{e-1}$, we have $\MarVal{\G^{Q_{e-1}}}[q] \leq x < u$. That is a contradiction with Lemma~\ref{lem:same_value_vector_parity_game_stopping_states}. 
\end{proof*}


\subsection{Proof of Lemma~\ref{lem:metric_work}, Page~\pageref{lem:metric_work}}
\label{subsubsec:proof_lem_metric_work}
First, let us formally define a way to compare coloring functions
.
\begin{definition}
	We define the (transitive) relation $\prec_{\brd}$ on coloring functions which corresponds to the lexicographic order. Specifically, for all virtual coloring functions $\colVirt_1,\colVirt_2: Q_u \rightarrow \brd$, we have $\colVirt_1 \prec_{\brd} \colVirt_2$ if and only if there is some $k \in \brd$ such that:
	\begin{itemize}
		\item for all $i \in \brck{k+1,e}$, we have $|\colVirt_1^{-1}[i]| = |\colVirt_2^{-1}[i]|$;
		\item $|\colVirt_1^{-1}[k]| < |\colVirt_1^{-1}[k]|$.
	\end{itemize}
	
	We write $\colVirt_1 =_{\brd} \colVirt_2$ when we have neither $\colVirt_1 \prec_{\brd} \colVirt_2$ or $\colVirt_1 \prec_{\brd} \colVirt_2$. We also write $\colVirt_1 \preceq_{\brd} \colVirt_2$ for $\colVirt_1 =_{\brd} \colVirt_2$ or $\colVirt_1 \prec_{\brd} \colVirt_2$.
\end{definition}

Straightforwardly, this relation ensures the following proposition.
\begin{proposition}
	\label{prop:no_infinite_sequence}
	Let $n \in \N$. There is no infinite sequence $(c_k)_{k \in \N} \in (\brck{0,e}^n)^\N$ such that $c_k \prec_{\brd} c_{k+1}$ for all $k \in \N$.
\end{proposition}
\begin{proof*}
	This is because the relation $\prec_{\brd}$ corresponds to a lexicographic order on vectors taking their values in $\brd$. 
\end{proof*}

We can then consider all steps of Algorithm~\ref{fig:ComputeEnv} one by one.
\begin{lemma}
	\label{lem:step_increase_order}
	Consider a virtual coloring function $\colVirt: Q_u \rightarrow \brd$  and some $k \in \brd$. For all partial environment functions $\msf{Ev}: Q_u \rightarrow \msf{Env}(Q)$ defined on $\brck{k+1,e}$:
	\begin{itemize}
		\item  for $(\colVirt_\msf{upd},\msf{Ev}_\msf{upd}) \gets \msf{UpdateColEnv}(k,\colVirt,\msf{Ev})$, we have $\colVirt \preceq_{\brd} \colVirt_\msf{upd}$;
	\end{itemize}
	For all environment functions $\msf{Ev}: Q_u \rightarrow \msf{Env}(Q)$:
	\begin{itemize}
		\item for $(\colVirt_\msf{Inc},\msf{Ev}_\msf{Inc}) \gets \msf{IncLeast}(\colVirt,\msf{Ev})$, we have $\colVirt \prec_{\brd} \colVirt_\msf{Inc}$;
	\end{itemize}
\end{lemma}	
\begin{proof*}
	The procedure $\msf{UpdateColEnv}$ consists in two part. First, the call Algorithm~\ref{fig:UpdCurSta}. This does not change the coloring function. Then, there is the call to Algorithm~\ref{fig:UpdNewSta}. By Proposition~\ref{prop:update_color_at_least_preivous_color}, if $\newC(q,\colVirt) = k$, then $\colVirt(q) \leq k$. Hence, the change of colors of the states does not decrease the colors of the states. It follows that the resulting coloring function $\colVirt_\msf{upd}$ is such that $\colVirt \preceq_{\brd} \colVirt_\msf{upd}$ 
	.
	
	Consider now the procedure $\msf{IncLeast}$. By definition, before the call of this algorithm, some states are colored with $c_{\min}$. Hence, the coloring function resulting of the call to Algorithm~\ref{fig:IncLeast} has more states colored with $c_{\min} + 2$, whereas states colored with a color higher than $c_{\min} + 2$ are left unchanged. Hence: $\colVirt \prec_{\brd} \colVirt_\msf{Inc}$.
\end{proof*}

We can now proceed to the proof of Lemma~\ref{lem:metric_work}.
\begin{proof*}
	By Lemma~\ref{lem:step_increase_order}, each call to the procedure $\msf{IncLeast}$ strictly increase (w.r.t. $\prec_{\brd}$) the coloring function. Furthermore, each call to the procedure $\msf{UpdateColEnv}$ does not decrease (w.r.t. $\prec_{\brd}$) the coloring function. Hence, by Proposition~\ref{prop:no_infinite_sequence}, the call to the procedure $\msf{IncLeast}$ is done only finitely many times.
\end{proof*}

\section{Proofs from Section~\ref{sec:positionally_max_gf}}
\subsection{Proof of Proposition~\ref{prop:determined_two_var_ok}, Page~\pageref{prop:determined_two_var_ok}}
\label{proof:prop_determined_two_var_ok}
First, let us formally recall the definition of deterministic and of determined game forms.
\begin{definition}[Deterministic and determined game forms]
	\label{def:determined_game_form}
	Consider a set of outcomes $\outComeNF$. A game form $\formNF = \langle \Act_\A,\Act_\B,\outComeNF,\outCNF \rangle \in \Frm(\outComeNF)$ is \emph{deterministic} if, for all $\sigma_\A \in \Dist(\Act_\A)$ and $\sigma_\B \in \Dist(\Act_\B)$, the probability distribution $\outCNF(\sigma_\A,\sigma_\B)$ is deterministic.
	
	Furthermore, the game form $\formNF \in \Frm(\outComeNF)$ is \emph{determined} if it is deterministic and, for all $v: \outComeNF \rightarrow \{ 0,1 \}$, either there is some $a \in \Act_\A$ such that $\outCNF(a,\Act_\B) = \{ 1 \}$ or there is some $b \in \Act_\B$ such that $\outCNF(\Act_\A,b) = \{ 0 \}$.
\end{definition}
\begin{proof}
	In \cite{BBSFSTTCS21}, it was shown that in all locally determined concurrent games (i.e. concurrent games where game forms are determined), both players have positional optimal strategies. Hence, determined game forms are positionally optimizable. 
	
	Consider now a game form $\formNF = \langle \Act_\A,\Act_\B,\distribSet,\outCNF \rangle$ such that $|\outComeNF| = 2$. Let $\outComeNF := \{ x,y \}$. Consider an environment $E = \langle c,e,p \rangle$ and the induced parity game $\G_Y$ for $Y := (\formNF,E)$. If $p(x),p(y) \in [0,1]$, then the game cannot loop and both players have optimal positional strategies. Assume that $p(x) \in [0,1]$ and $p(y) \notin [0,1]$. If there is a row of $y$ and $y$ is mapped w.r.t. $p$ to an even integer (if it is mapped to $\qinit$, we consider the integer $c$), then Player $\A$ can surely win by playing the corresponding row. Similarly, if there is a column of $y$ and $y$ is mapped w.r.t. $p$ to an odd integer, then Player $\B$ can surely win by playing the corresponding column. Otherwise, both players can play a uniform distribution over all actions that do not only lead to $y$, thus ensuring that the value of both strategies (and of the game) is equal to $p(x)$. The reasoning is identical if $p(x) \notin [0,1]$ and $p(y) \in [0,1]$. Finally, assume that $p(x),p(y) \notin [0,1]$ and assume that $x$ is mapped w.r.t. $p$ to the highest integer among the two. If there is a row of $y$ and $y$ is mapped to an even integer, then Player $\A$ can surely win by playing the corresponding row. Similarly, if there is a column of $y$ and $y$ is mapped to an odd integer, then Player $\B$ can surely win by playing the corresponding column. Otherwise, one of the players wins almost surely by playing a uniform probability distribution over all actions: if $x$ is mapped to an even integer, then Player $\A$ wins almost-surely, otherwise, Player $\B$ wins almost-surely. 
\end{proof}

\subsection{Proof of Proposition~\ref{prop:relevant_env_sufficient}, Page~\pageref{prop:relevant_env_sufficient}}
\label{appen:proof_prop_relevant_env_sufficient}
We prove the result for Player $\A$, the arguments are similar for Player $\B$. In the remainder of this section, we fix a set of outcomes $\outComeNF$ and a game form $\formNF = \langle \Act_\A,\Act_\B,\outComeNF,\outCNF \rangle$.
\begin{definition}
	\label{def:relevant_env_from_arbitrary_env}
	Consider some $n \geq 1$ and an (arbitrary) environment $E = \langle c,e,p \rangle \in \msf{Env}(\outComeNF)$ with $p: \outComeNF \rightarrow \{ \qinit \} \cup [0,1] \cup K_e$ with $\msf{Sz}_\A(E) = n$. Assume that no outcome $o \in \outComeNF$ is such that $p(o) = \qinit$. Let $\tilde{e} := \msf{Even}(e) \geq 2$. By definition, we have $n = \tilde{e} - c$. Let us define a relevant environment $E' = \langle c',e',p' \rangle \in \msf{Env}(\outComeNF)$ of size $n-1$. 
	
	We let:
	\begin{equation*}
	c' := \begin{cases}
	0 & \text{ if }c \text{ is even } \\
	1 & \text{ otherwise }
	\end{cases}
	\end{equation*}
	
	We also let $\msf{Im}_p := \{ n \in \brck{c+1,\tilde{e}-1} \mid \exists o \in \outComeNF,\; p(o) = k_i \}$ and $c = a_0 < a_1 < a_2 < \ldots < a_k$ be such that $\msf{Im}_p = \{ a_1,a_2,\ldots,a_k \}$. In particular, we have $\tilde{e} - 1 \geq c + k$. Let us define the function $f_E: \{ a_0,\ldots,a_k \} \rightarrow \N$ such that, letting $f_E(a_0) := c'$, for all $1 \leq i \leq k$: 
	\begin{equation*}
	f_E(a_i) := \begin{cases}
	f_E(a_{i-1}) & \text{ if }a_{i-1} \equiv a_i \mod 2 \\
	f_E(a_{i-1}) + 1 & \text{ otherwise }
	\end{cases}
	\end{equation*}
	Since for all $1 \leq i \leq k$, we have $f_E(a_i) \leq f_E(a_{i-1}) + 1$, we have $f_E(a_k) \leq c' + k$. We set $e' := f_E(a_k)$. Note also that, for all $1 \leq i \leq k$, we have $f(a_k) \geq c'$. With this choice, we have $\msf{Sz}(E') = e' - c' \leq k \leq \tilde{e}-1-c = \msf{Sz}_\A(E)-1 = n-1$.
	
	We can now define the function $p'$. We let $u := \MarVal{\G_{Y}}(\qinit)$ for $Y := (\formNF,E)$. Then, for all $o \in \outComeNF$, we let:
	\begin{equation*}
	p'(o) := \begin{cases}
	p(o) \in [0,1] & \text{ if }p(o) \in [0,1] \\
	u \in [0,1] & \text{ if }p(o) = \tilde{e} \\
	k_{f_E(\max(c,n))} \in K_{e'} & \text{ if }p(o) = k_n\in K_{\tilde{e}-1}
	\end{cases}
	\end{equation*}
\end{definition}

Let us now show a useful property about the function $f_E$ from Definition~\ref{def:relevant_env_from_arbitrary_env}. 
\begin{lemma}
	\label{lem:same_parity}
	Consider an (arbitrary) environment $E = \langle c,e,p \rangle \in \msf{Env}(\outComeNF)$ and the function $f_E: \{ a_0,\ldots,a_{k} \} \rightarrow \N$ from Definition~\ref{def:relevant_env_from_arbitrary_env}. It ensures the following: for all $0 \leq i \leq k$, $a_i$ and $f_E(a_i)$ have the same parity. In addition, for all $X \subseteq \{ a_0,\ldots,a_{k} \}$, $\max X$ has the same parity as $\max f_E(X)$. 
\end{lemma}
\begin{proof}
	Let us show by induction on $i \in \brck{0,k}$ the following property $\mathcal{P}(i)$: $a_i$ and $f(a_i)$ have the same parity
	. The property $\mathcal{P}(0)$ straightforwardly holds since $a_0 = c$ and $f_E(a_0) = c'$ have the same parity. Assume now that $\mathcal{P}(i-1)$ holds for some $1 \leq i \leq k$. If $a_i$ has the same parity as $a_{i-1}$, then $f_E(a_i) = f_E(a_{i-1})$ which has the same parity as $a_{i-1}$. Similarly, if $a_i$ does not have the same parity as $a_{i-1}$, then $f_E(a_i) = f_E(a_{i-1}) + 1$. Hence, $a_i$ and $f(a_i)$ have the same parity. 
	Overall, the property $\mathcal{P}(i)$ holds for all $i \in \brck{0,k}$. The second part of the lemma comes from the fact that $f_E$ is non-decreasing.
\end{proof}

Let us now show that the value of local parity games with an arbitrary environment $E$ is equal to the value of the parity game with the corresponding relevant environment.
\begin{lemma}
	\label{prop:relevant_env_no_decrease_value}
	Consider an (arbitrary) environment $E = \langle c,e,p \rangle \in \msf{Env}(\outComeNF)$ and the relevant environment $E' = \langle c',e',p' \rangle$ from Definition~\ref{def:relevant_env_from_arbitrary_env}. Let $Y := (\formNF,E)$ and $Y' := (\formNF,E')$. Then, $\MarVal{\G_Y}(q_\msf{init}) \leq \MarVal{\G_{Y'}}(q_\msf{init})$.
\end{lemma}

To prove this result we will use a result whose proof is not yet published. Specifically, consider a game $\G = \Games{\Aconc}{W}$ (with $W \subseteq Q^\omega$ an arbitrary  Borel winning objective), a finite set of labels $\msf{Label}$ and a labeling function $\msf{lb}: Q \rightarrow \msf{Label}$. Assume that $W$ is $\msf{lb}$-uniform, that is, for all $\rho,\rho' \in Q^\omega$, if for all $i \in \N$, we have $\msf{lb}(\rho_i) = \msf{lb}(\rho'_i)$, then $\rho \in W \Leftrightarrow \rho' \in W$. In that case, the value of the game $\G$ can be approximated (i.e. the value can approached arbitrarily closely) by label-strategies (for both players), that is strategies that only depend on the labels of the states visited (and the current state). For Player $\A$, such a strategy $\s_\A$ would satisfy that, for all $q \in Q$, for all $n \in \N$ and $\rho,\rho' \in Q^n$ such that for all $0 \leq i \leq n$, we have $\msf{lb}(\rho_i) = \msf{lb}(\rho'_i)$, then $\s_\A(\rho \cdot q) = \s_\A(\rho' \cdot q)$. This result follows from an adaptation of Martin's proof of the determinacy of Blackwell games and should be submitted in the next few months. We will explicitly mention in the next proof when we use this result.

\begin{proof}
	We let $Q_y$ and $q_Y \in Q_Y$ (resp. $Q_{Y'}$ and $q_{Y'}$) denote the set of non-stopping states in $\Aconc_Y$ and the state $\qinit \in Q_Y$ (resp. $\Aconc_{Y'}$ and the state $\qinit \in Q_{Y'}$). 
	
	Recall the function $f_E$ from Definition~\ref{def:relevant_env_from_arbitrary_env}. We let $f := f_E$ and we extend it such that $f(\tilde{e}) := d$ with $d$ the smallest integer of the same parity than $e$ such that $d \geq e'+1$. Then, we define an alternate game $\G_Y^f = \Games{\Aconc_Y^f}{\colFunc^f}$ which differs from the game $\G_Y$ only in the coloring function and in the objective. The coloring function $\colFunc^f$ in the arena $\Aconc_Y^f$ is such that, for all states $q \in Q_Y$, we have:
	\begin{equation*}
	\colFunc^f(q) := f(\max(c,\colFunc(q)))
	\end{equation*}
	where $\colFunc$ refers to the coloring function in the game $\G_Y$. Let us show that the value of both games $\G_Y^f$ and $\G_Y$ is the same. Consider a pair of strategies $(\s_\A,\s_\B) \in \msf{S}_\A^{\Aconc_Y} \times \msf{S}_\B^{\Aconc_Y} = \msf{S}_\A^{\Aconc_Y^f} \times \msf{S}_\B^{\Aconc_Y^f}$. First, we have:
	\begin{equation*}
	\sum_{x \in [0,1]} \Prb_{\Aconc_Y,q_Y}^{\s_\A,\s_\B}[Q_Y^* \cdot x] \cdot x = \sum_{x \in [0,1]} \Prb_{\Aconc_Y^f,q_Y}^{\s_\A,\s_\B}[Q_Y^* \cdot x] \cdot x
	\end{equation*}
	In addition, since all states $k_i$ loop back to $q_Y$ in $\Aconc_Y$ and $\Aconc_Y^f$, we have:
	\begin{equation*}
	\Prb_{\Aconc_Y,q_Y}^{\s_\A,\s_\B}[Q_Y^\omega \setminus (Q_Y^* \cdot q_Y)^\omega]  = 0 = \Prb_{\Aconc_Y^f,q_Y}^{\s_\A,\s_\B}[Q_Y^\omega \setminus (Q_Y^* \cdot q_Y)^\omega]
	\end{equation*}
	
	Furthermore, consider some infinite path $\rho \in (Q_Y^* \cdot q_Y)^\omega$ visiting infinitely often the central state $q_Y$. Let $X_\rho := \{ \colFunc[\msf{InfOft}(\rho)] \cap \brck{c,e} \}$ denote the set of colors, at least $c$, seen infinitely often in $\rho$ in $\G_Y$. Note that $X_\rho \neq \emptyset$ since $\colFunc(q_Y) = c$. We also let $X_\rho^f := \{ \colFunc^f[\msf{InfOft}(\rho)] \cap \brck{c',d} \}$ denote the set of colors, at least $c'$, seen infinitely often in $\rho$ in $\G_Y^f$. By definition of $\colFunc^f$, we have $X_\rho^f = f[X_\rho]$. Therefore, we have, by Lemma~\ref{lem:same_parity} and by definition of $d$:
	\begin{align*}
	\max X_\rho \text{ is even } \Leftrightarrow \max f(X_\rho) \text{ is even } \Leftrightarrow \max X_\rho^f \text{ is even } \\
	\end{align*}
	It follows that:
	\begin{equation*}
	\Prb_{\Aconc_Y,q_Y}^{\s_\A,\s_\B}[Q_Y^\omega \cap W_Y]  = \Prb_{\Aconc_Y^f,q_Y}^{\s_\A,\s_\B}[(Q_Y^f)^\omega \cap W_Y^f]
	\end{equation*}
	
	As this holds for all pair of strategies $(\s_\A,\s_\B) \in \msf{S}_\A^{\Aconc_Y} \times \msf{S}_\B^{\Aconc_Y} = \msf{S}_\A^{\Aconc_Y^f} \times \msf{S}_\B^{\Aconc_Y^f}$, it follows that $\MarVal{\G_Y}(q_Y) = \MarVal{\G_Y^f}(q_Y)$.
	
	Let us now relate the games $\G_Y^f$ and $\G_{Y'}$. 
	By using the unpublished result mentioned prior to this proof, one can show that almost-optimal Player-$\A$ strategies $\s_\A$ in $\G_{Y}^f$ can be chosen such that, for all $\rho,\rho' \in Q_Y^*$, if for all $i \in \N$, we have $\colFunc^f(\rho_i) = \colFunc^f(\rho'_i)$, then for all $q \in Q_Y$, we have $\s_\A(\rho \cdot q) = \s_\A(\rho' \cdot q)$. Such strategies are called $(\brck{c',d},\colFunc^f)$-uniform strategies. Similarly, in $\G_Y$, almost-optimal strategies (for both players) can be found among  $(\brck{c',e'},\colFunc')$-uniform strategies in $\G_{Y'}$. Furthermore, any strategy $(\brck{c',d},\colFunc^f)$-uniform strategy in $\Aconc_Y^f$ can be seen as a ($(\brck{c',e'},\colFunc')-$)strategy in $\Aconc_{Y'}$. On the other hand, a $(\brck{c',e'},\colFunc')$-uniform strategy in $\Aconc_{Y'}$ can be seen as a ($(\brck{c',d},\colFunc^f)-$)strategy in $\Aconc_Y^f$ up to defining what it does once the state $k_{\tilde{e}}$ --- colored by $d$ --- occurs.
	
	Consider any Player-$\A$ strategy $\s_\A \in \msf{S}_\A^{\Aconc_Y^f}$ that is $(\brck{c',d},\colFunc^f)$-uniform in $\Aconc_Y^f$. As mentioned above, it can be seen as a Player-$\A$ strategy in $\Aconc_{Y'}$. Let $\delta > 0$ and consider a $(\brck{c',e'},\colFunc')$-uniform Player-$\B$ strategy in $\Aconc_{Y'}$. Consider the Player-$\B$ strategy $\s_\B^\delta \in \msf{S}_\B^{\Aconc_Y^f}$ that mimics the strategy $\s_\B$ as long as $k_{\tilde{e}}$ is not seen. Once it is seen, the Player-$\B$ strategy $\s_\B^\delta$ switches to a $\delta$-optimal strategy in the game $\G_Y^f$.
	By definition of $\s_\B^\delta$, and since $u = \MarVal{\G_Y}(q_Y) = \MarVal{\G_Y^f}(q_Y)$, we have:
	\begin{equation*}
	\Prb^{\s_\A,\s_\B^\delta}_{\Aconc_{Y}^f,q_{Y}}[W(\colFunc^f) \cap (Q_{Y}^f)^* \cdot k_{\tilde{e}}] \leq
	\Prb^{\s_\A,\s_\B^\delta}_{\Aconc_{Y}^f,q_{Y}}[(Q_{Y}^f)^* \cdot k_{\tilde{e}}] \cdot (u + \delta)
	\end{equation*}
	Since the outcomes leading to $k_{\tilde{e}}$ in $\Aconc_y^f$ lead to the stopping state $u \in [0,1]$ in $\Aconc_{Y'}$, and all outcomes leading to a stopping state in $\Aconc_y^f$ lead to the same stopping state in $\Aconc_{Y'}$, it follows that: 
	\begin{equation*}
	\Prb_{\Aconc_Y^f,q_Y}^{\s_\A,\s_\B^\delta}[W(\colFunc^f) \cap Q_{Y}^* \cdot (k_{\tilde{e}} \cup [0,1])] 
	\leq \Prb_{\Aconc_{Y'},q_{Y'}}^{\s_\A,\s_\B}[W(\colFunc') \cap Q_{Y'}^* \cdot [0,1]] + \delta
	\end{equation*}
	In addition, we have:
	\begin{equation*}
	\Prb_{\Aconc_Y^f,q_Y}^{\s_\A,\s_\B^\delta}[(Q_{Y} \setminus \{ k_{\tilde{e}} \})^\omega] = \Prb_{\Aconc_{Y'},q_{Y'}}^{\s_\A,\s_\B}[(Q_{Y'})^\omega]
	\end{equation*}
	In addition, consider any outcome $o \in \outComeNF$ such that $p'(o) \notin [0,1]$ or equivalently such that $p(o) \notin [0,1] \cup \{ k_{\tilde{e}} \}$. Then, we have $p'(o) = k_{f(\max(c,n))}$ where $n \in \brck{0,\tilde{e}-1}$ is such that $p(o) = k_n$. Then, we have $\colFunc' \circ p'(o) = f(\max(c,n)) = \colFunc^f \circ p(o)$. Furthermore, we have $\colFunc'(q_{Y'}) = c' = \colFunc^f(q_Y)$. It follows that, assuming that $[0,1] \cup \{ k_{\tilde{e}} \}$ is not reached in $\Aconc_Y^f$ and that $[0,1]$ is not reached in $\Aconc_{Y'}$, the colors seen with $\s_\A$ and $\s_\B^\delta$ are the same as the colors seen with $\s_\A$ and $\s_\B$ in $\Aconc_Y^f$ (because the strategies $\s_\A,\s_\B,\s_\B^\delta$ only depend on the colors seen). Therefore: 
	\begin{equation*}
	\Prb_{\Aconc_Y^f,q_Y}^{\s_\A,\s_\B^\delta}[W(\colFunc^f) \cap (Q_{Y} \setminus \{ k_{\tilde{e}} \})^\omega] = \Prb_{\Aconc_{Y'},q_{Y'}}^{\s_\A,\s_\B}[W(\colFunc') \cap (Q_{Y'})^\omega]
	\end{equation*}
	
	Overall, we obtain that:
	\begin{align*}
	\MarVal{\G_Y}[\s_\A](q_Y) \leq \Prb_{\Aconc_Y^f,q_Y}^{\s_\A,\s_\B^\delta}[W(\colFunc^f)] & = \Prb_{\Aconc_Y^f,q_Y}^{\s_\A,\s_\B^\delta}[W(\colFunc^f) \cap (Q_{Y})^* \cdot (k_{\tilde{e}} \cup [0,1])] \\
	& + \Prb_{\Aconc_Y^f,q_Y}^{\s_\A,\s_\B^\delta}[W(\colFunc^f) \cap (Q_{Y} \setminus ([0,1] \cup \{ k_{\tilde{e}} \}))^\omega] \\
	& \leq \Prb_{\Aconc_{Y'},q_{Y'}}^{\s_\A,\s_\B}[W(\colFunc') \cap (Q_{Y'})^* \cdot [0,1]] + \delta \\
	& + \Prb_{\Aconc_{Y'},q_{Y'}}^{\s_\A,\s_\B}[W(\colFunc') \cap (Q_{Y'})^\omega] \\
	& = \Prb_{\Aconc_{Y'},q_{Y'}}^{\s_\A,\s_\B^\delta}[W(\colFunc')] + \delta
	\end{align*}
	Since this holds for all 
	Player-$\B$  $(\brck{c',e'},\colFunc')$-uniform strategies, 
	and for all $\delta > 0$, it follows that $\MarVal{\G_Y^f}[\s_\A](q_Y) \leq \MarVal{\G_{Y'}}[\s_\A](q_{Y'}) + \delta \leq \MarVal{\G_{Y'}}(q_{Y'}) + \delta$. Since this holds for all Player-$\A$ $(\brck{c,d},\colFunc)$-uniform strategies $\s_\A \in \msf{S}_\A^{\Aconc_Y}$, 
	it follows that $\MarVal{\G_Y}(q_Y) = \MarVal{\G_Y^f}(q_Y) \leq \MarVal{\G_{Y'}}(q_{Y'})$.
\end{proof}

We can now proceed to the proof of Proposition~\ref{prop:relevant_env_sufficient}.
\begin{proof}
	Consider some $n \in \N$. First, consider the case where $n = 0$. Then, an environment $E = \langle c,e,p \rangle$ with $\msf{Sz}_\A(E) = 0$ is such that $c = e$ is even. That is, the corresponding parity game $\G_{(\formNF,E)}$ is, from Player $\A$'s point-of-view, a safety game. Hence, positional optimal strategies exists for Player $\A$ in $\G_{(\formNF,E)}$ \cite{AM04b}. This is similar for Player $\B$. Furthermore, there is no relevant environment of size at most -1. Hence, the equivalence holds for $n = 0$. 
	
	Assume now that $n \geq 1$. Consider any relevant environment $E$ with $\msf{Sz}(E) \leq n-1$. Then, we have $\msf{Sz}_\A(E) \leq n$. Therefore, if $\formNF$ is positionally maximizable w.r.t. Player $\A$ up to $n$, then 
	there is an optimal $\GF$-strategy in $\formNF$ for Player $\A$ w.r.t. $(\formNF,E)$. 
	
	Consider now an (arbitrary) environment $E = \langle c,e,p \rangle \in \msf{Env}(\outComeNF)$ with $\msf{Sz}_\A(E) = n \geq 1$. Let $\tilde{e} := \msf{Even}(e)$. Clearly, the outcomes in $\outComeNF$ leading, w.r.t. $p$, to $\qinit$ of color $c$, can be redirected to $k_c$ of color $c$ (which then loops back to $\qinit$) without changing the outcome of the parity game induced by $E$. Hence, without loss of generality, we assume that no outcome $o \in \outComeNF$ is such that $p(o) = \qinit$. We can therefore consider the relevant environment $E' \in \msf{Env}(\outComeNF)$ from Definition~\ref{def:relevant_env_from_arbitrary_env} of size $\msf{Sz}(E') = n-1$. Let $Y := (\formNF,E)$ and $Y' := (\formNF,E')$ and assume that there is a Player-$\A$ $\GF$-strategy $\sigma_\A \in \Sigma_\A(\formNF)$ that is optimal w.r.t. $Y'$. Let us show that $\sigma_\A$ is also optimal w.r.t. to $Y$. We let $q_Y$ denote the state $q_\msf{init}$ in the game $\G_Y$ and $q_{Y'}$ denote the state $q_\msf{init}$ in the game $\G_{Y'}$. Let $u := \MarVal{\G_Y}(q_Y)$ and $u' := \MarVal{\G_{Y'}}(q_{Y'})$. By Lemma~\ref{prop:relevant_env_no_decrease_value}, we have $u \leq u'$. We want to apply Lemma~\ref{lem:opt_in_gf_equiv_opt_in_parity}. 
	
	Let us show that the Player-$\A$ positional strategy $\s^Y_\A(\sigma_\A)$ ensures item $(i)$ of this lemma, i.e. that it dominates the valuation $v_Y^u$ in the game $\G_Y$. This amount to show that $\va_{\Games{\formNF}{v_Y^u \circ p}}(\sigma_\A) \geq u$. Note that, again by Lemma~\ref{lem:opt_in_gf_equiv_opt_in_parity}, we have that $\va_{\Games{\formNF}{v_{Y'}^{u'} \circ p'}}(\sigma_\A) \geq u'$. Furthermore, we have: 
	\begin{equation*}
	v_Y^u \circ p + (u' - u) \geq  v_{Y'}^{u'} \circ p'
	\end{equation*}
	Therefore, for any Player-$\B$ $\GF$-strategy $\sigma_\B \in \Dist(\Act_\B)$, we have:
	\begin{align*}
	\outM_{\Games{\formNF}{v_{Y}^{u} \circ p}}(\sigma_\A,\sigma_\B) + u' - u & \geq \outM_{\Games{\formNF}{v_{Y}^{u'} \circ p'}}(\sigma_\A,\sigma_\B) \geq \va_{\Games{\formNF}{v_{Y}^{u'} \circ p' }}(\sigma_\A) \geq u'
	\end{align*}
	Since this holds for all Player-$\B$ $\GF$-strategies $\sigma_\B \in \Sigma_\B(\formNF)$, it follows that 
	\begin{equation*}
	\va_{\Games{\formNF}{v_{Y}^{u}}}(\sigma_\A) \geq u
	\end{equation*}
	
	Consider now item $(ii)$ from Lemma~\ref{lem:opt_in_gf_equiv_opt_in_parity}. Consider some Player-$\B$ action $b \in \Act_\B$ and assume that $\outM_{\Games{\formNF}{\mathbbm{1}_{p^{-1}[0,1]}}}(\sigma_\A,b) = 0$. If we have $\outM_{\Games{\formNF}{\mathbbm{1}_{p^{-1}[k_{\tilde{e}}]}}}(\sigma_\A,b) > 0$, then it follows that $\max (\mathsf{Color}(\formNF,p,\sigma_\A,b) \cup \{ c \})$ is even since $\tilde{e}$ is even and is the highest integer appearing the game $\G_Y$. Assume now that $\outM_{\Games{\formNF}{\mathbbm{1}_{p^{-1}[k_{\tilde{e}}]}}}(\sigma_\A,b) = 0$. Then, it follows that, by definition of $p'$, we have $\outM_{\Games{\formNF}{\mathbbm{1}_{(p')^{-1}[0,1]}}}(\sigma_\A,b) = 0$. Therefore, by Lemma~\ref{lem:opt_in_gf_equiv_opt_in_parity}, it follows that $\max (\mathsf{Color}(\formNF,p',\sigma_\A,b) \cup \{ c' \})$ is even. Furthermore, by definition of $p'$
	:
	\begin{align*}
	f_E[& \{ i \in \brck{c,\tilde{e}-1} \mid \outM_{\Games{\formNF}{\mathbbm{1}_{p^{-1}[k_i]}}}(\sigma_\A,b) > 0 \} \cup \{ c \}] \\
	= & \{ i \in \brck{c',e'} \mid \outM_{\Games{\formNF}{\mathbbm{1}_{(p')^{-1}[k_i]}}}(\sigma_\A,b) > 0 \} \cup \{ c' \}
	\end{align*}
	
	We have $\max \{ i \in \brck{c',e'} \mid \outM_{\Games{\formNF}{\mathbbm{1}_{(p')^{-1}[k_i]}}}(\sigma_\A,b) > 0 \} \cup \{ c' \} = \max (\mathsf{Color}(\formNF,p',\sigma_\A,b) \cup \{ c' \})$ is even. Therefore, by Lemma~\ref{lem:same_parity}, we have $\max \{ i \in \brck{c,\tilde{e}-1} \mid \outM_{\Games{\formNF}{\mathbbm{1}_{p^{-1}[k_i]}}}(\sigma_\A,b) > 0 \} \cup \{ c \} = \max (\mathsf{Color}(\formNF,p,\sigma_\A,b) \cup \{ c \})$ is also even.
	
	
	
	In fact, the Player-$\A$ $\GF$-strategy $\sigma_\A \in \Sigma_\A(\formNF)$ satisfies item $(ii)$ of Lemma~\ref{lem:opt_in_gf_equiv_opt_in_parity}. Therefore, it is optimal w.r.t. $Y$. 
\end{proof}

\subsection{Proof of Lemma~\ref{lem:gf_distinct}, Page~\pageref{lem:gf_distinct}}
\label{proof:lem_gf_distinct}
\begin{proof}
	Consider a relevant environment $E = \langle c,e,p \rangle \in \msf{Env}(\outComeNF_n)$ and $Y_n := (\formNF_n,E)$. We let $u := \MarVal{\G_{Y_n}}(\qinit)$. Note that since $E$ is relevant, we have $c \in \{ 0,1 \}$ and $p: \outComeNF_n \rightarrow [0,1] \cup K_n$. For all $o \in \outComeNF$, if $p(o) \notin [0,1]$, we let $n_o \in \brck{0,e}$ denote the integer that $t$ is mapped to w.r.t. $p$, i.e. the integer ensuring $p(o) = k_{n_o}$. Note that, since $E$ is relevant, for all $o \in \outComeNF_n$, we have $n_o \geq c$.
	
	Let us introduce some notations. We let $X_{[0,1]} := \{ o \in \outComeNF_n \mid p(o) \in [0,1] \}$, $X_{\msf{even}} := \{ o \in \outComeNF_n \mid n_o \text{ is even} \}$ and $X_{\msf{odd}} := p^{-1}[K_n] \setminus X_{\msf{even}}$. 
	We let $X_\msf{Lp} := p^{-1}[K_n] = X_{\msf{even}} \uplus X_{\msf{odd}}$. Note that we have $\outComeNF_n = X_{[0,1]} \uplus X_\msf{Lp}$. We define the valuation $v: \outComeNF_n \rightarrow [0,1]$ mapping each outcome to its value in the game $\G_{Y_n}$ such that, for all $o \in \outComeNF_n$:
	\begin{equation*}
	v(o) := \begin{cases}
	p(o) & \text{ if }o \in X_{[0,1]} \\
	u & \text{ otherwise }
	\end{cases}
	\end{equation*}
	That is, $v = v_{y_n}^u \circ p$.
	%
	For all $x \in \outComeNF_n$, we let:
	\begin{equation*}
	w(x) := \begin{cases}
	1 & \text{ if }x \in X_{\msf{even}} \\
	0 & \text{ if }x \in X_{\msf{odd}} \\
	p(x) & \text{ if }x \in X_{[0,1]}  
	\end{cases}
	\end{equation*}
	For all $x \in \outComeNF_n$, the value $w(x) \in [0,1]$ corresponds to the value of the game $\G_{Y_n}$ if $x$ is seen (indefinitely, if it is in $X_\msf{Lp}$). This can be generalized to a pair of outcomes: for all $x,x' \in \outComeNF_n$, we let:
	\begin{equation*}
	w(x,x') := \begin{cases}
	w(x) & \text{ if }x,x' \in X_\msf{Lp} \text{ and }n_x \leq n_{x'} \\
	w(x') & \text{ if }x,x' \in X_\msf{Lp} \text{ and }n_{x'} \leq n_x \\
	p(x) & \text{ if }x \in X_\msf{[0,1]}, x' \in X_\msf{Lp} \\
	p(x') & \text{ if }x' \in X_\msf{[0,1]}, x \in X_\msf{Lp} \\
	\frac{1}{2} \cdot (p(x) + p(x')) & \text{ if } x,x' \in X_\msf{[0,1]}
	\end{cases}
	\end{equation*}
	
	\begin{figure}[!t]
		\centering
		\includegraphics[scale=1]{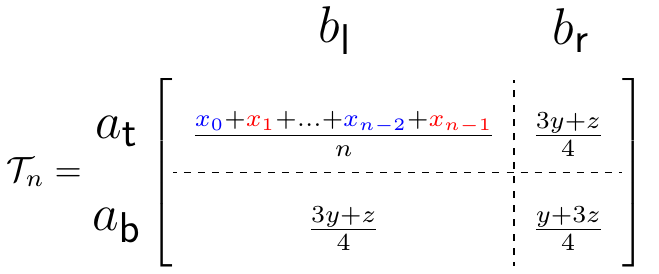}
		\caption{The game form $\mathcal{T}_n$.}
		\label{fig:GameFormExtracted}  
	\end{figure}
	
	Finally, we let $\mathcal{T}_n$ denote the $2 \times 2$ game form at the top left of the game form $\formNF_n$ from Figure~\ref{fig:GameFormDetail}. It is depicted in Figure~\ref{fig:GameFormExtracted}. We let $u' := \va_{\Games{\mathcal{T}_n}{v}}$. Now, there are several cases:
	\begin{itemize}
		\item Assume that $y,z \in X_\msf{Lp}$. Then, $u \in \{ 0,1 \}$, and playing $a_\msf{Ex} \in \Act_\A^n$ (resp. $b_\msf{Ex} \in \Act_\B^n$) is optimal for Player $\A$ (resp. $\B$) w.r.t. $Y_n$. 
		\item Assume now that $y \in X_\msf{Lp}$ and $z \in X_{[0,1]}$. Then, $u = p(z)$, as argued in the proof sketch, and playing $a_\msf{Ex} \in \Act_\A^n$ (resp. $b_\msf{Ex} \in \Act_\B^n$) is optimal for Player $\A$ (resp. $\B$) w.r.t. $Y_n$. 
		This is similar if $y \in X_{[0,1]}$ and $z \in X_{\msf{Lp}}$.
		\item Assume that $p(y),p(z) \in [0,1]$ and $p(y) \leq p(z)$. In that case, we have $u = \frac{3 p(y) + p(z)}{4}$ and playing action $a_\msf{Ex}$ (resp. $b_\msf{Ex}$) is optimal for Player $\A$ (resp. $\B$) w.r.t. $Y_n$.
		\item Let us now assume that $y,z \in X_{[0,1]}$ and $p(z) < p(y)$. In that case, we have:
		\begin{equation*}
		\frac{3 p(z) + p(y)}{4} \leq u \leq \frac{p(z) + 3 p(y)}{4}
		\end{equation*}
		We argue that Player $\B$ always has an optimal $\GF$-strategy w.r.t. $Y_n$.
		\begin{itemize}
			\item If $u = \frac{p(z) + 3 p(y)}{4}$, then positionally playing action $b_\msf{Ex}$ is optimal for Player $\B$. We now assume that $u < \frac{p(z) + 3 p(y)}{4}$.
			\item If, for some even $i \in \brck{0,n-1}$, we have $w(x_{i}) \leq u$, then playing action $b_i$ is optimal. We now assume that for all even $i \in \brck{0,n-1}$, we have $w(x_{i}) > u$ (i.e. $x_i$ is mapped to either a real greater then $u$ or even integer). Furthermore, if for any even $i \in \brck{2,n-1}$, we have $w(x_{i},x_{i-1}) \leq u$, then playing action $b_{i,i-1}$ is optimal. We now assume that for all even $i \in \brck{2,n-1}$, we have $w(x_{i},x_{i-1}) > u$. 
			\item It follows that, for any odd $i \in \brck{0,n-1}$, $w(x_{i}) \leq u$ and $w(x_{i},x_{i-1}) \leq u$ --- otherwise Player $\A$ could ensure that the value of the game is more than $u$ by playing the corresponding row, i.e. action $a_i$ or $a_{i,i-1}$. It also implies that $u' \leq u$. Indeed, assume that it is not the case, i.e. $u' > u$. Consider a Player-$\A$ $\GF$-strategy $\sigma_\A \in \Dist(\{ a_\msf{t},a_\msf{b} \})$ whose value in the game in normal form $\Games{\mathcal{T}_n}{v}$ is greater than $u$. Then, the Player-$\A$ positional strategy $\s_\A$ in the game $\G_{Y_n}$ that plays the $\GF$-strategy $\sigma_\A$ in $\qinit$ 
			parity dominates the valuation $v_{(\formNF_n,E)}^{r}$ for some $r > u$\footnote{Potentially, if it is not already the case, up to playing action $a_\msf{b}$ with small yet positive probability.}, due to the assumptions we made in the previous item. In fact, $u' \leq u$. 
			
			Note that this implies $X_{\msf{Ex}} \neq \emptyset$. Indeed, assume that $X_{\msf{Ex}} = \emptyset$. Let $\varepsilon > 0$. Consider a Player-$\A$ $\GF$-strategy $\sigma_\A$ in the game form $\mathcal{T}_n$ that plays action $a_\msf{t}$ with probability $1 - \varepsilon$ and action $a_\msf{b}$ with probability $\varepsilon$. Then, for some small enough, yet positive, $\varepsilon > 0$ and since $\frac{p(z) + 3 p(y)}{4} > u$ and $v(x_0),\ldots,v(x_{n-1}) = u$, such a $\GF$-strategy has value more than $u$ in the game in normal form $\Games{\mathcal{T}_n}{v}$. Hence, $X_{\msf{Ex}} = \emptyset$.
			\item Let us exhibit a Player-$\B$ $\GF$-strategy that is optimal w.r.t. $Y_n$. Consider a Player-$\B$ $\GF$-strategy $\sigma_\B \in \Dist(\{b_\msf{l},b_\msf{r}\})$ that is optimal in that game in normal form $\Games{\mathcal{T}_n}{v}$. Let $\s_\B$ be the Player-$\B$ positional strategy in the game $\G_{Y_n}$ that plays the $\GF$-strategy $\sigma_\B$ at $\qinit$. Recall that we assume that for any odd $i \in \brck{0,n-1}$, $w(x_{i}) \leq u$ and $w(x_{i},x_{i-1}) \leq u$, $\frac{3 p(z) + p(y)}{4} \leq u$ and $X_{\msf{Ex}} \neq \emptyset$. Therefore, the strategy $\s_\B$ dominates the valuations $v_{Y_n}^u$ and no Player-$\A$ action can, against the strategy $\s_\B$, make the game loop indefinitely on $\qinit$ while ensuring that the highest color seen is even, almost-surely. In fact, this strategy $\s_\B$ parity dominates the valuations $v_{Y_n}^u$. 
			Hence, the Player-$\B$ $\GF$-strategy $\sigma_\B$ is optimal w.r.t. $(Y_n)$, recall Lemma~\ref{lem:opt_in_gf_equiv_opt_in_parity}.
		\end{itemize}
		Let us now consider the case of Player $\A$. 
		\begin{itemize}
			\item If 
			$u = \frac{3 p(z) + p(y)}{4}$, then playing action $a_\msf{Ex}$ is optimal. We now assume that $\frac{3 p(z) + p(y)}{4} < u$.
			\item It follows that, for all even $i \in \brck{0,n-1}$, we have $w(x_{i}) \geq u$ and for all even $i \in \brck{2,k-1}$, we have $w(x_{i, i-1}) \geq u$. Otherwise Player $\B$ could play the action $b_i$ or $b_{i,i-1}$ and ensure that the value of the game $\G_{Y_n}$, from $\qinit$, is less than $u$.
			\item Furthermore, for any odd $i \in \brck{1,n-1}$, if either $w(x_{i}) \geq u$ or $w(x_{i},x_{i-1}) \geq u$, then playing action $a_i$ or action $a_{i,i-1}$ is optimal. We now assume that, for all odd $i \in \brck{A,k-1}$, $w(x_{i}) < u$ and $w(x_{i},x_{i-1}) < u$. 
			\item Assume that $X_\msf{Ex} \neq \emptyset$. Then, we have $u' \geq u$. This is analogous to the previous case for Player $\B$. Indeed, assume that $u' < u$. Then, consider a Player-$\B$ $\GF$-strategy $\sigma_\B \in \Dist(\{b_\msf{l},b_\msf{r}\})$ whose value in the game in normal form $\Games{\mathcal{T}_n}{v}$ is less than $u$. Then, a Player-$\B$ positional strategy $\s_\B$ in the game $\G_{Y_n}$ that plays the $\GF$-strategy $\sigma_\B$ in $\qinit$	parity dominates the valuation $v_{Y_n}^{r}$ for some $r < u$, due to the assumptions we made in the previous item and since we assume that $\frac{3 p(z) + p(y)}{4} < u$. In fact, $u' \geq u$. 
			
			Therefore, in the case where $X_\msf{Ex} \neq \emptyset$, Player $\A$ has an optimal $\GF$-strategy w.r.t. $Y_n$ that consists in playing optimally (in $\Dist(a_\msf{t},a_\msf{b})$) in the game in normal form $\va{\Games{\mathcal{T}_n}{v}}$, due to the assumptions of the previous item.
			\item Let us now assume that $X_\msf{Ex} = \emptyset$. 
			Since, for all even $i \in \brck{0,n-1}$, we have $w(x_{i}) \geq u$, it follows that $n_{x_i}$ is even and at least equal to $c$. Similarly, since for all odd $i \in \brck{1,n-1}$, we have $w(x_{i}) < u$. That is, $m_{x_i}$ is odd or less than $c$. In addition, for all even $i \in \brck{2,n-1}$, we have $w(x_{i,i-1}) \geq u$. That is, $n_{x_i} > n_{x_{i-1}}$. Similarly, for all odd $i \in \brck{1,n-1}$, we have $w(x_{i},x_{i-1}) < u$, it follows that $n_{x_i} > n(x_{i-1})$. We have $c \leq n_{x_0} < n_{x_1} < \ldots < n_{x_{n-2}} < n_{x_{n-1}} \leq e$
			. This cannot happen if $e - c \leq n-2$. 
		\end{itemize}
	\end{itemize}
	Let us now exhibit a relevant environment $E$ of size $n-1$ w.r.t. which Player $\A$ has no optimal strategy. Consider the environment where $c := 0$, $e := n-1$, $p(z) := 0$, $p(y) := 1$ and for all $i \in \brck{0,n-1}$, we have $p(x_i) := k_i$. We do have $\msf{Sz}(E) = n-1$. Furthermore, the value $u$ of the parity game $\G_{Y_n}$ from $\qinit$ is equal to $\frac{p(z) + 3 p(y)}{4} = \frac{3}{4}$. Indeed, Player $\B$ ensures that $u \leq \frac{3}{4}$ by playing action $b_\msf{Ex}$. Furthermore, for all $\varepsilon > 0$, Player $\A$ can play with probability $1-\varepsilon$ action $a_\msf{t}$ and with probability $\varepsilon$ action $a_\msf{b}$, thus ensuring that $u \geq \frac{3}{4} - \varepsilon$.
	
	Consider now an arbitrary Player-$\A$ $\GF$-strategy $\sigma_\A$ and the corresponding positional strategy $\s_\A$ in the game $\G_{Y_n}$. If $\sigma_\A$ plays with positive probability action $a_\msf{b}$ or action $a_\msf{Ex}$, then the value of $\s_\A$ is less that $u$ since $\frac{p(y) + 3 p(z)}{4} < \frac{p(z) + 3 p(y)}{4}$ (Player $\B$ can play the action $b_\msf{r}$). However, if $\sigma_\A$ does not play these actions with positive probability, then consider the Player $\B$ strategy $\s_\B$ that plays action $b_\msf{l}$ with probability 1. It ensures that the value of the strategy $\s_\A$ is 0. The reason why is the following: the highest index that the variables $x_0,x_1,\ldots,x_{n-1}$ are mapped to w.r.t. $p$ is $n-1$ and it is odd. Furthermore, all variables $x_1,x_3,\ldots,x_{n-1}$ are mapped to odd indices. Finally, the highest index that the variables $(x_1,x_0)$ are mapped to is $1$ and it is odd. This is also the case for $(x_3,x_2)$, ...,$(x_{n-1},x_{n-2})$. Hence, the highest color seen infinitely often with $\s_\A$ and $\s_\B$ is almost-surely odd, and surely no stopping state is reached.
\end{proof}

\end{document}